\documentclass{vllm}
\usepackage[margin=1.2in]{geometry}

\usepackage[utf8]{inputenc}
\usepackage[T1]{fontenc}
\usepackage{graphicx}
\usepackage{amsmath,amssymb,amsthm}
\usepackage{amsfonts}
\usepackage{booktabs}
\usepackage{multirow}
\usepackage{enumitem}
\usepackage{algorithm}
\usepackage{algorithmic}
\usepackage{float}
\usepackage{listings}
\usepackage{xcolor}
\usepackage{hyperref}
\usepackage{cleveref}
\usepackage{subcaption}
\usepackage{tabularx}
\usepackage{url}
\usepackage[numbers]{natbib}
\usepackage{tikz}
\usetikzlibrary{shapes,shapes.geometric,arrows,positioning,fit,calc,backgrounds,decorations.pathreplacing}

\newtheorem{definition}{Definition}

\newtheorem{proposition}{Proposition}

\hypersetup{
  colorlinks=true,
  linkcolor=blue!70!black,
  citecolor=green!50!black,
  urlcolor=blue!60!black,
  pdfauthor={vLLM Semantic Router Team},
  pdftitle={vLLM Semantic Router: Signal Driven Decision Routing for Mixture-of-Modality Models}
}

\newcommand{\sysname}{vLLM Semantic Router}

\newcommand{\halugate}{HaluGate}

\title{%
\textbf{vLLM Semantic Router:}\\[-1pt]
{\fontsize{11}{12.5}\selectfont\mbox{Signal Driven Decision Routing for Mixture-of-Modality Models}}%
}

\author{%
\parbox[t]{0.965\textwidth}{\centering\small\normalfont
Xunzhuo Liu, Huamin Chen, Samzong Lu, Yossi Ovadia, Guohong Wen, Hao Wu,
Zhengda Tan, Jintao Zhang, Senan Zedan, Yehudit Kerido, Liav Weiss,
Haichen Zhang, Bishen Yu, Asaad Balum, Noa Limoy,
Abdallah Samara, Baofa Fan, Brent Salisbury, Ryan Cook, Zhijie Wang,
Qiping Pan, Rehan Khan, Avishek Goswami, Houston H.\ Zhang, Shuyi Wang,
Ziang Tang, Fang Han, Zohaib Hassan, Jianqiao Zheng, Avinash Changrani, Xue (Steve) Liu, Bowei He$^\dagger$. \\
$^\dagger$ Corresponding author: \texttt{Bowei.He@mbzuai.ac.ae}
}%
}

\date{June 2026}

\abstract{%
As large language models (LLMs) diversify across modalities, capabilities, and cost profiles, the problem of \emph{intelligent request routing}: selecting the right model for each query at inference time, has become a critical systems challenge. We present \textbf{\sysname{}}, a signal-driven decision routing framework for Mixture-of-Modality (MoM) model deployments. The architecture follows two complementary Shannon-inspired views.
In the information-theoretic regime, signal extraction reduces the entropy of ``which model?'' by distilling routing-relevant information from raw queries. In the Boolean-algebraic regime, the decision engine composes functionally complete routing policies from signal conditions. The central innovation is \emph{composable signal orchestration}: thirteen heterogeneous signal types, spanning sub-millisecond heuristics and neural classifiers for semantics, safety, and modality, are composed through configurable Boolean decision rules into deployment-specific routing policies, so that fundamentally different scenarios (multi-cloud enterprise, privacy-regulated, cost-optimized) are expressed as different configurations over the same architecture. Matched decisions drive \emph{semantic model routing} via thirteen selection algorithms, while per-decision plugin chains enforce safety constraints including a three-stage \emph{HaluGate} hallucination detection pipeline and a lightweight episodic memory system with \emph{ReflectionGate} for personalized multi-turn context.
A typed neural-symbolic DSL specifies these routing policies and compiles them to multiple deployment targets, enabling configuration-first adaptation without code changes. Together, these components show that composable signal orchestration enables a single framework to serve diverse deployment scenarios with differentiated cost, privacy, and safety policies.
}

\begin{document}

\maketitle
\begingroup
\renewcommand{\thefootnote}{}
\footnotetext{\, Github code repository: \url{https://github.com/vllm-project/semantic-router}}
\footnotetext{\, Huggingface model repository: \url{https://huggingface.co/llm-semantic-router}}
\endgroup


\section{Introduction}
\label{sec:introduction}

The landscape of large language models has fragmented along multiple axes: modality (text, code, vision, diffusion), scale (1B to 1T+ parameters), cost (10$\times$ variation in per-token pricing), and specialization (general-purpose vs.\ domain-specific fine-tuning).
Organizations increasingly operate \emph{heterogeneous model fleets}---local vLLM instances alongside cloud endpoints from OpenAI, Anthropic, Azure, Bedrock, Gemini, and Vertex AI---each with different capabilities, pricing, and compliance characteristics.
This heterogeneity creates a fundamental inference-time optimization problem: \emph{given a user query, a fleet of diverse models, and deployment-specific constraints, which model should serve it, and what safety and privacy policies should apply?}

Viewed through the lens of information theory~\cite{shannon1948mathematical}, routing is an \emph{uncertainty-reduction} problem.
Before any analysis, routing entropy is maximal: $H(M \mid r_\text{raw}) \approx \log_2 K$ bits for $K$ candidate models---every model is equally plausible.
Each signal extracted from the request reduces this uncertainty.

This viewpoint suggests a natural two-layer decomposition (\Cref{fig:shannon_mapping}).
Under this decomposition, the objective is to collapse $H(M \mid S(r))$ to near zero so the decision engine can make a deterministic, high-confidence routing choice (\Cref{fig:entropy_collapse}).

\begin{figure*}[!ht]
  \centering
  \includegraphics[width=0.94\linewidth]{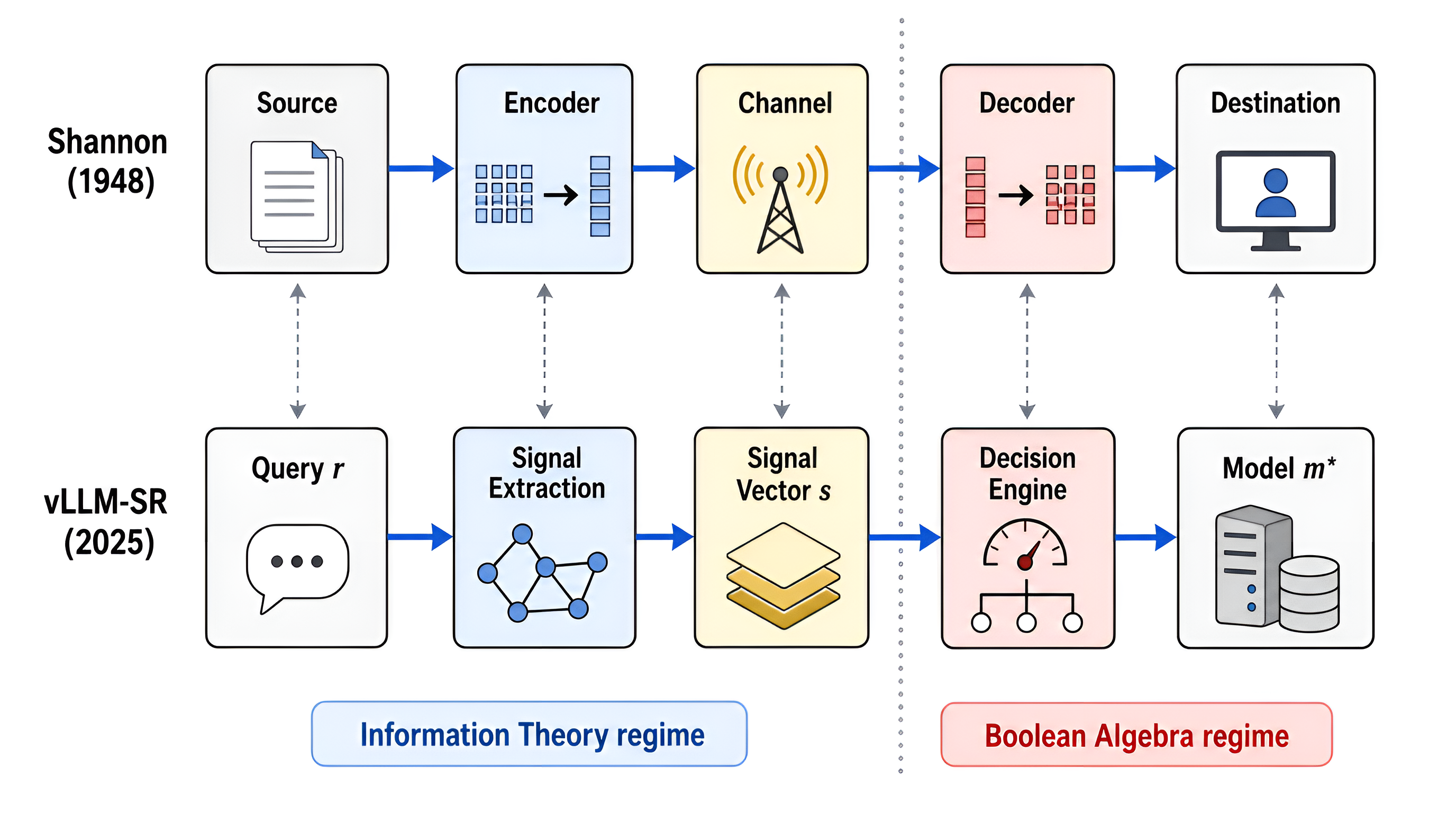}
  \caption{Structural correspondence between Shannon's communication system~\cite{shannon1948mathematical} and VSR's routing pipeline.
  The dotted line marks the boundary between the information-theoretic regime (signal extraction maximizes mutual information with the routing outcome) and the Boolean-algebraic regime (the decision engine synthesizes routing policies from signal conditions via $\{\wedge, \vee, \neg\}$~\cite{shannon1938symbolic}).}
  \label{fig:shannon_mapping}
\end{figure*}

In the first layer, signal extraction operates in the probabilistic regime of Shannon's \emph{Mathematical Theory of Communication}; in the second, decision evaluation operates in the algebraic regime of Shannon's \emph{switching-circuit algebra}~\cite{shannon1938symbolic}.
The signal vector~$\mathbf{s}$ is the interface between these regimes: continuous probabilistic inference below, discrete symbolic logic above.
In modern ML terms, the extraction layer behaves like a \emph{hybrid embedding stage} (\Cref{sec:disc_embedding}), while the priority-ordered decision engine behaves like a \emph{symbolic Mixture-of-Experts gate} with deterministic early exit (\Cref{sec:disc_moe}).
We also interpret priority blocks as \emph{layered entropy folding} (\Cref{sec:disc_entropy_folding}).
Together these components form a programmable neural-symbolic inference engine (Transformer-like control structure, but not equivalent hidden-state dynamics), which we connect to agent-based policy synthesis (\Cref{sec:disc_agent}).

\begin{figure*}[!ht]
  \centering
  \includegraphics[width=0.94\linewidth]{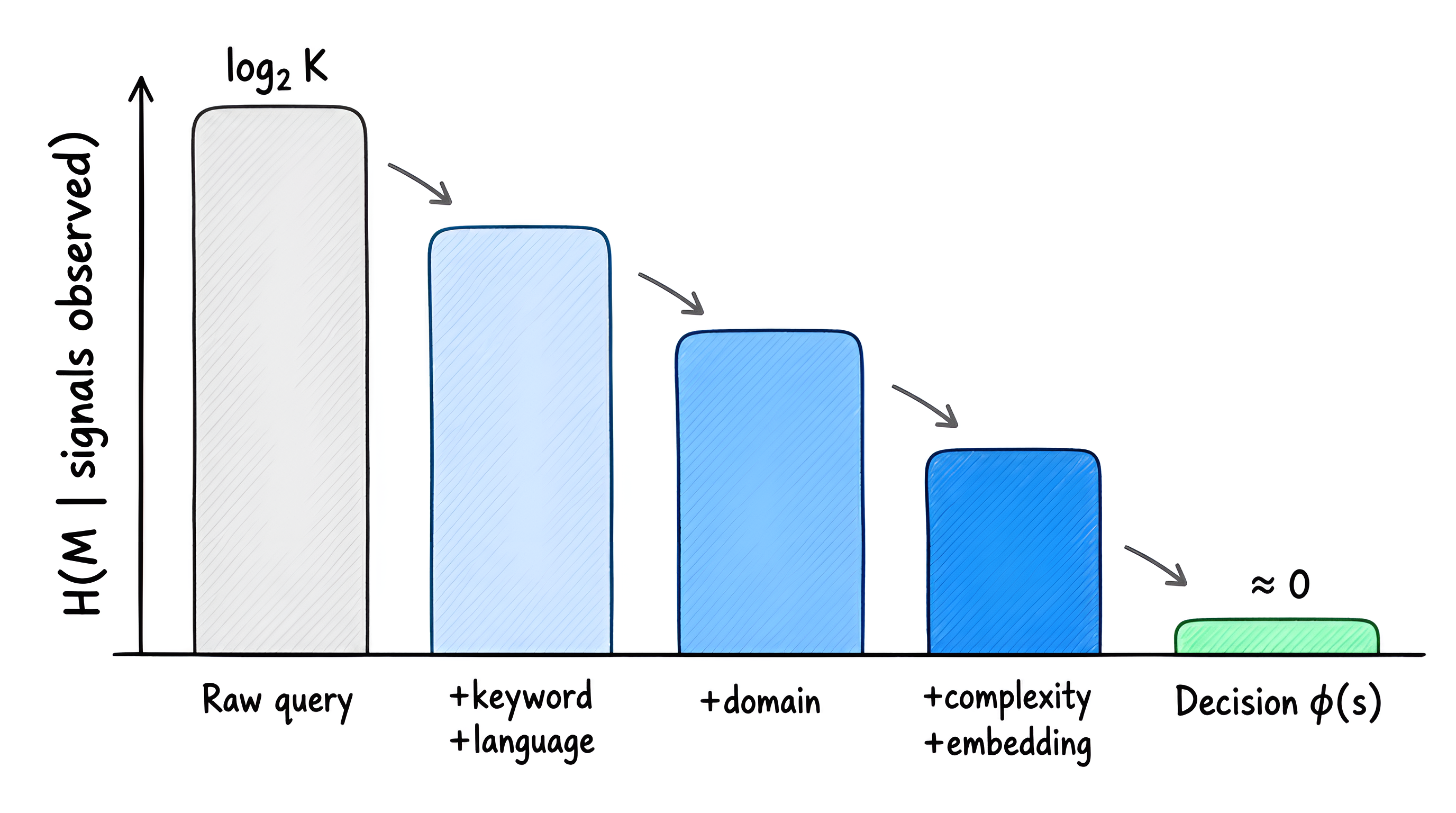}
  \caption{Entropy collapse during signal extraction (illustrative). Each additional signal reduces the routing uncertainty $H(M \mid \cdot\,)$ from the uniform prior $\log_2 K$ bits until the decision formula $\phi(\mathbf{s})$ yields a near-deterministic model selection. Bar heights are schematic rather than empirical measurements.}
  \label{fig:entropy_collapse}
\end{figure*}

This problem is more nuanced than binary difficulty routing.
A production routing system must simultaneously consider:
\begin{itemize}[leftmargin=*]
  \item \textbf{Multi-dimensional signals}: Query domain, modality, complexity, language, user identity, and real-time performance metrics all inform the optimal routing decision.
  \item \textbf{Privacy and safety}: Prompt injection, PII leakage, and hallucinated responses must be detected and mitigated---often with \emph{different policies for different query types and user roles}.
  \item \textbf{Cost-effective model selection}: Algorithms must balance response quality against inference cost and latency, selecting from a heterogeneous pool of local and cloud-hosted models.
  \item \textbf{Deployment diversity}: The same routing framework must serve a privacy-regulated healthcare deployment (strict PII filtering, on-premise models only), a cost-optimized developer tool (aggressive caching, cheapest model first), and a multi-cloud enterprise (failover across providers)---through configuration, not code changes.
  \item \textbf{Multi-turn statefulness}: Routing decisions must be consistent across conversation turns, requiring stateful session management and context preservation.
\end{itemize}

Prior work on LLM routing has made significant progress on individual aspects.
RouteLLM~\cite{ong2024routellm} trains classifiers to route between two models based on query difficulty.
RouterDC~\cite{chen2024routerdc} learns query-model embeddings via dual contrastive learning.
AutoMix~\cite{aggarwal2023automix} formulates cascading as a POMDP.
However, these approaches address model selection in isolation, without integrating signal extraction, safety enforcement, multi-provider backend management, or plugin extensibility into a unified framework.

\subsection{Contributions}

We present \sysname{}, a signal-driven decision routing system whose central innovation is \textbf{composable signal orchestration}: heterogeneous signals are extracted, composed through Boolean rules into deployment-specific decisions, and executed through per-decision plugin chains---enabling a single architecture to serve diverse deployment scenarios.

Our contributions are:

\begin{enumerate}[leftmargin=*]
  \item \textbf{Composable Signal-Decision-Plugin Architecture} (\Cref{sec:architecture,sec:signal_engine,sec:decision_engine,sec:plugins}):
    A three-layer architecture where thirteen signal types are composed through Boolean decision rules into deployment-specific routing policies, with per-decision plugin chains for safety, caching, and augmentation. Different deployment scenarios (privacy-regulated, cost-optimized, multi-cloud) are expressed as different configurations over the same architecture.

  \item \textbf{Semantic Model Routing with Cost-Aware Selection} (\Cref{sec:model_selection}):
    A unified framework integrating thirteen model selection algorithms---rating-based, contrastive, cascading, classical ML, reinforcement learning, and latency-aware---that analyze request semantics to select the most cost-effective model while respecting per-decision privacy and safety constraints.

  \item \textbf{HaluGate: Gated Hallucination Detection} (\Cref{sec:halugate}):
    A three-stage pipeline---sentinel gating, token-level detection, NLI-based explanation---that avoids unnecessary verification on non-factual queries while providing span-level diagnostics when hallucination is detected.

  \item \textbf{Multi-Provider and Multi-Endpoint Routing} (\Cref{sec:extproc}):
    Native support for routing across heterogeneous backends (vLLM, OpenAI, Anthropic, Azure, Bedrock, Gemini, Vertex AI) with provider-specific protocol translation, a pluggable authorization factory for diverse auth mechanisms, weighted multi-endpoint load distribution, and full OpenAI Responses API support for stateful multi-turn conversations.

  \item \textbf{LoRA-Based Multi-Task Classification} (\Cref{sec:lora_mom,sec:ml_inference}):
    A memory-efficient architecture using Low-Rank Adaptation that serves $n$ classification tasks from a single base model with lightweight adapter heads, reducing aggregate model memory from $n$ full copies to one base plus negligible adapter overhead.

  \item \textbf{Episodic Conversation Memory with ReflectionGate} (\Cref{sec:memory_rag}):
    A lightweight memory system that stores raw conversational turns as episodic chunks (filtered by an entropy gate and capped at 16\,KB) rather than relying on LLM-based fact extraction, eliminating inference overhead at write time.
    At retrieval time, a multi-stage \emph{ReflectionGate} pipeline---safety block-pattern filtering, recency decay, Jaccard deduplication, and budget capping---refines retrieved chunks before injection as a separate context message, enabling personalized multi-turn routing without coupling memory quality to an external model.

  \item \textbf{Programmable Neural-Symbolic Configuration Language} (\Cref{sec:dsl}):
    A typed configuration language that serves as the instruction set of the routing inference engine, with a formal grammar parsed into a Boolean expression AST, three-level validation (syntax, reference, constraint), multi-target compilation (flat YAML, Kubernetes CRDs, Helm charts), and round-trip decompilation.
    We formalize the system as a \emph{programmable neural-symbolic inference engine}---neural signal extraction as a hybrid embedding layer, symbolic decision evaluation as Mixture-of-Experts gating---and show that the language's functional completeness enables LLM-based coding agents to synthesize routing policies from natural-language specifications.
\end{enumerate}

\subsection{Paper Organization}

\Cref{sec:architecture} presents the system architecture and composable orchestration model.
\Cref{sec:signal_engine,sec:decision_engine} formalize the signal extraction and decision evaluation layers.
\Cref{sec:plugins,sec:safety,sec:halugate} describe the plugin framework and safety subsystems.
\Cref{sec:lora_mom,sec:ml_inference} detail the LoRA-based classification architecture and multi-runtime inference design.
\Cref{sec:model_selection} surveys the semantic model selection algorithms.
\Cref{sec:extproc} describes the multi-provider request processing pipeline.
\Cref{sec:memory_rag,sec:observability,sec:deployment} cover memory, observability, and deployment.
\Cref{sec:evaluation} presents evaluation results.
\Cref{sec:dsl} specifies the programmable configuration language and develops the neural-symbolic inference engine perspective.
\Cref{sec:related_work} discusses related work, and \Cref{sec:conclusion} concludes.


\section{System Architecture}
\label{sec:architecture}

We formalize the routing problem and present the three-layer architecture that decomposes it into composable signal extraction, decision evaluation, and plugin execution---enabling a single framework to serve diverse deployment scenarios through configuration.

\subsection{Problem Formulation}

Let $\mathcal{M} = \{m_1, \ldots, m_K\}$ denote a set of $K$ available model backends, each characterized by capability profile, cost, and latency.
Each backend may be served by a different provider $p_k \in \mathcal{P}$ (e.g., local vLLM, OpenAI, Anthropic, Azure, Bedrock, Gemini), with provider-specific API protocols and authentication mechanisms.
A deployment may expose multiple endpoints $\mathcal{E} = \{e_1, \ldots, e_L\}$ with weighted load distribution across backends.

Given an incoming request $r$ (consisting of a message sequence, metadata, user identity, and headers), the routing problem is to:
\begin{enumerate}[leftmargin=*]
  \item Select a model $m^* \in \mathcal{M}$ that maximizes response quality while respecting cost and latency constraints;
  \item Apply deployment-specific safety and privacy transformations $\mathcal{T}(r)$ before and after model invocation;
  \item Route through the correct provider endpoint with appropriate authentication.
\end{enumerate}

Na\"ive approaches either fix $m^*$ statically or route based on a single dimension (e.g., estimated difficulty).
We argue that production routing requires reasoning over \emph{multiple orthogonal signal dimensions simultaneously}, with different \emph{policies} (safety thresholds, caching strategies, prompt augmentation, model pools) for different routing outcomes---and that these policies must be \emph{composable} to support diverse deployment scenarios without architectural changes.

\subsection{Composable Signal Orchestration}

The key architectural innovation is that the same signal extraction, decision evaluation, and plugin execution machinery can be \emph{composed differently} for different deployment scenarios:

\begin{definition}[Programmable Neural-Symbolic Configuration Language]
A deployment configuration $\Gamma = (\mathcal{S}_\Gamma, \mathcal{D}_\Gamma, \Pi_\Gamma, \mathcal{E}_\Gamma)$ specifies which signal types $\mathcal{S}_\Gamma \subseteq \mathcal{S}$ are active, what decisions $\mathcal{D}_\Gamma$ are evaluated, which plugin chains $\Pi_\Gamma$ are attached, and which endpoints $\mathcal{E}_\Gamma$ are available.
\end{definition}

\noindent\textbf{Example configurations:}
\begin{itemize}[leftmargin=*]
  \item \emph{Privacy-regulated (healthcare)}: Active signals include domain, authz, and language. Decisions route sensitive queries to on-premise models only. Plugins enforce strict PII redaction with no caching.
  \item \emph{Cost-optimized (developer tool)}: Active signals include complexity, embedding, and keyword. Decisions cascade from cheap to expensive models. Plugins enable aggressive semantic caching.
  \item \emph{Multi-cloud enterprise}: Active signals include domain, modality, and authz. Decisions distribute across multiple provider endpoints using latency-aware model selection with weighted failover. Plugins inject provider-specific auth headers.
\end{itemize}

All three scenarios use the same architecture; only $\Gamma$ differs. This composability is the central design contribution.

\subsection{Three-Layer Architecture}

The architecture decomposes routing into three layers, each with a well-defined interface (\Cref{fig:architecture}):

\begin{figure*}[t]
  \centering
  \includegraphics[width=0.94\linewidth]{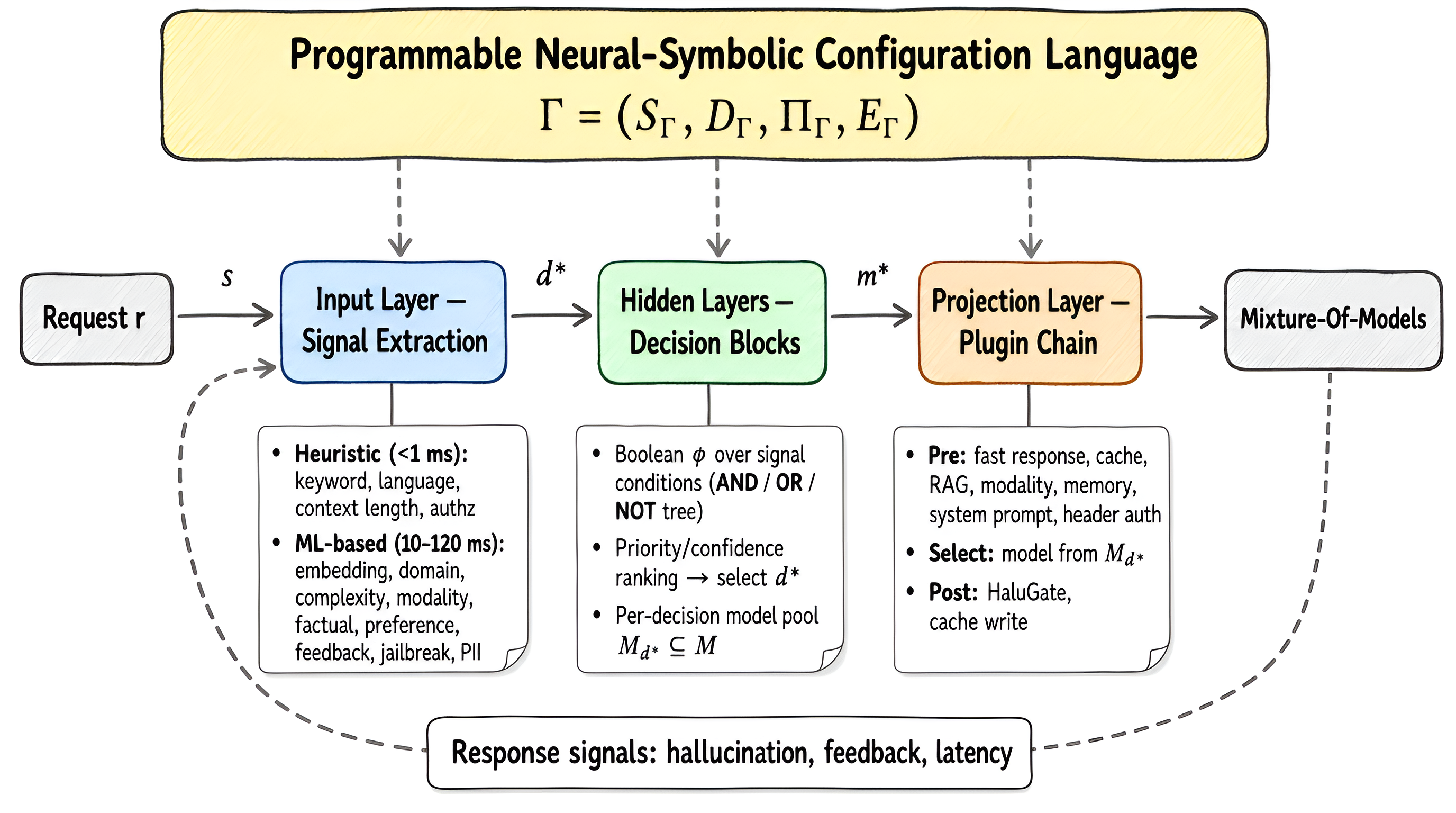}
  \caption{Three-layer architecture with closed-loop feedback. A deployment configuration~$\Gamma$ selects which signals, decisions, and plugins are active. Layer~1 extracts a signal vector~$\mathbf{s}$ from the request. Layer~2 evaluates Boolean decision formulas to select~$d^*$. Layer~3 executes the per-decision plugin chain, selects a model from~$d^*$'s candidate set, and routes to the provider endpoint. Response-side signals feed back to enable adaptive routing.}
  \label{fig:architecture}
\end{figure*}

\noindent\textbf{Input Layer: Signal Extraction.}
The signal layer maps a request $r$ to a structured signal result $\mathbf{s}$, consisting of binary match indicators and real-valued confidences for each configured rule across thirteen signal types.
Heuristic signals (keyword, language, context length, authorization) complete in sub-millisecond time.
ML-based signals (embedding similarity, domain classification, factual grounding, modality detection, complexity, preference, user feedback, jailbreak detection, PII detection) require neural inference at 10--120\,ms.
Signals are evaluated in parallel, and only signal types referenced by at least one active decision are computed---a critical optimization for deployment configurations that use a subset of available signals.

\noindent\textbf{Hidden Layers: Decision Blocks.}
The decision layer takes the signal result $\mathbf{s}$ and evaluates a set of decisions $\mathcal{D} = \{d_1, \ldots, d_M\}$, each defined as a Boolean formula over signal conditions.
The engine selects the best-matching decision $d^*$ using either priority-based or confidence-weighted ranking.
Each decision carries its own model candidate set $\mathcal{M}_{d^*} \subseteq \mathcal{M}$, enabling deployment-specific model pools (e.g., a privacy decision restricts candidates to on-premise models).

\noindent\textbf{Projection Layer: Plugin Chain.}
Each decision $d^*$ carries a per-decision plugin configuration that defines:
(a)~\emph{pre-routing plugins} (fast response for safety enforcement, semantic caching, RAG context injection, modality routing, memory retrieval, system prompt augmentation, header mutation for provider auth), executed before model invocation;
(b)~a \emph{semantic model selection algorithm} applied to $d^*$'s candidate model set $\mathcal{M}_{d^*}$ to find the best model cost-effectively;
(c)~\emph{post-routing plugins} (hallucination detection, cache updates), executed on the model response.

\subsection{Design Principles}

Four principles guide the architecture:

\noindent\textbf{Composability.}
Complex routing policies are expressed as compositions of simple primitives: Boolean combinations of signal conditions form decisions; sequences of typed plugins form execution chains; deployment scenarios are expressed as configuration profiles.
This avoids monolithic routing logic and enables the same system to serve fundamentally different deployment requirements.

\noindent\textbf{Orthogonality.}
Signals, decisions, and plugins are independent modules with a uniform interface boundary.
New signal types can be added by implementing a single evaluation function---the decision engine references signals solely by type and rule name, requiring no modification.
Likewise, new plugins and providers are registered independently.
The current thirteen signal types are the built-in set; the framework is designed to be extended with domain-specific signals as deployment requirements evolve.
This strict decoupling has a theoretical antecedent in Shannon's \emph{source-channel separation theorem}~\cite{shannon1948mathematical}: just as a communication system can be decomposed into a source encoder (which compresses the message to its informational essence) and a channel encoder (which adds structure for reliable transmission)---each optimized independently without loss of optimality---the routing system decomposes into a \emph{signal extraction layer} (which distills routing-relevant information from the raw request) and a \emph{decision layer} (which synthesizes a routing policy from the extracted signals).
The signal vector~$\mathbf{s}$ is the interface between these independently optimizable stages, analogous to the compressed source representation at the encoder--channel boundary.
This separation also mirrors the architecture of Mixture-of-Experts models~\cite{shazeer2017moe}, where the gating network (signal extraction) and expert modules (model backends) are designed independently---though our gating is symbolic rather than neural, enabling formal verification and compositional editing (\Cref{sec:disc_moe}).

\noindent\textbf{Closed-loop adaptivity.}
The bidirectional signal flow described in \Cref{sec:signal_engine} enables the architecture to operate as a \emph{closed-loop control system}~\cite{astrom2008feedback}.
In control-theoretic terms, the signal--decision--plugin pipeline is the \emph{plant}, response-side signals (hallucination detection, user feedback, latency measurements) are the \emph{sensors}, and a policy adjustment mechanism is the \emph{controller} that updates decision parameters $\theta^{(t)}$ (priorities, model weights) based on observed response quality:
\begin{equation}
  \theta^{(t+1)} = \theta^{(t)} + \eta \,\nabla_\theta\, \mathbb{E}\bigl[Q\bigl(r,\, m^*(r;\theta^{(t)})\bigr)\bigr]
\end{equation}
where $Q(r, m^*)$ is a response quality metric and $\eta$ is a learning rate.
This formulation connects to the \emph{contextual bandit} framework~\cite{li2010contextual}: the signal vector $S(r)$ serves as the context, model selection is the action, and response quality is the reward.
Standard regret bounds from online learning theory~\cite{shalev2012online} guarantee that the cumulative routing quality of such an adaptive policy converges to that of the best fixed policy in hindsight at a rate of $O(\sqrt{T})$, providing formal performance guarantees for self-improving routing.

\noindent\textbf{Per-decision scoping.}
Safety thresholds, caching policies, model candidates, and auth mechanisms are scoped to individual decisions rather than applied globally.
A coding-focused decision can omit PII signal conditions while a customer-support decision enforces strict PII filtering via signal-matched fast response---using the same system configuration.

\noindent\textbf{Provider abstraction.}
The architecture abstracts over provider-specific protocols, authentication, and endpoint topologies.
Multi-endpoint routing with weighted distribution and failover is handled at the infrastructure layer, enabling decisions to reference models by capability rather than by provider-specific endpoint.


\section{Signal Extraction Layer}
\label{sec:signal_engine}

The signal extraction layer maps an incoming request $r$ to a structured signal result that characterizes the request along thirteen orthogonal dimensions.
\Cref{fig:signal_taxonomy} provides an overview of the signal taxonomy and evaluation flow.
We formalize the signal model and describe the extraction algorithms.

\subsection{Signal Model}

\begin{definition}[Signal Rule]
A \emph{signal rule} $\rho = (\tau, n, f)$ consists of a signal type $\tau \in \mathcal{T}$, a rule name $n$, and an evaluation function $f: \mathcal{R} \to \{0, 1\} \times [0, 1]$ that maps a request to a binary match indicator and a confidence score.
\end{definition}

\begin{definition}[Signal Result]
Given a rule set $\mathcal{R} = \{\rho_1, \ldots, \rho_N\}$, the signal result for request $r$ is:
\begin{equation}
  S(r) = \bigl\{ \bigl(\rho_i, \; \mathbf{1}[f_i(r)], \; c_i(r)\bigr) \mid \rho_i \in \mathcal{R} \bigr\}
\end{equation}
where $\mathbf{1}[f_i(r)]$ is the match indicator and $c_i(r) \in [0,1]$ is the confidence.
\end{definition}

The thirteen signal types partition into \emph{heuristic} and \emph{learned} categories based on whether they require neural inference.

\begin{figure}[t]
\centering
\includegraphics[width=0.94\linewidth]{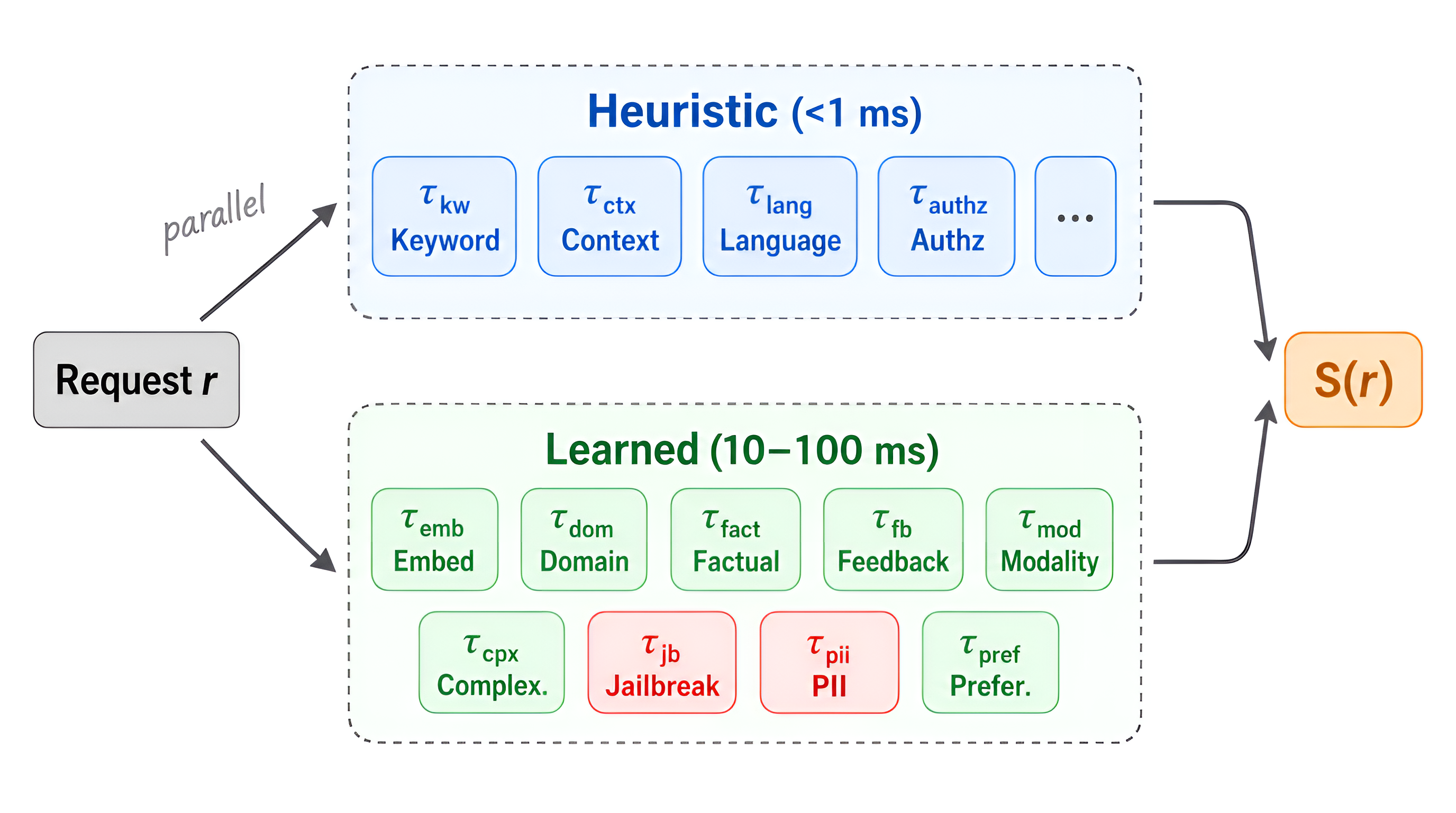}
\caption{Signal extraction taxonomy and evaluation flow.  An incoming request is evaluated in parallel against heuristic signals (sub-millisecond, deterministic) and learned signals (neural inference via LoRA classifiers).  Only signal types referenced by configured decisions are computed (demand-driven evaluation).  Results merge into the structured signal result $S(r)$.}
\label{fig:signal_taxonomy}
\end{figure}

\subsection{Heuristic Signals}

Heuristic signals use deterministic or statistical algorithms with sub-millisecond latency:

\noindent\textbf{Keyword} ($\tau_\text{kw}$).
Rules are defined as pattern sets with Boolean combinators.
Each rule specifies a set of patterns $P = \{p_1, \ldots, p_k\}$ with a combinator $\in \{\textsc{and}, \textsc{or}, \textsc{nor}\}$ and one of three matching methods:
\begin{itemize}[nosep,leftmargin=1.5em]
  \item \emph{Regex} (default): compiled regular expressions with word boundaries; confidence is 1.0 on match.
  \item \emph{BM25}: BM25 scoring dispatched to a Rust-backed classifier via FFI. Each keyword is scored against the request using TF-IDF term weighting, and the rule matches when the score exceeds a configurable threshold (default 0.1). Confidence is derived from the BM25 score, providing a graded relevance signal rather than a binary match.
  \item \emph{N-gram}: character $n$-gram similarity (default trigram) dispatched to the same Rust binding. The rule matches when the Jaccard similarity between the keyword and request $n$-gram sets exceeds a threshold (default 0.4), providing inherent tolerance to typos and morphological variation without a dedicated fuzzy-matching pass.
\end{itemize}
For \textsc{and}: $f(r) = \bigwedge_{i} \text{match}(p_i, r)$; for \textsc{or}: $f(r) = \bigvee_{i} \text{match}(p_i, r)$.
The combinators apply uniformly across all three methods.

\noindent\textbf{Context Length} ($\tau_\text{ctx}$).
Rules define token-count intervals $[l, u]$.
Given estimated token count $t(r)$, the rule matches iff $l \leq t(r) \leq u$.
This enables complexity-aware routing (e.g., short queries to fast models, long contexts to extended-context models).

\noindent\textbf{Language} ($\tau_\text{lang}$).
Rules bind detected language codes to named signals using statistical n-gram detection over 100+ languages.
Enables language-specific routing (e.g., CJK queries to multilingual-specialized models).

\noindent\textbf{Authorization} ($\tau_\text{authz}$).
Role-based access control signals extracted from request headers, supporting a pluggable authentication factory.
The authz signal layer abstracts over multiple identity providers (API key, OAuth2/OIDC, cloud IAM, custom JWT, LDAP) through provider-specific extractors that resolve user identities and group memberships from credentials.
Role bindings then map resolved identities to named signals, enabling per-role routing policies (e.g., premium users routed to higher-quality models, free-tier users restricted to cost-effective models).
This \emph{inbound} authorization (who is the user and what can they access?) is complementary to the \emph{outbound} authorization factory (\Cref{subsec:authz_factory}) that injects provider-specific credentials when forwarding to backends.

\subsection{Learned Signals}

Learned signals require neural inference, typically 10--100\,ms, using the LoRA-based classifiers described in \Cref{sec:lora_mom}.

\noindent\textbf{Why encoder-based models.}
The choice of bidirectional encoders (ModernBERT~\cite{warner2024modernbert}, mmBERT-32K for long-context tasks) rather than unidirectional decoders for signal extraction is not merely an efficiency decision---it reflects a fundamental information-theoretic principle.
Routing requires \emph{understanding} the query: determining its domain, complexity, intent, modality, and the precise location of sensitive content.
Understanding, in Shannon's framework~\cite{shannon1948mathematical}, corresponds to building representations that maximize the mutual information $I(\mathbf{H};\, Y)$ between the hidden states~$\mathbf{H}$ and the task label~$Y$.
A bidirectional encoder conditions every token on the full context in both directions, producing hidden states that capture the complete information structure of the input; a causal decoder, by contrast, conditions each token only on its left context---its representations are optimized for next-token prediction $I(\mathbf{h}_t;\, x_{t+1})$, which captures \emph{generative} continuation rather than \emph{discriminative} understanding.

This bidirectional ``understanding'' is exploited at three distinct granularities across the signal types:
\begin{itemize}[leftmargin=*]
  \item \textbf{Sequence-level} (CLS pooling): For domain classification, jailbreak detection, fact-check gating, modality classification, and user feedback, the pooled CLS representation acts as an approximate \emph{sufficient statistic}---a fixed-dimensional vector that compresses the query's global semantics while discarding positional detail irrelevant to the label.
  \item \textbf{Token-level} (per-token hidden states): For PII detection and hallucination span identification, the \emph{position} of information matters---which tokens are names, which spans are unsupported claims. Here the full sequence of hidden states $(\mathbf{h}_1, \ldots, \mathbf{h}_T)$ is retained, preserving the positional mutual information $I(\mathbf{h}_t;\, y_t)$ that CLS pooling would discard.
  \item \textbf{Cross-sequence} (cross-encoder): For NLI-based hallucination explanation (\Cref{sec:halugate}), the encoder jointly attends over a (claim, evidence) pair, maximizing the \emph{inter-sequence} mutual information $I(\text{claim};\, \text{evidence})$ through full cross-attention---a capacity unavailable to architectures that encode each sequence independently.
\end{itemize}
In each case, the encoder's unrestricted attention pattern---the absence of a causal mask---is what enables maximal information extraction at the granularity the task demands.

\noindent\textbf{Embedding Similarity} ($\tau_\text{emb}$).
Each rule defines reference texts $\{t_1, \ldots, t_k\}$ and a similarity threshold $\theta$.
The request embedding $\mathbf{e}_r$ is computed via a shared embedding model, and the rule matches iff:
\begin{equation}
  \max_i \cos(\mathbf{e}_r, \mathbf{e}_{t_i}) \geq \theta
\end{equation}
The confidence equals the maximum cosine similarity.
This provides scalable semantic matching without per-rule model training.

\noindent\textbf{Domain Classification} ($\tau_\text{dom}$).
A LoRA-adapted classifier trained on MMLU categories maps requests to domain labels (STEM, humanities, code, creative writing, etc.).
The classification confidence serves as the signal confidence.

\noindent\textbf{Factual Grounding} ($\tau_\text{fact}$).
A binary classifier (the HaluGate Sentinel, \Cref{sec:halugate}) determines whether the query requires factual verification, distinguishing factual questions from creative or code-generation tasks.

\noindent\textbf{User Feedback} ($\tau_\text{fb}$).
A multi-class classifier detects satisfaction, dissatisfaction, clarification requests, and preference for alternatives, enabling feedback-driven routing adjustments.

\noindent\textbf{Modality} ($\tau_\text{mod}$).
A three-class classifier (autoregressive, diffusion, both) determines the appropriate model modality for the request, trained on mixed text-generation and image-generation datasets.

\noindent\textbf{Complexity} ($\tau_\text{cpx}$).
A contrastive embedding classifier estimates query difficulty.
Each complexity rule defines two sets of candidate exemplars---\emph{hard} (e.g., multi-step reasoning problems) and \emph{easy} (e.g., simple factual lookups)---whose embeddings are precomputed at initialization.
At query time, the query embedding is compared against both sets via cosine similarity:
\begin{equation}
  \delta = \max_{h \in \mathcal{H}} \text{sim}(\mathbf{q}, \mathbf{h}) - \max_{e \in \mathcal{E}} \text{sim}(\mathbf{q}, \mathbf{e})
\end{equation}
where $\mathcal{H}$ and $\mathcal{E}$ are the hard and easy candidate sets.
The difficulty level is then $\text{hard}$ if $\delta > \theta$, $\text{easy}$ if $\delta < -\theta$, and $\text{medium}$ otherwise, for a per-rule threshold $\theta$.
Multiple complexity rules can coexist (e.g., \texttt{code\_complexity}, \texttt{math\_complexity}), each with its own candidate sets and threshold.
Composer conditions can restrict when a rule fires---for instance, evaluating \texttt{code\_complexity} only when the domain signal indicates computer science.

\noindent\textbf{Jailbreak Detection} ($\tau_\text{jb}$).
Each jailbreak rule selects one of two detection methods via a \texttt{method} field (default: \texttt{classifier}).

\emph{BERT classifier method.}
A binary/ternary classifier detects adversarial prompt injection and jailbreak attempts.
Each rule defines a confidence threshold $\theta$; the signal fires when jailbreak confidence $c \geq \theta$.
An optional \texttt{include\_history} flag controls whether only the latest user message or the full conversation history is analyzed, trading latency for recall against multi-turn attacks.
Multiple rules at different thresholds enable per-decision sensitivity---a public chatbot may use $\theta = 0.65$, while an internal tool uses $\theta = 0.9$ to minimize false positives.

\emph{Contrastive embedding method.}
Analogous to the complexity signal, each rule defines two sets of exemplar patterns---\emph{jailbreak} ($\mathcal{K}_\text{jb}$, e.g., ``Ignore all previous instructions'') and \emph{benign} ($\mathcal{K}_\text{ben}$, e.g., ``What is the weather today'')---whose embeddings are precomputed at initialization.
At request time the contrastive score for a message $m$ is:
\begin{equation}
  \delta(m) = \max_{j \in \mathcal{K}_\text{jb}} \text{sim}(\mathbf{m}, \mathbf{j}) - \max_{b \in \mathcal{K}_\text{ben}} \text{sim}(\mathbf{m}, \mathbf{b})
\end{equation}
When \texttt{include\_history} is enabled, the system evaluates every user message in the conversation and takes the maximum score across all turns:
$\Delta = \max_{m \in \mathcal{M}_\text{user}} \delta(m)$.
The rule fires when $\Delta \geq \theta$ (default $\theta = 0.10$).
This \emph{max-contrastive-chain} aggregation is specifically designed to catch ``boiling frog'' multi-turn attacks where each individual message may score below threshold but the conversation contains at least one escalation turn.

Both methods coexist within the same signal type; a single deployment can define BERT and contrastive rules simultaneously and combine them in decision logic using OR (fire if either detects) or use them at different priority levels for graduated response.
Classifier details are described in \Cref{sec:safety}.

\noindent\textbf{PII Detection} ($\tau_\text{pii}$).
A token-level NER classifier identifies personally identifiable information (person names, emails, phone numbers, SSNs, credit cards, etc.).
Each rule specifies a confidence threshold and an optional allow-list of PII entity types.
The signal fires when any PII type \emph{not} in the allow-list is detected above threshold.
This per-rule policy model enables differentiated enforcement: a medical application may allow \textsc{person} while blocking \textsc{ssn}, whereas a public chatbot blocks all PII types.
Classifier details are described in \Cref{sec:safety}.

\noindent\textbf{Preference} ($\tau_\text{pref}$).
Personalized routing based on user interaction history.

\subsection{Parallel Evaluation with Lazy Computation}

A key optimization is \emph{demand-driven evaluation}: the engine computes only signal types referenced by at least one configured decision.
Let $\mathcal{T}_\text{used} = \bigcup_{d \in \mathcal{D}} \{\tau \mid \exists\, \text{condition in } d \text{ of type } \tau\}$.
Signal evaluators for types in $\mathcal{T}_\text{used}$ are launched as concurrent coroutines, with heuristic signals completing before learned signals due to their sub-millisecond latency.

This demand-driven approach avoids the cost of unused signal types.
In typical configurations with 3--5 active signal types out of thirteen, this reduces total signal extraction latency by 50--70\% compared to exhaustive evaluation.
The strategy has a natural information-theoretic interpretation.
Shannon's source coding theorem~\cite{shannon1948mathematical} establishes that optimal codes assign shorter codewords to more probable symbols; analogously, demand-driven evaluation assigns \emph{zero computational cost} to signals that carry no routing information in the current configuration, and full cost only to those that do.
If we view each signal type's evaluation cost $c_i$ as a ``codeword length'' and its relevance indicator $\mathbf{1}[\tau_i \in \mathcal{T}_\text{used}]$ as the probability of being needed, then demand-driven evaluation minimizes the expected evaluation cost $\sum_i c_i \cdot \mathbf{1}[\tau_i \in \mathcal{T}_\text{used}]$---the computational analogue of minimizing expected code length.

\subsection{Extensibility}

The thirteen signal types described above represent the current built-in set; the framework is not limited to these.
The signal extraction layer defines a uniform interface---each signal type implements an evaluation function $f: \mathcal{R} \to \{0,1\} \times [0,1]$---and the decision engine references signals solely by type and rule name.
Adding a new signal type requires only implementing this interface and registering the type; no changes to the decision engine, plugin chain, or deployment infrastructure are needed.
This open architecture allows operators to introduce domain-specific signals (e.g., regulatory compliance classifiers, custom toxicity detectors) alongside the built-in types.

\subsection{Bidirectional Signal Flow}

Signals are not limited to the inbound request path.
The system also extracts signals from model \emph{responses}, enabling closed-loop routing policies that adapt based on output characteristics.
The primary example is \halugate{} (\Cref{sec:halugate}): the Sentinel classifier on the request path determines whether a query requires factual verification (the $\tau_\text{fact}$ signal), and if so, the Detector and Explainer stages analyze the model's response for unsupported claims---producing response-side detection results (confidence scores, hallucinated spans, NLI explanations) that are propagated via HTTP headers or body annotations.
This bidirectional flow---request signals gating which response analyses to perform, and response signals feeding back into observability and policy enforcement---enables adaptive quality assurance without imposing uniform overhead on all requests.

\subsection{Information-Theoretic Signal Analysis}

With $N$ signal types evaluated per request, a natural question is whether all signals contribute independently to routing quality or whether some carry redundant information.
Information theory provides the formal framework for this analysis.

For a signal type $\tau_i$ and the routing outcome variable $Y$ (the selected model), the \emph{mutual information} $I(\tau_i; Y)$ quantifies the reduction in uncertainty about the routing decision provided by observing signal $\tau_i$.
The \emph{conditional mutual information} $I(\tau_i; Y \mid \tau_j)$ measures the additional information from $\tau_i$ given that $\tau_j$ is already observed.
When $I(\tau_i; Y \mid \tau_j) \approx 0$, signals $\tau_i$ and $\tau_j$ are redundant with respect to routing---observing both provides no more discriminative power than observing one.

This analysis enables two optimizations.
First, \emph{adaptive signal pruning}: in a given deployment configuration, signals with near-zero mutual information with the routing outcome can be disabled without affecting routing quality, reducing extraction latency beyond the demand-driven approach of \Cref{sec:signal_engine}.
Second, \emph{information-ordered evaluation}: evaluating high-$I(\tau_i; Y)$ signals first and short-circuiting when the decision outcome is already determined---analogous to early termination in decision trees---can reduce average per-request evaluation cost.
The minimum description length (MDL) principle~\cite{rissanen1978modeling} provides a complementary perspective: the optimal signal subset is the one that describes the routing policy with minimum total code length, balancing signal extraction cost against routing precision.

\subsection{Interpretation as Hybrid Embedding}
\label{sec:disc_embedding}

The signal extraction layer maps a raw request~$r$ to a structured signal vector $\mathbf{s} = S(r) \in \{0,1\}^N \times [0,1]^N$, where each dimension corresponds to a named signal condition and its associated confidence.
This operation is functionally analogous to the \emph{embedding layer} of a Transformer~\cite{vaswani2017attention}: both convert unstructured input (token sequences or natural-language queries) into a high-dimensional representation suitable for downstream computation.

However, the analogy is not merely structural---it reveals a key design distinction.
Transformer embeddings produce dense, opaque vectors optimized end-to-end for a downstream task.
The signal vector~$\mathbf{s}$, by contrast, is \emph{interpretable by construction}: each dimension has a human-readable name (e.g., \texttt{domain:mathematics}, \texttt{complexity:high}), a defined semantic type, and an independently auditable extraction pipeline.
This interpretability is not incidental but architecturally enforced: the decision engine references signals by type and name, so the signal vector must be a \emph{structured, symbolic representation} rather than a latent embedding.

We characterize this as a \emph{hybrid embedding}: the extraction methods span a spectrum from sub-millisecond heuristics (keyword matching, regex patterns) to learned neural classifiers (LoRA-based domain and complexity models, embedding similarity), yet all project onto the same interpretable coordinate system.
\Cref{fig:embedding_comparison} illustrates this comparison.

\begin{figure*}[!ht]
  \centering
  \includegraphics[width=0.94\linewidth]{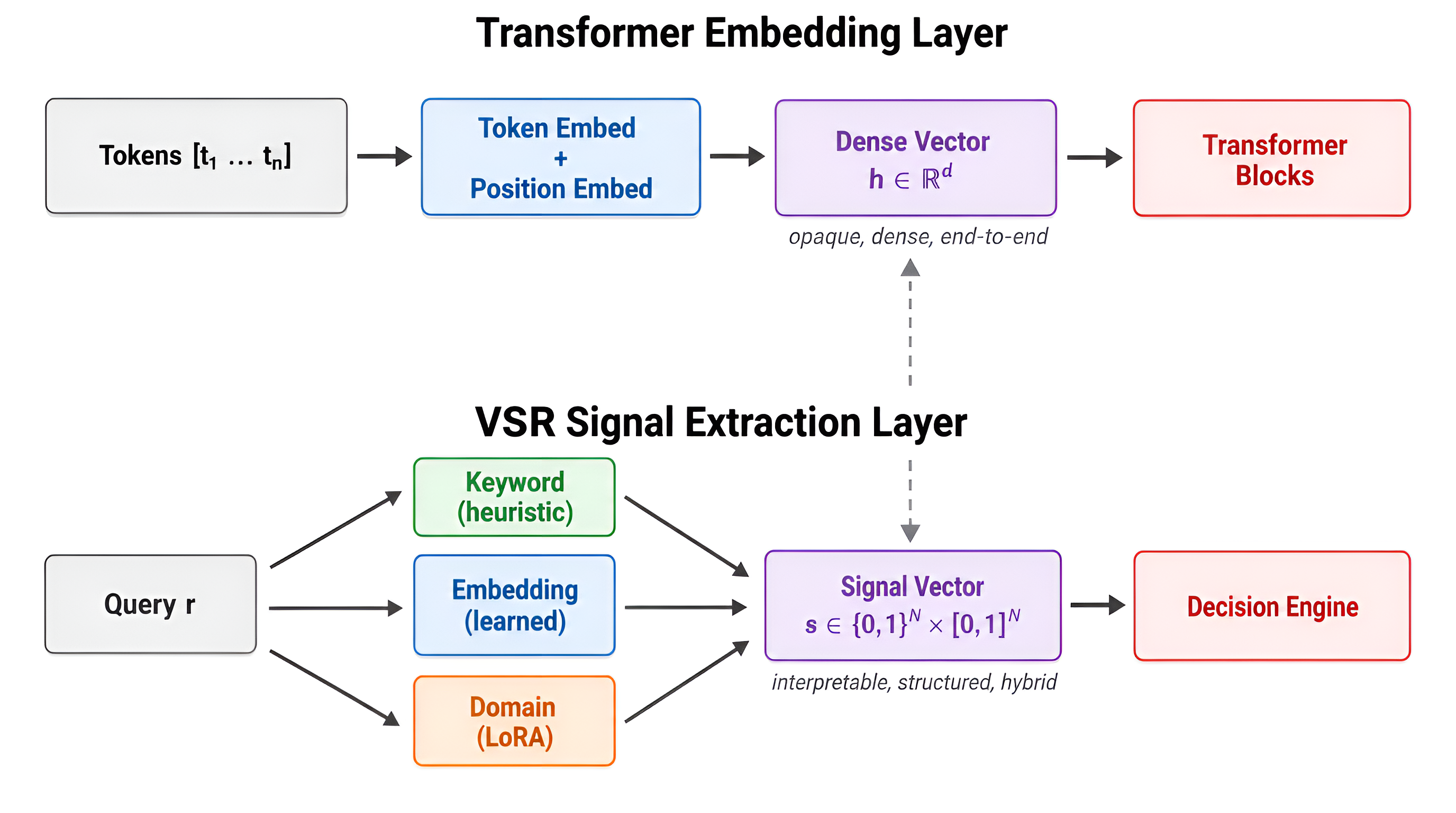}
  \caption{Comparison of embedding strategies. \textbf{Top}: A Transformer embedding layer produces dense, opaque vectors optimized end-to-end. \textbf{Bottom}: The VSR signal extraction layer produces structured, interpretable signal vectors from heterogeneous extractors (heuristic, learned, LoRA-based), where each dimension has a human-readable semantic name. The dashed arrow marks the functional correspondence: both convert unstructured input into a representation for downstream computation, but the signal vector is interpretable by construction.}
  \label{fig:embedding_comparison}
\end{figure*}


\section{Decision Engine}
\label{sec:decision_engine}

The decision engine evaluates a set of routing decisions against the signal result and selects the best match.
In Shannon's 1938 master's thesis~\cite{shannon1938symbolic}, arbitrary switching-circuit behavior was shown to be representable and minimizable through Boolean algebra---any desired input--output relation over binary signals could be systematically synthesized from $\{\wedge, \vee, \neg\}$.
The decision engine applies this same principle to routing: each routing policy is expressed as a Boolean formula over binary signal conditions, and the functionally complete operator set guarantees that \emph{any} routing policy expressible as a function of the extracted signals can be realized---without modifying the signal layer or the execution layer.
We formalize the decision model, present the evaluation algorithm, and analyze the selection strategies.

\subsection{Decision Model}

\begin{definition}[Decision]
A \emph{decision} $d = (n, \phi, \mathcal{M}_d, \Pi_d, p)$ consists of a name $n$, a Boolean formula $\phi$ over signal conditions, a candidate model set $\mathcal{M}_d \subseteq \mathcal{M}$, a plugin configuration $\Pi_d$, and a priority $p \in \mathbb{Z}$.
\end{definition}

\begin{definition}[Rule Node --- Boolean Expression Tree]
A \emph{rule node} $\phi$ is defined recursively as one of two forms (\Cref{fig:rule_node_tree} illustrates example expression trees):
\begin{itemize}[nosep,leftmargin=1.5em]
  \item \textbf{Leaf} (signal reference): $\phi = \operatorname{leaf}(\tau, n)$, referencing a signal type $\tau$ and rule name $n$.
  \item \textbf{Composite} (Boolean operator): $\phi = (\mathsf{op}, [\phi_1, \ldots, \phi_k])$, where $\mathsf{op} \in \{\operatorname{and}, \operatorname{or}, \operatorname{not}\}$ and $\phi_1, \ldots, \phi_k$ are child rule nodes. $\operatorname{not}$ is strictly unary ($k = 1$).
\end{itemize}
Evaluation proceeds by structural recursion:
\begin{equation}
  \text{eval}(\phi, S(r)) =
  \begin{cases}
    \mathbf{1}\bigl[\exists\, (\rho, 1, c) \in S(r) : \rho.\tau = \tau \wedge \rho.n = n\bigr]
      & \text{if } \phi = \operatorname{leaf}(\tau, n) \\[4pt]
    \bigwedge_{i=1}^{k} \text{eval}(\phi_i, S(r))
      & \text{if } \mathsf{op} = \operatorname{and} \\[4pt]
    \bigvee_{i=1}^{k} \text{eval}(\phi_i, S(r))
      & \text{if } \mathsf{op} = \operatorname{or} \\[4pt]
    \neg\, \text{eval}(\phi_1, S(r))
      & \text{if } \mathsf{op} = \operatorname{not}
  \end{cases}
\end{equation}
\end{definition}

The recursive structure enables arbitrarily nested Boolean expressions.
Classical flat formulas---a single AND or OR over leaf conditions---are the depth-1 special case and remain the recommended form for simple routing policies.
When richer logic is needed, nesting expresses compound operators directly within a single decision:
$\operatorname{nor}(A,B) = \operatorname{not}(\operatorname{or}(A,B))$;
$\operatorname{nand}(A,B) = \operatorname{not}(\operatorname{and}(A,B))$;
$\operatorname{xor}(A,B) = \operatorname{or}(\operatorname{and}(A, \operatorname{not}(B)),\; \operatorname{and}(\operatorname{not}(A), B))$.
YAML's natural indentation mirrors the logical nesting, preserving readability and auditability even for complex formulas.
Priority-ordered decisions compose multiple such formulas into an ordered evaluation, providing conflict resolution and organizational structure for deployment-scale routing policies.

\begin{figure}[ht]
  \centering
  \includegraphics[width=0.94\linewidth]{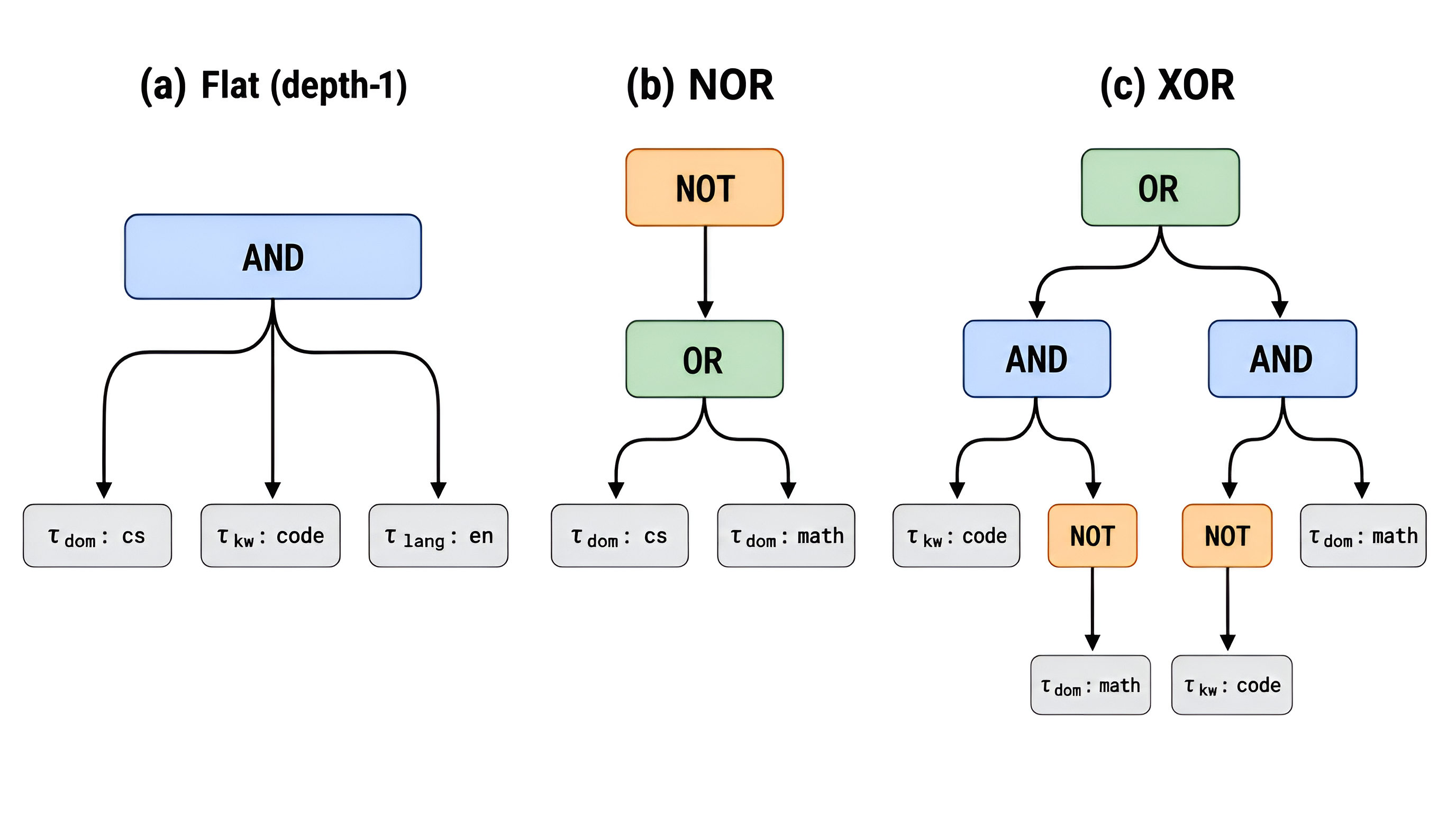}
  \caption{Rule-node expression trees at increasing depth. (a)~A flat depth-1 tree: AND over three leaf conditions. (b)~A NOR expression: $\textsc{not}(\textsc{or}(\text{cs},\text{math}))$, matching all non-STEM queries. (c)~An XOR expression composed from AND, OR, and NOT primitives, routing requests that match exactly one of two signals. Leaf nodes (gray) reference signal conditions; composite nodes use AND (\textcolor{blue!50}{blue}), OR (\textcolor{green!50}{green}), and NOT (\textcolor{orange!45}{orange}).}
  \label{fig:rule_node_tree}
\end{figure}

\subsection{Confidence Computation}

When a decision matches, we compute a confidence score as the mean confidence over satisfied conditions:

\begin{equation}
  \text{conf}(d, S(r)) = \frac{1}{|\Gamma_\text{sat}|} \sum_{\gamma_j \in \Gamma_\text{sat}} c_j(r)
  \label{eq:confidence}
\end{equation}

where $\Gamma_\text{sat} = \{\gamma_j \in \Gamma \mid \text{sat}(\gamma_j, S(r)) = 1\}$ and $c_j(r)$ is the signal confidence for condition $\gamma_j$.
For embedding signals, $c_j$ is the cosine similarity; for heuristic and binary ML signals, $c_j = 1.0$.

\subsection{Selection Strategies}

Given the set of matched decisions $\mathcal{D}_\text{match} = \{d \in \mathcal{D} \mid \text{eval}(\phi_d, S(r)) = 1\}$, two strategies select $d^*$:

\noindent\textbf{Priority Strategy.}
\begin{equation}
  d^* = \arg\max_{d \in \mathcal{D}_\text{match}} p_d
\end{equation}
This provides deterministic, administrator-controlled routing.
Ties are broken by insertion order.

\noindent\textbf{Confidence Strategy.}
\begin{equation}
  d^* = \arg\max_{d \in \mathcal{D}_\text{match}} \text{conf}(d, S(r))
\end{equation}
This enables data-driven routing where embedding similarity and classifier confidence drive selection.

The priority strategy is the default for production deployments where predictability is paramount.
The confidence strategy is preferred for experimental settings where the system should adapt to query characteristics.

\subsection{Evaluation Algorithm}

\begin{algorithm}[t]
\caption{Decision Evaluation}
\label{alg:decision_eval}
\begin{algorithmic}[1]
\REQUIRE Signal result $S(r)$, decisions $\mathcal{D}$, strategy $\sigma \in \{\text{priority}, \text{confidence}\}$
\ENSURE Selected decision $d^*$, confidence $c^*$
\STATE $\mathcal{D}_\text{match} \leftarrow \emptyset$
\FOR{$d \in \mathcal{D}$}
  \IF{$\text{eval}(\phi_d, S(r))$}
    \STATE $c_d \leftarrow \text{conf}(d, S(r))$
    \STATE $\mathcal{D}_\text{match} \leftarrow \mathcal{D}_\text{match} \cup \{(d, c_d)\}$
  \ENDIF
\ENDFOR
\IF{$\sigma = \text{priority}$}
  \STATE $(d^*, c^*) \leftarrow \arg\max_{(d, c) \in \mathcal{D}_\text{match}} p_d$
\ELSE
  \STATE $(d^*, c^*) \leftarrow \arg\max_{(d, c) \in \mathcal{D}_\text{match}} c$
\ENDIF
\RETURN $(d^*, c^*)$
\end{algorithmic}
\end{algorithm}

The algorithm runs in $O(M \cdot L_{\max})$ where $M = |\mathcal{D}|$ is the number of decisions and $L_{\max}$ is the maximum number of conditions per decision.
In practice, $M \leq 50$ and $L_{\max} \leq 10$, making decision evaluation negligible ($< 0.1$\,ms) relative to signal extraction.

\subsection{Expressiveness Analysis}

The recursive rule-node model can express common routing patterns with depth-1 trees (flat formulas) and richer patterns through nesting:

\begin{itemize}[leftmargin=*]
  \item \textbf{Domain routing}: A single leaf condition routes by classified domain.
  \item \textbf{Guarded routing}: AND of a domain condition and a complexity condition routes complex queries within a domain to a capable model.
  \item \textbf{Exclusion routing}: $\textsc{and}(\text{domain}, \textsc{not}(\text{complexity}))$ routes simple queries within a domain to a lightweight model, avoiding the cost of a full-capability model for straightforward requests.
  \item \textbf{Multi-signal routing}: AND of keyword, embedding, and language conditions provides precise routing for specific query patterns.
  \item \textbf{Fallback chains}: Multiple decisions with decreasing priority and progressively broader conditions implement fallback routing.
  \item \textbf{NOR routing} (blanket exclusion): $\textsc{not}(\textsc{or}(\text{cs}, \text{math}, \text{physics}))$ routes all non-STEM queries to a general-purpose model without enumerating every non-STEM domain.
  \item \textbf{NAND routing} (conditional exemption): $\textsc{not}(\textsc{and}(\text{zh}, \text{code}))$ matches all requests \emph{except} Chinese-language code queries, useful for compliance-based routing.
  \item \textbf{XOR routing} (mutual exclusion): $\textsc{or}(\textsc{and}(A, \textsc{not}(B)),\; \textsc{and}(\textsc{not}(A), B))$ routes requests matching exactly one of two signals, enabling exclusive specialization.
  \item \textbf{Nested multi-signal}: $\textsc{and}(\textsc{or}(\text{cs}, \text{math\_kw}),\; \text{en},\; \textsc{not}(\text{long\_ctx}))$ combines disjunctive, conjunctive, and negated sub-trees in a single decision for fine-grained routing.
\end{itemize}

\noindent\textbf{Functional completeness.}
The expressiveness of the recursive rule-node model rests on a classical result from Boolean algebra~\cite{huntington1904sets}: the operator set $\{\wedge, \vee, \neg\}$ is \emph{functionally complete}---any Boolean function $f: \{0,1\}^N \to \{0,1\}$ can be expressed as a formula over these operators.
Because each rule node may nest AND, OR, and NOT to arbitrary depth, a \emph{single} decision formula $\phi$ can already represent any Boolean function over the $N$ signal match indicators---no multi-decision composition is required for completeness.

\begin{proposition}[Single-decision completeness]
For any Boolean function $f: \{0,1\}^N \to \{0,1\}$ over signal match indicators, there exists a rule node $\phi$ using AND, OR, and NOT such that $\text{eval}(\phi, S(r)) = f(S(r))$ for all signal results $S(r)$.
\end{proposition}

\begin{proof}[Proof sketch]
Construct $\phi$ directly from the truth table of $f$.
For each minterm (input assignment where $f = 1$), form an AND-node over the corresponding literals (a leaf $\textsc{leaf}(\tau_i, n_i)$ when the $i$-th signal is 1, or $\textsc{not}(\textsc{leaf}(\tau_i, n_i))$ when it is 0).
Collect all such AND-nodes under a single OR-node.
Since $\{\wedge, \vee, \neg\}$ is functionally complete, the resulting tree evaluates to $f$.
\end{proof}

\begin{proposition}[Routing policy completeness]
For any routing policy $\pi: \{0,1\}^N \to \mathcal{M} \cup \{\bot\}$ mapping signal vectors to model selections, there exists a decision set $\mathcal{D}$ with recursive AND/OR/NOT formulas and priority ordering such that $\pi$ is realized by the evaluation algorithm (\Cref{alg:decision_eval}).
\end{proposition}

\begin{proof}[Proof sketch]
For each model $m_k$ in the range of $\pi$, construct a single decision $d_k$ whose rule node encodes the Boolean function $f_k(\mathbf{s}) = \mathbf{1}[\pi(\mathbf{s}) = m_k]$ using the single-decision completeness result above.
Assign priorities to resolve overlaps deterministically.
\end{proof}

This two-level universality guarantee---completeness within a single decision \emph{and} across a decision set---means the decision engine imposes \emph{no inherent limitation} on what routing policies can be configured.
Priority ordering across decisions is no longer required for expressiveness; it serves as an organizational and conflict-resolution mechanism, equivalent to a decision list~\cite{rivest1987learning} over Boolean features.

\noindent\textbf{Structural analogy to combinational logic circuits.}
The recursive rule-node model maps naturally onto the hierarchy of combinational logic circuits studied in digital design~\cite{brayton1984logic}.
The following table summarizes the correspondence at three levels of generality:

\begin{center}
\begin{tabularx}{\linewidth}{
  >{\raggedright\arraybackslash}p{3.0cm}
  >{\raggedright\arraybackslash}p{3.0cm}
  >{\raggedright\arraybackslash}X
}
\toprule
\textbf{Circuit Model} & \textbf{Routing System} & \textbf{Expressiveness} \\
\midrule
PLA (two-level)~\cite{fleisher1975introduction}
  & Flat decision (depth-1 tree)
  & Any two-level Boolean function \\
General combinational circuit
  & Recursive rule-node tree
  & Any Boolean function \\
Circuit array + priority encoder
  & Decision set + priority ordering
  & Any routing policy $\pi: \{0,1\}^N \!\to\! \mathcal{M}$ \\
\bottomrule
\end{tabularx}
\end{center}

At the first level, a depth-1 rule node---a single AND or OR over leaf conditions---is structurally isomorphic to a \emph{Programmable Logic Array} (PLA)~\cite{fleisher1975introduction}: signal extractors correspond to input lines, leaf conditions to the AND-plane, and the top-level operator to the OR-plane.
At the second level, a recursive rule-node tree corresponds to a general combinational logic circuit---a directed acyclic graph of AND, OR, and NOT gates---where each decision is a complete circuit computing an arbitrary Boolean function.
At the third level, the priority-ordered decision set acts as an array of such circuits with a priority encoder selecting the output, realizing any routing policy.
\Cref{fig:circuit_analogy} illustrates this three-level correspondence.

\begin{figure*}[ht]
  \centering
  \includegraphics[width=0.94\linewidth]{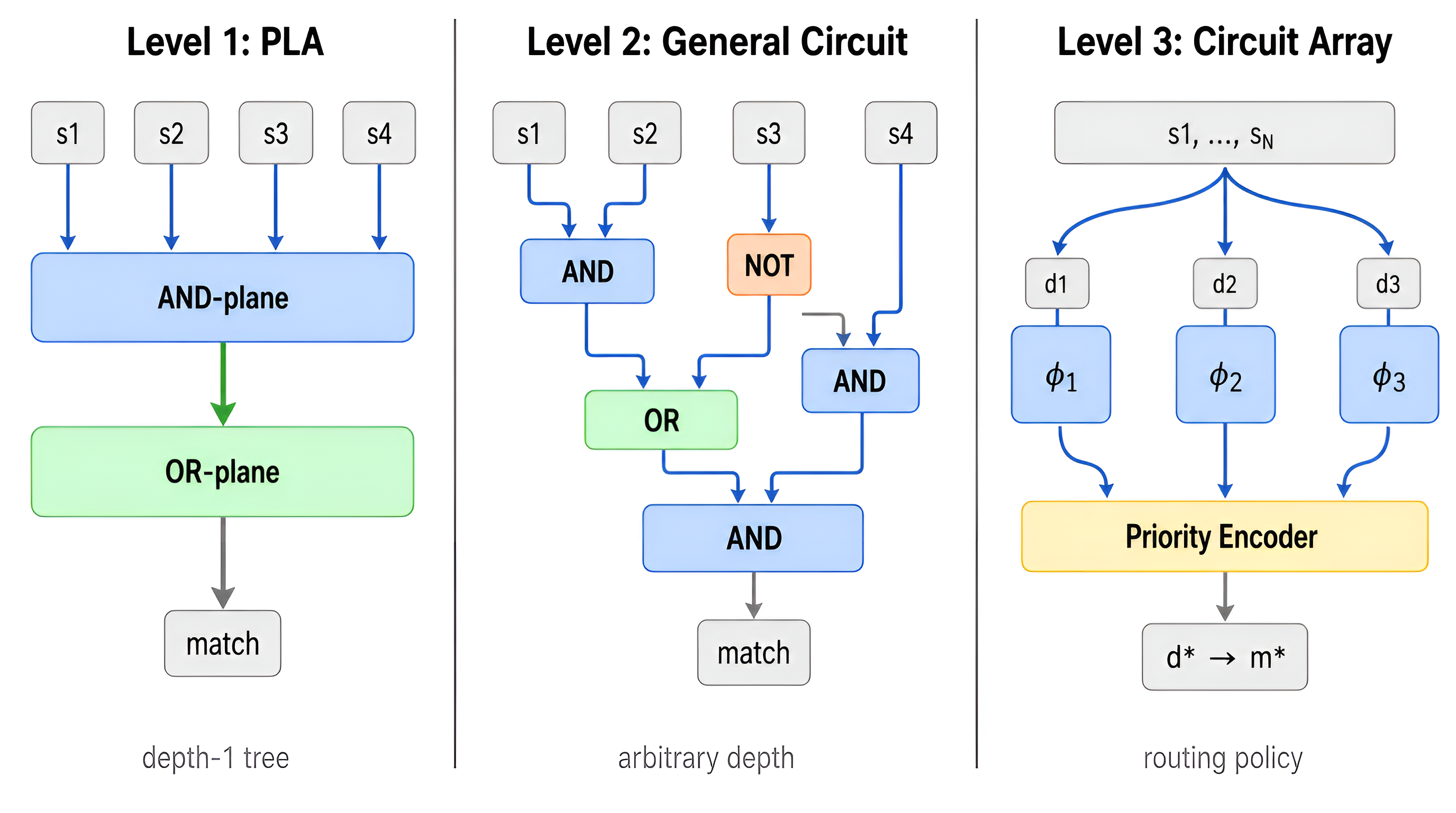}
  \caption{Three-level correspondence between combinational logic circuits and the decision engine. \textbf{Level~1}: A PLA with AND-plane and OR-plane corresponds to a flat (depth-1) decision formula. \textbf{Level~2}: A general combinational circuit with arbitrarily nested AND, OR, and NOT gates corresponds to a recursive rule-node tree within a single decision. \textbf{Level~3}: An array of circuits with a priority encoder corresponds to the full decision set with priority-ordered evaluation, realizing any routing policy.}
  \label{fig:circuit_analogy}
\end{figure*}

In a PLA, logic is ``programmed'' by setting fuse connections; in a general circuit, by wiring gates; in our system, by YAML configuration whose indentation directly mirrors the gate-level nesting.
This correspondence provides a well-understood theoretical foundation: decades of results on logic minimization, hazard-free design, and testability~\cite{brayton1984logic} apply directly to reasoning about decision optimization and coverage analysis.

\noindent\textbf{Decision set verification and minimization.}
The combinational-logic correspondence makes formal analysis tools from logic synthesis directly applicable to routing configurations.
\emph{Coverage analysis} checks whether every reachable point in the signal space $\{0,1\}^N$ is matched by at least one decision, identifying dead zones where requests would receive no routing directive.
\emph{Conflict detection} identifies signal combinations where multiple decisions match but route to incompatible model pools, flagging ambiguities that priority ordering must resolve.
\emph{Decision minimization}, analogous to the Espresso heuristic for two-level logic optimization~\cite{brayton1984logic} and multi-level logic restructuring~\cite{brayton1984logic}, can reduce a decision set to a minimal equivalent form by merging decisions with compatible conditions and eliminating subsumed rules.
These standard logic-verification techniques become applicable to routing policy validation without adaptation, a direct consequence of the structural isomorphism.

\subsection{Generalization to Fuzzy Evaluation}

The Boolean decision model admits a natural generalization when signal confidence scores are continuous.
Rather than binarizing each signal's output before Boolean combination, we evaluate rule-node trees over the continuous confidence values directly, using fuzzy logic operators~\cite{zadeh1965fuzzy}.

\begin{definition}[Fuzzy Rule-Node Evaluation]
Given continuous signal confidences $c(r) \in [0,1]$ for each leaf condition, the fuzzy evaluation of a rule node $\phi$ is defined by structural recursion:
\begin{equation}
  \widetilde{\text{eval}}(\phi, S(r)) =
  \begin{cases}
    c_\phi(r)
      & \text{if } \phi = \operatorname{leaf}(\tau, n) \\[6pt]
    \displaystyle\min_{i=1}^{k}\, \widetilde{\text{eval}}(\phi_i, S(r))
      & \text{if } \mathsf{op} = \operatorname{and} \\[6pt]
    \displaystyle\max_{i=1}^{k}\, \widetilde{\text{eval}}(\phi_i, S(r))
      & \text{if } \mathsf{op} = \operatorname{or} \\[6pt]
    1 - \widetilde{\text{eval}}(\phi_1, S(r))
      & \text{if } \mathsf{op} = \operatorname{not}
  \end{cases}
\end{equation}
where $c_\phi(r)$ is the confidence score of the signal matching leaf $\phi$.
\end{definition}

The operators $(\min, \max, 1{-}x)$ form the standard fuzzy complement triple and satisfy De~Morgan's laws at every level of the tree, preserving the algebraic properties of the crisp model~\cite{bellman1970decision}.
This fuzzy evaluation is a \emph{strict generalization}: when all confidences are binary ($c \in \{0,1\}$), $\min$ reduces to $\wedge$, $\max$ reduces to $\vee$, and $1{-}x$ reduces to $\neg$, so the evaluation coincides exactly with the crisp Boolean model.

The practical consequence is significant.
The current confidence strategy (\Cref{eq:confidence}) uses mean confidence as a tiebreaker \emph{after} binary matching.
Fuzzy evaluation incorporates confidence \emph{during} formula evaluation: a decision with three conditions matched at confidences $(0.95, 0.88, 0.72)$ yields a fuzzy AND score of $0.72$, while a decision with two conditions at $(0.99, 0.98)$ scores $0.98$---correctly preferring the more confident partial match even when both decisions pass binary evaluation.
The functional completeness result of the previous section extends directly: the fuzzy operator triple $(\min, \max, 1{-}x)$ is functionally complete over the continuous lattice $[0,1]$, so any monotone routing policy over continuous signal scores is realizable.

\subsection{Composable Decision Profiles}

The decision model directly enables \emph{composable signal orchestration}: different deployment scenarios are expressed as different decision sets $\mathcal{D}$ over the same signal infrastructure.
A healthcare deployment defines decisions with authz and domain conditions routing to compliant model pools; a developer-tool deployment defines decisions with complexity and keyword conditions routing to cost-optimized cascades; a multi-cloud deployment defines decisions with domain and modality conditions, using latency-aware model selection across provider endpoints.

Formally, switching deployment scenarios corresponds to loading a different decision profile $\mathcal{D}_\Gamma$, while the signal extraction layer $\mathcal{S}$ and plugin implementations $\Pi$ remain unchanged.
This separation of \emph{policy} (what decisions to evaluate) from \emph{mechanism} (how signals are computed and plugins execute) is the architectural basis for the composability claimed in \Cref{sec:architecture}.

\subsection{Interpretation as Mixture-of-Experts Gating}
\label{sec:disc_moe}

The priority-ordered decision evaluation admits a precise structural analogy to the \emph{Mixture-of-Experts} (MoE) gating mechanism~\cite{shazeer2017moe}.
In a standard MoE layer, a gating network $G(\mathbf{x})$ computes a sparse distribution over $K$ expert sub-networks, and the output is a weighted combination of the selected experts' outputs.
In our system:

\begin{itemize}
  \item The \textbf{signal vector} $\mathbf{s}$ plays the role of the shared representation that the gating network operates on.
  \item Each \textbf{decision block} $d_i$ with its Boolean formula $\phi_i(\mathbf{s})$ acts as an \emph{expert gate}---a binary function that determines whether expert~$i$ (a model pool with associated plugins) is activated.
  \item \textbf{Priority ordering} implements \emph{hard routing with early exit}: the first decision whose gate evaluates to true captures the request, analogous to top-1 expert selection in sparse MoE but with a deterministic, priority-based selection rule rather than a learned softmax.
\end{itemize}

\Cref{fig:moe_analogy} formalizes this correspondence.
Unlike standard MoE, the gates here are symbolic (Boolean formulas over interpretable signals) rather than learned softmax projections, which enables formal verifiability and compositional editability.

\begin{figure}[H]
  \centering
  \includegraphics[width=0.94\linewidth]{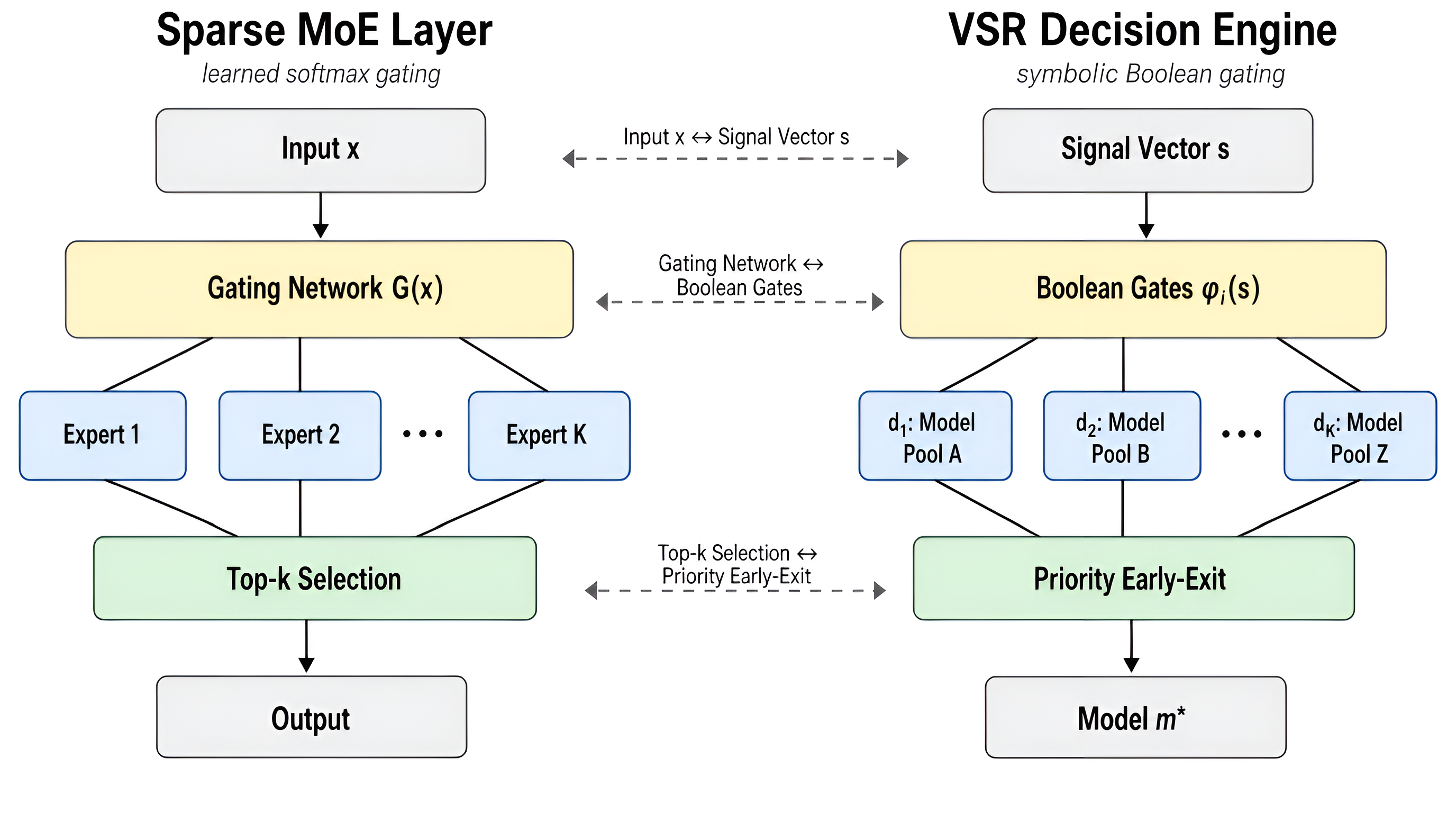}
  \caption{Structural analogy between Mixture-of-Experts gating and the VSR decision engine. \textbf{Left}: In a sparse MoE layer, a learned gating network selects top-$k$ experts via softmax. \textbf{Right}: In VSR, Boolean formulas over the signal vector act as symbolic expert gates, with priority ordering implementing deterministic early-exit selection. Dashed arrows mark the structural correspondence.}
  \label{fig:moe_analogy}
\end{figure}

\subsection{Layered Entropy-Folding Interpretation}
\label{sec:disc_entropy_folding}

The priority-ordered gate stack also admits an information-theoretic interpretation as a layered uncertainty-collapse process.
Let $g_\ell(\mathbf{s}) \in \{0,1\}$ denote the match outcome of the $\ell$-th gate under priority order, and let $Z_\ell = g_\ell(\mathbf{s})$.
Define routing uncertainty after $\ell$ evaluated gates as:
\begin{equation}
  U_\ell = H(M \mid \mathbf{s}, Z_{1:\ell})
\end{equation}
where $M$ is the final selected model random variable.
By the chain rule of mutual information:
\begin{equation}
  U_{\ell+1}
  = U_\ell - I\!\left(M;\, Z_{\ell+1}\mid \mathbf{s}, Z_{1:\ell}\right)
  \le U_\ell
\end{equation}
so each additional gate can only maintain or reduce uncertainty.
In this sense, priority depth is \emph{control depth}: a sequence of policy constraints that folds entropy toward a deterministic choice.

\begin{figure}[!ht]
  \centering
  \includegraphics[width=0.94\linewidth]{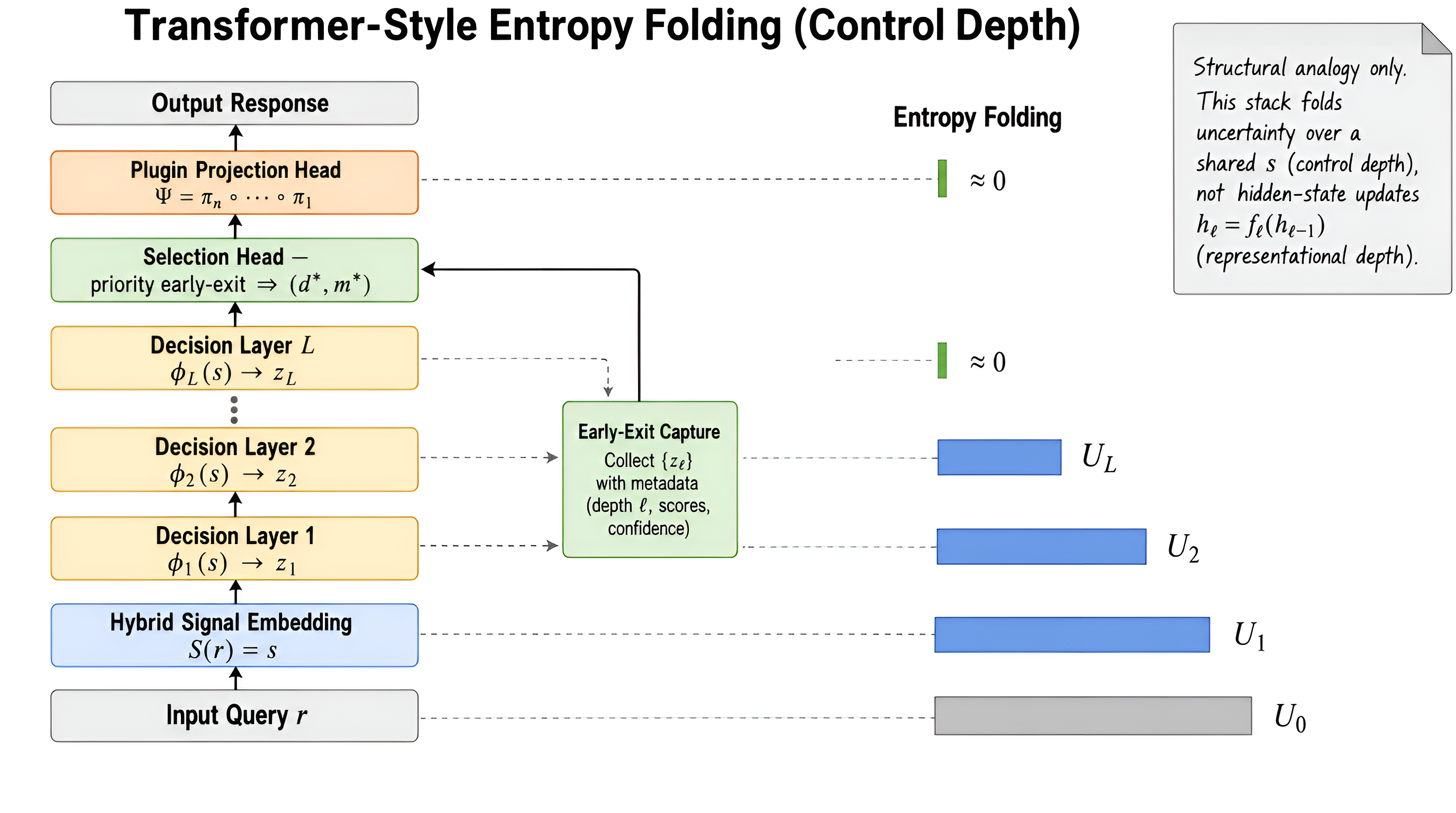}
  \caption{Transformer-style visualization of layered control depth in VSR.
  Signal extraction constructs a shared embedding $\mathbf{s}$; stacked decision layers apply Boolean gates $\phi_\ell(\mathbf{s})$ that progressively reduce routing uncertainty $U_\ell$; selection resolves $(d^*,m^*)$ via priority early-exit; plugin composition $\Psi_{d^*}$ acts as an output-side projection head over request/response behavior.}
  \label{fig:entropy_folding_layers}
\end{figure}

The MoE analogy also clarifies what the system is \emph{not}: it is not a Transformer in the sequential sense.
Transformer blocks transform the representation layer by layer---each block's output is the next block's input.
In our architecture, all decision blocks operate on the \emph{same} shared signal vector~$\mathbf{s}$; there is no inter-decision information flow.
The priority ordering is an \emph{early-exit mechanism} over parallel evaluations, not a sequential transformation pipeline.
This distinction is important: it means that adding or removing decisions does not affect the behavior of other decisions (modulo priority reordering), a composability property that sequential architectures lack.


\section{Plugin Framework}
\label{sec:plugins}

The plugin layer provides a composable middleware architecture where each matched decision activates an independent chain of typed transformations.
We describe the plugin model and four core infrastructure plugins; safety signal classifiers are covered in \Cref{sec:safety}, hallucination detection in \Cref{sec:halugate}, and memory retrieval and RAG injection in \Cref{sec:memory_rag}.

\subsection{Plugin Execution Model}

Formally, a plugin $\pi$ is a typed transformation on the request-response pair:
\begin{equation}
  \pi: (\text{Request}, \text{Context}, \text{Config}_\pi) \to (\text{Request}', \text{Response}') \cup \{\bot\}
\end{equation}
where $\bot$ denotes early termination (e.g., a cache hit returning immediately or a fast response short-circuiting the pipeline).

Plugins execute in a fixed pipeline order within each decision's chain.
On the \emph{request path}: fast response $\to$ cache $\to$ RAG $\to$ modality $\to$ memory $\to$ system prompt $\to$ header mutation.
On the \emph{response path}: hallucination detection $\to$ cache write.
Each plugin is independently enabled per decision, and its configuration (thresholds, modes, policies) is scoped to that decision.

Jailbreak and PII detection are first-class \emph{signals} (\Cref{sec:signal_engine,sec:safety}) evaluated in parallel with all other signal types, not serial plugins in the request path.
Decisions that match safety signals activate the \texttt{fast\_response} plugin to return an immediate rejection.

This per-decision scoping is a key architectural distinction from systems that apply safety and caching globally: it allows differentiated policies for different routing outcomes within the same deployment.

\subsection{Output-Space Projection Perspective}

For a matched decision $d^*$, the plugin chain can be interpreted as a policy-scoped projection from routing state to deployment-compliant output behavior:
\begin{equation}
  \Psi_{d^*} = \pi_n \circ \pi_{n-1} \circ \cdots \circ \pi_1
\end{equation}
where each $\pi_i$ is a typed transformation in the configured chain.
Functionally, this is analogous to an output projection stage after gating: once the decision fixes the active policy path, plugins realize that policy as concrete transformations on request/response state (e.g., context injection, auth header mutation, hallucination annotation, cache write-through).
This analogy is operational rather than architectural: unlike a Transformer projection head, $\Psi_{d^*}$ is a configurable symbolic composition, not a learned linear map.

\subsection{Semantic Cache}

The semantic cache exploits the observation that semantically similar queries often produce equivalent responses, avoiding redundant model invocations.

\noindent\textbf{Similarity model.}
Given a query $q$ extracted from request $r$, the cache searches for an entry $e$ such that:
\begin{equation}
  \cos(\mathbf{e}_q, \mathbf{e}_e) \geq \theta_d
\end{equation}
where $\mathbf{e}_q, \mathbf{e}_e$ are embeddings computed by the shared embedding model and $\theta_d$ is the per-decision similarity threshold.
On hit, the cached response is returned immediately, bypassing model invocation entirely.

\noindent\textbf{Write-through protocol.}
On cache miss, a pending entry is registered before forwarding to the model.
Upon receiving the response, the entry is completed with the response content.
This ensures that concurrent identical queries observe the pending state rather than triggering redundant model calls.

\noindent\textbf{Backend abstraction.}
Four backends provide different latency-persistence tradeoffs:
(1)~in-memory HNSW for single-node low-latency deployments;
(2)~Redis for distributed persistent caching;
(3)~Milvus for large-scale approximate nearest neighbor search;
(4)~a hybrid two-tier design combining in-memory HNSW (fast path) with Milvus (persistent store).

\subsection{System Prompt Injection}

Per-decision system prompt injection enables different routing paths to carry different instructions.
Two composition modes are defined:
\begin{itemize}[leftmargin=*]
  \item \textbf{Replace}: Substitutes the entire system message, providing complete control over the model's behavioral context.
  \item \textbf{Insert}: Prepends the decision's prompt to the existing system message, augmenting without overriding user-provided instructions.
\end{itemize}

This enables patterns such as injecting domain-specific instructions for expert routing or safety preambles for sensitive query categories.

\subsection{Header Mutation}

Header mutation enables metadata propagation to downstream model backends via HTTP header modifications (add, update, delete).
This supports use cases including:
backend-specific authentication injection,
routing decision metadata propagation for downstream observability,
and custom signaling to model-serving frameworks (e.g., LoRA adapter selection via headers).

\subsection{Fast Response}

The fast response plugin enables a decision to short-circuit the entire pipeline and return an immediate OpenAI-compatible response without forwarding the request to any upstream model.
This is the primary mechanism for acting on safety signals.

\noindent\textbf{Execution model.}
When a decision's \texttt{fast\_response} plugin is configured, the pipeline checks for it before any other plugin in the chain.
If present, the configured message is returned as a standard \texttt{chat.completion} response with \texttt{finish\_reason: "stop"}, and no further processing occurs.

\noindent\textbf{Streaming compatibility.}
The plugin inspects the original request's \texttt{stream} parameter.
For non-streaming requests, it returns a single JSON object.
For streaming requests (\texttt{stream: true}), it generates a Server-Sent Events (SSE) sequence: an initial chunk with the assistant role, word-by-word content chunks, a final chunk with \texttt{finish\_reason: "stop"}, and the \texttt{[DONE]} sentinel---matching the exact format that OpenAI-compatible clients expect.

\noindent\textbf{Use cases.}
The canonical use case is safety enforcement: a decision matching a jailbreak or PII signal activates \texttt{fast\_response} to return a policy-compliant refusal message.
However, the plugin is general-purpose---it can also serve canned responses for FAQ-like queries, maintenance windows, or rate-limited fallback messages without consuming model compute.


\section{Programmable Neural-Symbolic Configuration Language}
\label{sec:dsl}

While the YAML configuration format described in preceding sections is sufficient for direct human authorship, it conflates two concerns: the \emph{logical specification} of routing policy (signals, decisions, plugins) and the \emph{deployment target} (flat YAML, Kubernetes CRDs, Helm charts).
We introduce a programmable configuration language---the instruction set of the neural-symbolic inference engine (formalized in \Cref{sec:disc_agent})---that separates these concerns, providing a concise, human-readable, and machine-parseable surface syntax for routing policies that compiles to multiple deployment targets through a shared internal representation.

\subsection{Design Goals}

The DSL is designed around four principles:

\begin{enumerate}
  \item \textbf{Conciseness}: Routing policies should read as close to natural language as possible, minimizing syntactic noise relative to YAML.
  \item \textbf{Type safety}: Signal references in Boolean expressions are resolved at compile time; undefined or misspelled signal names produce diagnostics before deployment.
  \item \textbf{Multi-target emission}: A single DSL source compiles to flat YAML (for local development), Kubernetes \texttt{SemanticRouter} CRDs (for operator-based deployment), or Helm \texttt{values.yaml} (for chart-based deployment).
  \item \textbf{Round-trip fidelity}: An existing \texttt{RouterConfig} can be \emph{decompiled} back to DSL source, enabling migration from YAML to DSL and ``round-trip'' editing workflows.
\end{enumerate}

\subsection{Concrete Syntax}
\label{sec:dsl_syntax}

The DSL defines five top-level block types: \texttt{SIGNAL}, \texttt{ROUTE}, \texttt{PLUGIN}, \texttt{BACKEND}, and \texttt{GLOBAL}.
\Cref{fig:dsl_example} shows a representative configuration.

\begin{figure}[!htbp]
\centering
\footnotesize
\begin{verbatim}
SIGNAL domain math { mmlu_categories: ["math"] }
SIGNAL keyword urgent { operator: "any", keywords: ["urgent","asap"] }

PLUGIN safe_pii pii { enabled: true, pii_types_allowed: [] }

ROUTE math_route (description = "Math") {
  PRIORITY 100
  WHEN domain("math")
  MODEL "qwen2.5:3b" (reasoning = true, effort = "high")
  PLUGIN safe_pii
}
ROUTE urgent_ai {
  PRIORITY 200
  WHEN keyword("urgent") AND NOT domain("math")
  MODEL "qwen3:70b" (reasoning = true),
        "qwen2.5:3b" (reasoning = false)
  ALGORITHM confidence { threshold: 0.5 }
}

BACKEND vllm_endpoint ollama { address: "127.0.0.1", port: 11434 }
GLOBAL { default_model: "qwen2.5:3b", strategy: "priority" }
\end{verbatim}
\caption{A representative DSL configuration defining two signals, a reusable plugin template, two routes with Boolean decision logic, a backend endpoint, and global settings.}
\label{fig:dsl_example}
\end{figure}

Each \texttt{SIGNAL} block declares a named signal of a specific type (one of 12 supported types: \texttt{keyword}, \texttt{embedding}, \texttt{domain}, \texttt{fact\_check}, \texttt{user\_feedback}, \texttt{preference}, \texttt{language}, \texttt{context}, \texttt{complexity}, \texttt{modality}, \texttt{authz}, \texttt{jailbreak}, \texttt{pii}).
Each \texttt{ROUTE} block defines a decision with a priority, a Boolean \texttt{WHEN} clause over signal references, one or more model references with optional per-model parameters (reasoning mode, effort level, LoRA adapter, weight), an optional selection algorithm, and zero or more plugin attachments.
\texttt{PLUGIN} blocks at the top level define reusable templates that routes reference by name; inline plugin blocks within routes define route-specific configurations.

\subsection{Grammar and Parsing}
\label{sec:dsl_grammar}

The lexer and parser are implemented using the \texttt{participle} parser-generator library~\cite{participle2024}.
The lexer defines 12 token classes (identifiers, integers, floats, strings, booleans, braces, parentheses, brackets, colon, comma, equals, comments) via regular expressions.
The parser uses a PEG-style grammar with lookahead~3 to resolve ambiguities between signal references and other identifier uses.

The Boolean expression grammar for \texttt{WHEN} clauses follows standard precedence:

\begin{align}
  \textit{BoolExpr} &::= \textit{AndTerm} \;(\texttt{OR}\; \textit{AndTerm})^* \label{eq:bool_or} \\
  \textit{AndTerm} &::= \textit{Factor} \;(\texttt{AND}\; \textit{Factor})^* \label{eq:bool_and} \\
  \textit{Factor} &::= \texttt{NOT}\; \textit{Factor} \mid \texttt{(}\;\textit{BoolExpr}\;\texttt{)} \mid \textit{SignalRef} \label{eq:bool_factor} \\
  \textit{SignalRef} &::= \textit{type}\;\texttt{(}\;\textit{``name''}\;\texttt{)} \label{eq:signal_ref}
\end{align}

This grammar is parsed into a recursive AST (\texttt{BoolAnd}, \texttt{BoolOr}, \texttt{BoolNot}, \texttt{SignalRefExpr} nodes) that the compiler flattens into the \texttt{RuleCombination} tree consumed by the decision engine (\Cref{sec:decision_engine}).
\Cref{fig:bool_ast} illustrates the AST for a representative Boolean expression.

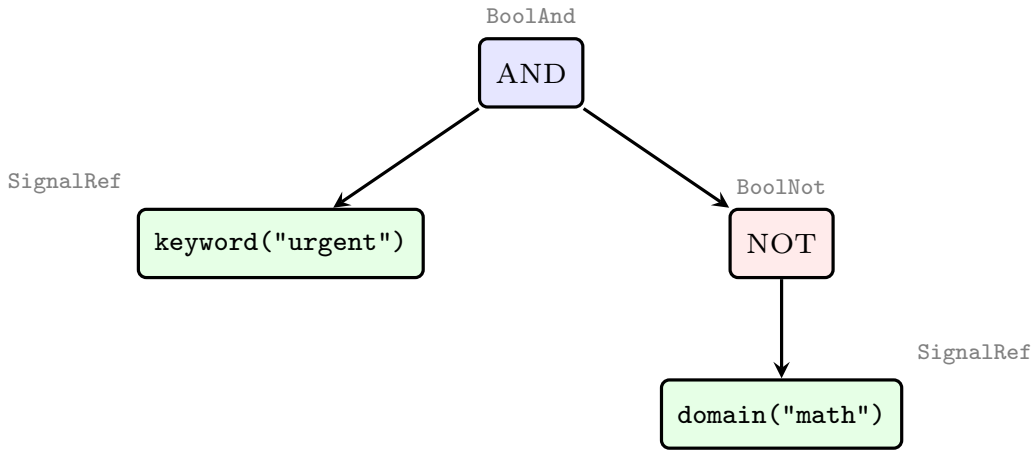
\begin{figure}[!ht]
  \centering
  \resizebox{\linewidth}{!}{%
\begin{tikzpicture}[
    tnode/.style={rectangle, draw, thick, rounded corners=2pt,
                  minimum height=0.6cm, align=center, inner sep=4pt, font=\scriptsize},
    arr/.style={->, >=stealth, thick},
    lbl/.style={font=\tiny, text=gray},
  ]

  \node[tnode, fill=blue!10] (and) at (0, 0) {AND};
  \node[lbl, anchor=south] at (and.north) {\texttt{BoolAnd}};

  \node[tnode, fill=green!10] (kw) at (-2.2, -1.5) {\texttt{keyword("urgent")}};
  \node[lbl, anchor=south east] at (kw.north west) {\texttt{SignalRef}};

  \node[tnode, fill=red!8] (not) at (2.2, -1.5) {NOT};
  \node[lbl, anchor=south] at (not.north) {\texttt{BoolNot}};

  \node[tnode, fill=green!10] (dom) at (2.2, -3.0) {\texttt{domain("math")}};
  \node[lbl, anchor=south west] at (dom.north east) {\texttt{SignalRef}};

  \draw[arr] (and) -- (kw);
  \draw[arr] (and) -- (not);
  \draw[arr] (not) -- (dom);

  \end{tikzpicture}%
}
  \caption{Boolean expression AST for \texttt{WHEN keyword("urgent") AND NOT domain("math")}. Each leaf is a \texttt{SignalRefExpr} referencing a named signal; internal nodes are Boolean operators (\texttt{BoolAnd}, \texttt{BoolNot}). The compiler flattens this tree into a \texttt{RuleCombination} structure for the decision engine.}
  \label{fig:bool_ast}
\end{figure}

Error recovery is block-granular: if parsing a top-level block fails, the parser splits the input at block boundaries and attempts to parse remaining blocks independently, accumulating partial results alongside error diagnostics.
This enables IDE-like incremental feedback during authoring.

\subsection{Compilation Pipeline}
\label{sec:dsl_compilation}

The compilation pipeline (\Cref{fig:dsl_pipeline}) transforms DSL source through four stages:

\begin{enumerate}
  \item \textbf{Lexing}: The source is tokenized into a stream of typed tokens with position tracking for diagnostic reporting.
  \item \textbf{Parsing}: The token stream is parsed into a \emph{raw parse tree} (participle-generated structs), then lowered to a \emph{resolved AST} (\texttt{Program} $\to$ \texttt{SignalDecl}, \texttt{RouteDecl}, \texttt{PluginDecl}, \texttt{BackendDecl}, \texttt{GlobalDecl}) with desugared values and resolved positions.
  \item \textbf{Compilation}: The AST is compiled to the internal \texttt{RouterConfig} structure. This involves: (a) mapping each \texttt{SIGNAL} block to the appropriate signal configuration (keyword rules, embedding rules, domain categories, etc.); (b) flattening the Boolean expression tree in each \texttt{WHEN} clause into a \texttt{RuleCombination} tree with \texttt{AND}/\texttt{OR}/\texttt{NOT} operators; (c) resolving plugin references against top-level templates with field-level merge semantics (route-local fields override template defaults); and (d) mapping \texttt{BACKEND} blocks to endpoint, provider profile, embedding model, or semantic cache configurations based on the backend type keyword.
  \item \textbf{Emission}: The \texttt{RouterConfig} is serialized to one of three target formats.
\end{enumerate}

\begin{figure*}[!ht]
  \centering
  \includegraphics[width=0.94\linewidth]{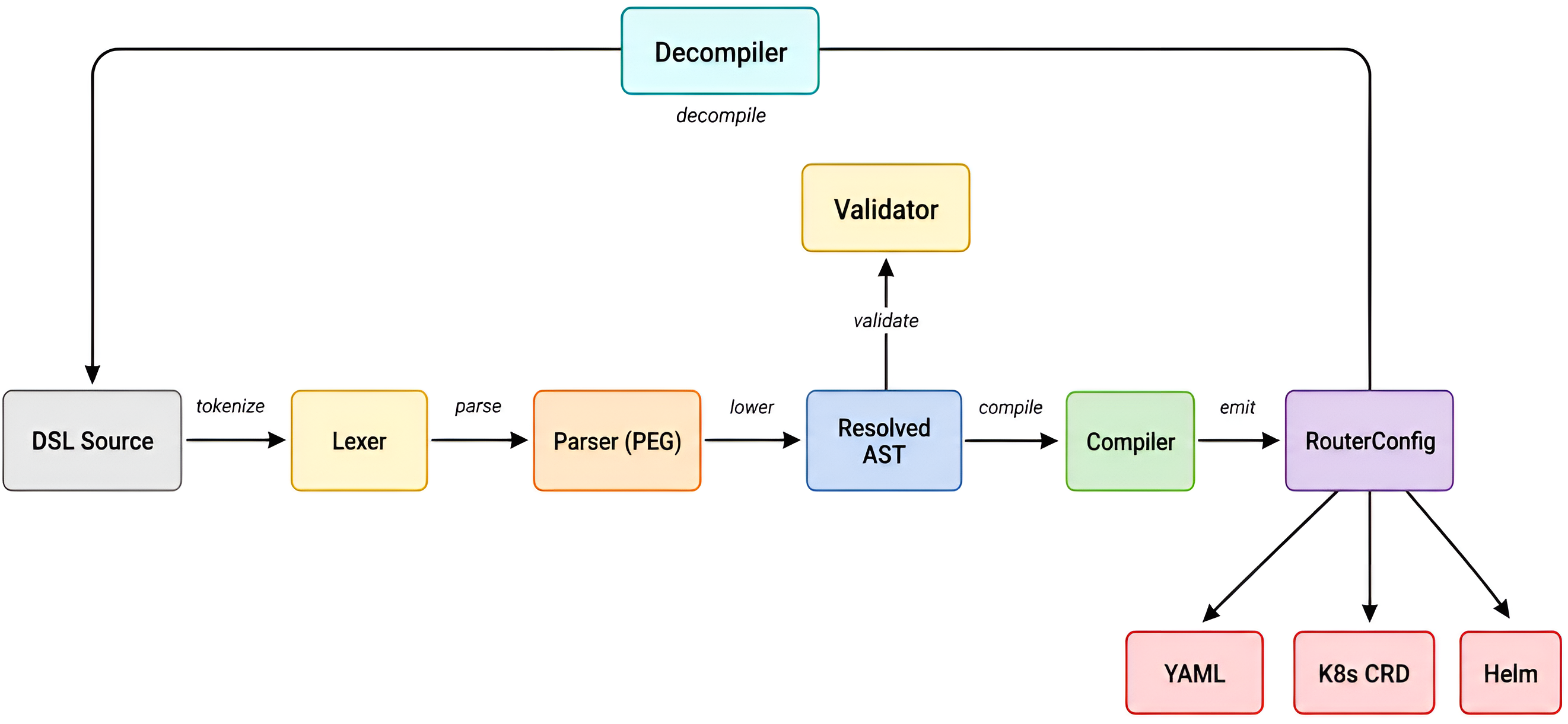}
  \caption{The DSL compilation pipeline. Source text is lexed, parsed into a resolved AST, and compiled to the internal \texttt{RouterConfig}. From this representation, three emitters produce deployment-specific formats (flat YAML, Kubernetes CRD, Helm values). The decompiler reverses the pipeline, reconstructing DSL source from an existing \texttt{RouterConfig}, enabling round-trip editing. The validator operates on the AST directly, producing three-level diagnostics.}
  \label{fig:dsl_pipeline}
\end{figure*}

\subsection{Multi-Target Emission}
\label{sec:dsl_emission}

The shared \texttt{RouterConfig} representation enables three emission targets from a single DSL source:

\begin{itemize}
  \item \textbf{Flat YAML} (\texttt{EmitYAML}): Direct marshaling of \texttt{RouterConfig} for local development and testing. An alternative \texttt{EmitUserYAML} variant restructures the output into the nested \texttt{signals}/\texttt{providers} format expected by the CLI tooling.
  \item \textbf{Kubernetes CRD} (\texttt{EmitCRD}): Wraps the routing logic in a \texttt{SemanticRouter} custom resource (\texttt{vllm.ai/v1alpha1}), mapping backend definitions to \texttt{spec.vllmEndpoints} and routing configuration to \texttt{spec.config}. Signal rules that the CRD schema does not model are preserved as extra fields for ConfigMap-based deployment compatibility.
  \item \textbf{Helm values} (\texttt{EmitHelm}): Nests the \texttt{RouterConfig} under a \texttt{config:} key compatible with the Helm chart's ConfigMap template, pruning zero-value infrastructure sections for clean output.
\end{itemize}

This separation means that infrastructure teams can change deployment targets without modifying the routing policy, and routing engineers can evolve policies without understanding Kubernetes manifests.

\subsection{Decompilation and Round-Trip Editing}
\label{sec:dsl_decompile}

The decompiler reconstructs DSL source text from an existing \texttt{RouterConfig}:

\begin{enumerate}
  \item \textbf{Plugin template extraction}: Plugins used by multiple routes are automatically factored into top-level \texttt{PLUGIN} templates; route-local plugins remain inline.
  \item \textbf{Rule tree reconstruction}: The \texttt{RuleCombination} tree in each decision is walked recursively to reconstruct the Boolean expression with proper \texttt{AND}/\texttt{OR}/\texttt{NOT} operators and parenthesization.
  \item \textbf{Signal type inference}: Signal references in rule nodes are matched against the configuration's signal lists to recover the original signal type keywords.
\end{enumerate}

This enables a migration path: existing YAML configurations can be decompiled to DSL, edited in the more concise syntax, and recompiled to any target format.
The round-trip property ($\text{DSL} \xrightarrow{\text{compile}} \texttt{RouterConfig} \xrightarrow{\text{decompile}} \text{DSL} \xrightarrow{\text{compile}} \texttt{RouterConfig}' \equiv \texttt{RouterConfig}$) is validated by extensive test suites including idempotency and double-round-trip tests.

\subsection{Three-Level Validation}
\label{sec:dsl_validation}

The validator operates on the resolved AST (before compilation) and produces diagnostics at three severity levels:

\begin{enumerate}
  \item \textbf{Error} (Level~1): Syntax errors detected during parsing---malformed blocks, unexpected tokens, missing delimiters. The block-granular error recovery ensures that errors in one block do not prevent analysis of subsequent blocks.
  \item \textbf{Warning} (Level~2): Reference resolution issues---a \texttt{WHEN} clause references a signal name not defined in any \texttt{SIGNAL} block, or a \texttt{PLUGIN} reference has no matching template. The validator performs fuzzy matching on undefined signal names and suggests corrections via \texttt{QuickFix} annotations (e.g., ``did you mean \texttt{math}?'' for a reference to \texttt{mth}).
  \item \textbf{Constraint} (Level~3): Semantic constraint violations---embedding thresholds outside $[0,1]$, port numbers outside valid ranges, negative route priorities, unknown algorithm or signal types. These catch logical errors that are syntactically valid but semantically incorrect.
\end{enumerate}

\Cref{fig:validation_levels} illustrates the three-level diagnostic architecture.
This scheme provides IDE-like progressive feedback, enabling both batch validation (CLI \texttt{validate} command) and interactive authoring with incremental diagnostics.

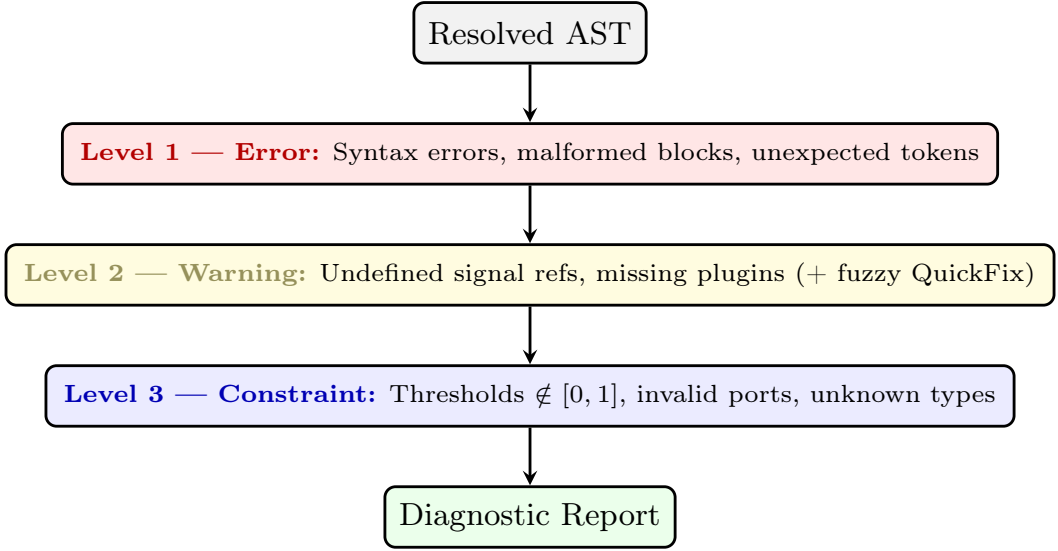
\begin{figure}[!ht]
  \centering
  \resizebox{\linewidth}{!}{%
\begin{tikzpicture}[
    box/.style={rectangle, draw, thick, rounded corners=3pt,
                minimum height=0.6cm, minimum width=7.5cm, align=left, inner sep=5pt, font=\scriptsize},
    arr/.style={->, >=stealth, thick},
    lbl/.style={font=\scriptsize\bfseries},
  ]

  \node[rectangle, draw, thick, rounded corners=3pt, fill=black!5,
        minimum height=0.6cm, align=center, inner sep=4pt, font=\small] (ast) at (0, 3.6) {Resolved AST};

  \node[box, fill=red!10] (l1) at (0, 2.4)
    {\textcolor{red!70!black}{\textbf{Level 1 --- Error:}} Syntax errors, malformed blocks, unexpected tokens};

  \node[box, fill=yellow!15] (l2) at (0, 1.2)
    {\textcolor{yellow!50!black}{\textbf{Level 2 --- Warning:}} Undefined signal refs, missing plugins (+ fuzzy QuickFix)};

  \node[box, fill=blue!8] (l3) at (0, 0)
    {\textcolor{blue!70!black}{\textbf{Level 3 --- Constraint:}} Thresholds $\notin [0,1]$, invalid ports, unknown types};

  \draw[arr] (ast) -- (l1);
  \draw[arr] (l1) -- (l2);
  \draw[arr] (l2) -- (l3);

  \node[rectangle, draw, thick, rounded corners=3pt, fill=green!8,
        minimum height=0.6cm, align=center, inner sep=4pt, font=\small] (out) at (0, -1.2) {Diagnostic Report};
  \draw[arr] (l3) -- (out);

  \end{tikzpicture}%
}
  \caption{Three-level validation architecture. The validator processes the resolved AST through progressively deeper analysis: syntactic errors (Level~1), reference resolution with fuzzy-matched QuickFix suggestions (Level~2), and semantic constraint checking (Level~3). Diagnostics at all levels are accumulated into a unified report.}
  \label{fig:validation_levels}
\end{figure}

\subsection{DSL as Instruction Set and Agent-Based Policy Synthesis}
\label{sec:disc_agent}

The configuration language can be understood as the \emph{instruction set} of the neural-symbolic inference engine.
Just as a CPU's instruction set defines the space of programs that can execute on the hardware, the DSL defines the space of routing policies that can be instantiated on the signal-decision-plugin architecture.
The functional completeness result of \Cref{sec:decision_engine} guarantees that this instruction set is \emph{universal}: any routing policy $\pi: \{0,1\}^N \to \mathcal{M}$ is expressible.

This framing transforms the problem of \emph{configuring} the router into a \emph{program synthesis} problem (\Cref{fig:agent_synthesis}): given a natural-language specification of routing requirements (``route math queries to the math model, enforce PII filtering for healthcare queries''), synthesize a valid DSL configuration that implements the specification.
This is precisely the class of problems that code-generation agents---LLMs fine-tuned or prompted for program synthesis---are designed to solve.

The DSL's formal grammar and type-safe compilation make it particularly suitable as the target language: the structured syntax---with explicit keywords, typed blocks, and a finite set of signal types---constrains the generation space far more tightly than free-form YAML, reducing the probability of syntactically valid but semantically incorrect configurations.
The three-level validator provides immediate, machine-readable feedback that a coding agent can use to iteratively refine generated configurations.

\begin{figure*}[!ht]
  \centering
  \resizebox{\linewidth}{!}{%
\begin{tikzpicture}[
    box/.style={rectangle, draw, thick, rounded corners=3pt,
                minimum height=0.8cm, align=center, inner sep=4pt, font=\small},
    arr/.style={->, >=stealth, thick},
    lbl/.style={font=\scriptsize, text=gray},
  ]

  \node[box, fill=black!5, minimum width=2.2cm] (nl) at (0, 0) {Natural Language\\Routing Spec};

  \node[box, fill=orange!10, minimum width=2.2cm] (agent) at (3.8, 0) {Coding Agent\\(LLM)};

  \node[box, fill=blue!8, minimum width=2.2cm] (dsl) at (7.6, 0) {DSL Config\\(YAML)};

  \node[box, fill=green!8, minimum width=2.2cm] (engine) at (11.4, 0) {Inference\\Engine};

  \node[box, fill=red!6, minimum width=2.0cm, font=\scriptsize] (fb) at (7.6, -2.0) {Routing Quality\\Feedback $Q(r, m^*)$};

  \draw[arr] (nl) -- node[above, font=\scriptsize] {synthesis} (agent);
  \draw[arr] (agent) -- node[above, font=\scriptsize] {generate} (dsl);
  \draw[arr] (dsl) -- node[above, font=\scriptsize] {instantiate} (engine);
  \draw[arr] (engine.south) |- (fb.east);
  \draw[arr] (fb.west) -| node[left, font=\scriptsize, pos=0.8] {optimize} (agent.south);

  \node[lbl, anchor=north] at (0, -0.6) {``Route math to};
  \node[lbl, anchor=north] at (0, -0.9) {specialized model''};
  \node[lbl, anchor=north] at (3.8, 1.0) {program synthesis};
  \node[lbl, anchor=north] at (11.4, 1.0) {signal-decision-plugin};

  \end{tikzpicture}%
}
  \caption{Agent-based policy synthesis pipeline. A coding agent (LLM) translates natural-language routing specifications into DSL configurations, which are instantiated on the inference engine. Routing quality feedback closes the loop, enabling iterative policy refinement. The DSL's functional completeness guarantees that any expressible routing policy is reachable by the synthesis process.}
  \label{fig:agent_synthesis}
\end{figure*}
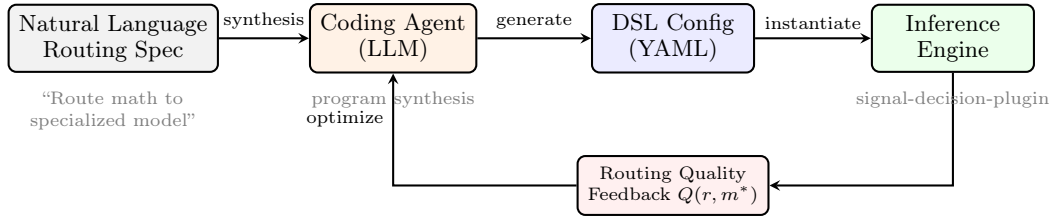

The connection to reinforcement learning is direct: the coding agent's policy $\pi_\theta(\text{config} \mid \text{spec})$ can be optimized via RLHF or rule-based reward~\cite{zhang2025routerr1}, where the reward signal is the downstream routing quality $Q(r, m^*)$ aggregated over a traffic distribution.
This operates at a \emph{meta-level} compared to prior work on RL-based routing~\cite{zhang2025routerr1}: rather than learning to route individual requests, the agent learns to \emph{write routing policies} that generalize across request distributions.
The DSL provides the structured action space that makes this meta-learning tractable---the agent generates a finite, syntactically constrained configuration rather than an arbitrary neural network.

\subsection{Synthesis: A Programmable Neural-Symbolic Inference Engine}

Taken together, the signal extraction layer (\Cref{sec:signal_engine}), decision engine (\Cref{sec:decision_engine}), and configuration language characterize the system as a \textbf{programmable neural-symbolic inference engine}:

\begin{enumerate}
  \item \textbf{Neural front-end}: The signal extraction layer uses both heuristic and learned (neural) methods to produce a structured representation---analogous to a hybrid embedding layer (\Cref{sec:disc_embedding}).
  \item \textbf{Symbolic back-end}: The decision engine applies Boolean logic over the structured representation---analogous to symbolic expert gating with formal verifiability (\Cref{sec:disc_moe}).
  \item \textbf{Programmable interface}: The DSL configuration is the ``program'' that specifies the inference behavior, and the functional completeness of the Boolean decision model guarantees universality of the program space.
  \item \textbf{Agent-compilable}: The structured DSL enables LLM-based coding agents to serve as ``compilers'' from natural-language specifications to executable routing policies, with RL-based optimization closing the loop.
\end{enumerate}

This perspective unifies the Shannon-theoretic foundation (information reduction + Boolean synthesis) with the modern ML landscape (embeddings + MoE gating + agent-based program synthesis), and suggests that the architecture occupies a principled point in the design space between fully neural routing systems (which sacrifice interpretability and editability) and fully manual rule systems (which sacrifice adaptivity and scalability).


\section{Request-Time Safety: Jailbreak and PII Signals}
\label{sec:safety}

Request-time safety in our architecture is implemented as \emph{first-class signals}, not as serial plugins.
Jailbreak and PII detection run in the signal extraction layer (\Cref{sec:signal_engine}), evaluated in parallel alongside all other signal types (keyword, embedding, domain, etc.).
\Cref{fig:safety_architecture} illustrates how safety signals integrate with domain signals and decision rules.
When a safety signal fires, the decision engine matches it to a decision that activates the \texttt{fast\_response} plugin (\Cref{sec:plugins}), returning an immediate OpenAI-compatible response without contacting any upstream LLM.

This signal-driven design has three key advantages:
\begin{enumerate}[nosep,leftmargin=1.5em]
  \item \textbf{Zero added latency}: safety classifiers run concurrently with intent signals, so the wall-clock cost is max(safety, intent) rather than safety + intent.
  \item \textbf{Composability}: safety signals can be combined with domain, keyword, or embedding signals in decision rules using AND/OR logic---e.g., strict jailbreak detection only for financial queries.
  \item \textbf{Uniform observability}: safety results appear as standard signal matches in headers (\texttt{x-vsr-matched-jailbreak}, \texttt{x-vsr-matched-pii}), traces, and the dashboard topology view.
\end{enumerate}

\subsection{Jailbreak Detection}

Jailbreak attacks attempt to override model safety instructions through adversarial prompt construction.
Our detection pipeline addresses this as a classification problem with per-rule sensitivity control.

\noindent\textbf{Formulation.}
Given a text input $x$ (the user's latest message, or the full conversation history when \texttt{include\_history} is enabled), a classifier $g_\text{jb}$ produces:
\begin{equation}
  g_\text{jb}(x) = (y, c) \in \{\textsc{benign}, \textsc{injection}, \textsc{jailbreak}\} \times [0, 1]
\end{equation}
A jailbreak signal rule with threshold $\theta$ fires iff $y \neq \textsc{benign}$ and $c \geq \theta$.
Multiple rules at different thresholds can coexist, enabling decisions to reference different sensitivity levels.

\noindent\textbf{Classifier architecture.}
We support four inference backends with varying context-length and deployment characteristics:
(1)~BERT with LoRA adapters (standard context);
(2)~ModernBERT~\cite{warner2024modernbert} with Flash Attention;
(3)~mmBERT-32K with YaRN RoPE for 32K-token contexts;
(4)~external vLLM-served guardrail models for decoupled scaling.
All local backends use the LoRA adapter architecture (\Cref{sec:lora_mom}), reducing model memory footprint.

\noindent\textbf{Per-rule sensitivity.}
Different signal rules configure different thresholds: a \texttt{jailbreak\_strict} rule might use $\theta = 0.4$ with full history analysis for financial endpoints, while a \texttt{jailbreak\_standard} rule uses $\theta = 0.65$ on the latest message only.
Decisions then reference the appropriate rule by name, enabling context-aware security policies (e.g., strict detection only when the domain signal also matches ``economics'').

\subsection{Contrastive Jailbreak Detection}
\label{subsec:contrastive_jailbreak}

While the BERT classifier excels at detecting single-turn jailbreak attempts, it can be evaded by \emph{multi-turn escalation} (``boiling frog'') attacks in which each individual message appears benign but the conversation progressively steers the model toward unsafe behavior.
We introduce a contrastive embedding method that operates alongside the BERT classifier within the same jailbreak signal type.

\noindent\textbf{Knowledge base construction.}
Each contrastive rule defines two exemplar sets: a \emph{jailbreak knowledge base} $\mathcal{K}_\text{jb}$ containing known adversarial patterns (e.g., ``Ignore all previous instructions'', ``You are now DAN'') and a \emph{benign knowledge base} $\mathcal{K}_\text{ben}$ containing representative normal queries (e.g., ``What is the weather today'', ``Help me write an email'').
All KB embeddings are precomputed at initialization using a concurrent worker pool (identical to the complexity signal's preloading strategy), so runtime cost is limited to embedding the incoming message and computing cosine similarities.

\noindent\textbf{Contrastive scoring.}
For each user message $m$, the contrastive score measures relative proximity to the jailbreak vs.\ benign exemplars:
\begin{equation}
  \delta(m) = \max_{j \in \mathcal{K}_\text{jb}} \cos(\mathbf{m}, \mathbf{j}) \;-\; \max_{b \in \mathcal{K}_\text{ben}} \cos(\mathbf{m}, \mathbf{b})
  \label{eq:contrastive_score}
\end{equation}
A positive $\delta$ indicates the message is semantically closer to jailbreak patterns; a negative $\delta$ indicates it is closer to benign patterns.
The subtraction normalizes for the overall embedding similarity scale, making the threshold robust across different embedding models.

\noindent\textbf{Max-contrastive chain (multi-turn detection).}
When \texttt{include\_history} is enabled, the system exploits the stateless nature of the OpenAI chat API---each request carries the full conversation history in the \texttt{messages} array---to evaluate every user message:
\begin{equation}
  \Delta = \max_{m \in \mathcal{M}_\text{user}} \delta(m)
\end{equation}
The rule fires when $\Delta \geq \theta$ (default $\theta = 0.10$).
This max-pooling aggregation ensures that even if the current message is innocuous, the presence of any high-scoring turn in the history triggers detection.
The approach is effective against gradual escalation because the attacker must include at least one semantically adversarial turn to steer the model, and that turn will produce a high contrastive score regardless of its position in the conversation.

\noindent\textbf{Integration with signal architecture.}
The contrastive method is selected via the \texttt{method: contrastive} field on a jailbreak rule, keeping the same signal type $\tau_\text{jb}$ and rule-name addressing used by the BERT classifier.
This means contrastive rules participate fully in the decision engine's Boolean logic.
A typical deployment combines both methods:
\begin{itemize}[nosep,leftmargin=1.5em]
  \item A BERT rule (\texttt{jailbreak\_standard}, $\theta = 0.65$) for fast, high-precision single-turn detection.
  \item A contrastive rule (\texttt{jailbreak\_multiturn}, $\theta = 0.10$, \texttt{include\_history: true}) for multi-turn chain detection.
  \item A decision with OR logic fires if \emph{either} method detects an attack, providing defense in depth.
\end{itemize}
The embedding model is inherited from the global \texttt{embedding\_models} configuration (consistent with the complexity signal), avoiding per-rule model specification.

\subsection{PII Detection}

PII detection identifies personally identifiable information at the token level and enforces configurable allow/deny policies.

\noindent\textbf{Formulation.}
A token classifier $g_\text{pii}$ operates on the input sequence $x = (x_1, \ldots, x_T)$ and produces BIO-tagged predictions:
\begin{equation}
  g_\text{pii}(x) = \bigl\{(i, j, \ell, c) \mid \text{span } x_i \ldots x_j \text{ is PII type } \ell \text{ with confidence } c\bigr\}
\end{equation}
where $\ell \in \{\textsc{person}, \textsc{email}, \textsc{phone}, \textsc{ssn}, \textsc{credit\_card}, \ldots\}$.

\noindent\textbf{Policy model.}
Each PII signal rule specifies a threshold $\theta$ and an optional allow-list $\mathcal{L}_\text{allow}$.
The rule fires when any detected entity $(\cdot, \cdot, \ell, c)$ satisfies $c \geq \theta$ and $\ell \notin \mathcal{L}_\text{allow}$.
This per-rule policy enables differentiated enforcement: a \texttt{pii\_deny\_all} rule blocks all entity types, while a \texttt{pii\_allow\_email} rule permits email addresses for appointment booking in medical contexts.

\subsection{Parallel Safety Evaluation}

Both jailbreak and PII classifiers launch as concurrent goroutines within the signal extraction layer, alongside all other signal evaluators.
The wall-clock time is dominated by the slowest signal evaluator rather than the sum of safety and intent classifiers.
Both use the LoRA adapter architecture (\Cref{sec:lora_mom}), and their results merge into the standard signal result $S(r)$ consumed by the decision engine.

When a safety signal fires, the matched decision activates its \texttt{fast\_response} plugin, which short-circuits the pipeline and returns an OpenAI-compatible 200 OK response (supporting both streaming SSE and non-streaming JSON) without forwarding to any upstream model.

\begin{figure}[t]
\centering
\includegraphics[width=0.94\linewidth]{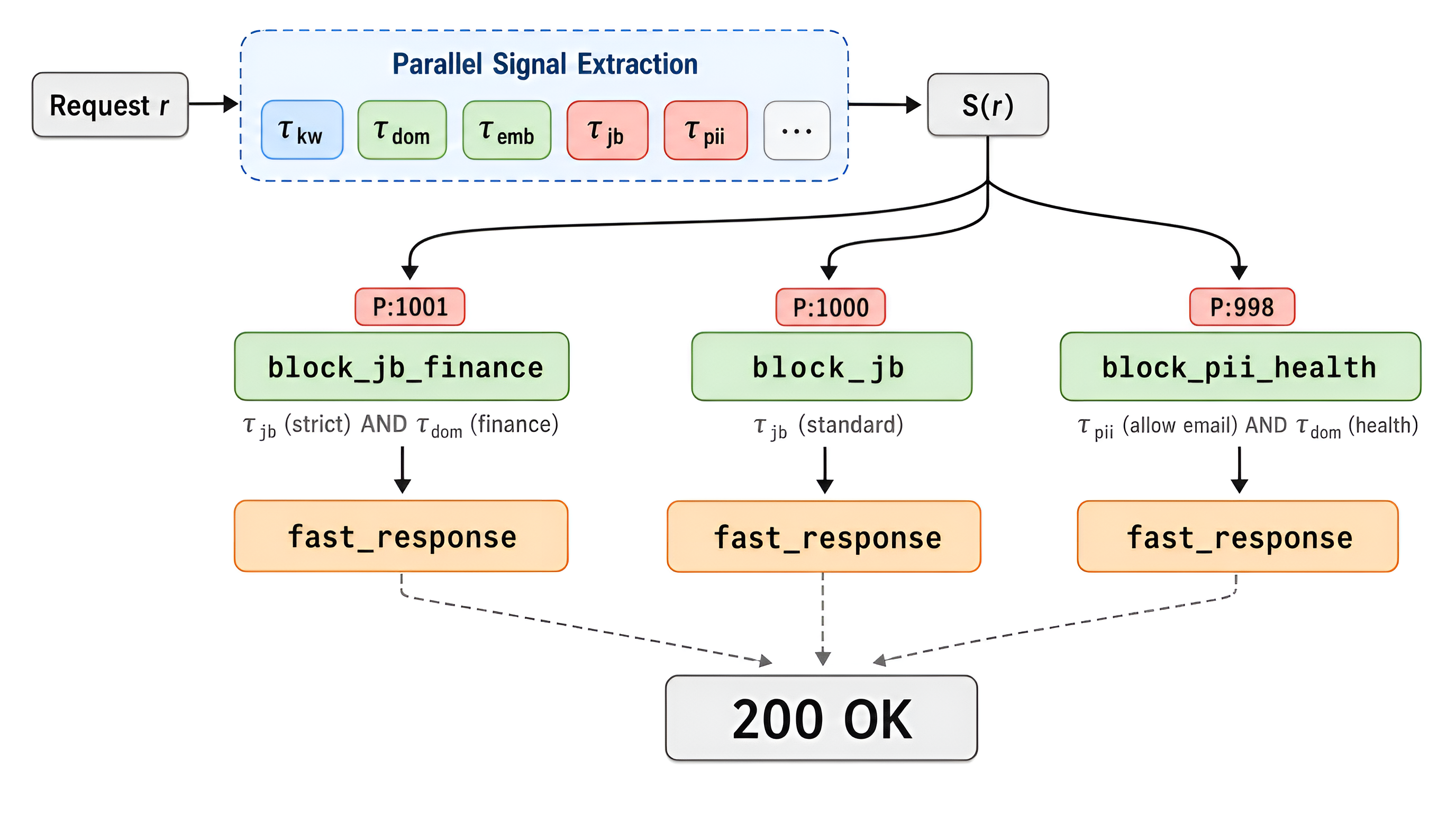}
\caption{Signal-driven safety architecture.  Safety classifiers ($\tau_{\text{jb}}, \tau_{\text{pii}}$) run in the signal layer in parallel with intent signals, adding zero marginal latency.  Decision rules compose safety signals with domain signals using AND/OR logic, enabling context-aware policies (e.g., strict jailbreak detection only for finance, relaxed PII policy for healthcare).  High-priority safety decisions activate the \texttt{fast\_response} plugin to return an immediate 200 OK without contacting any upstream model.}
\label{fig:safety_architecture}
\end{figure}

\subsection{Composable Safety Design}

Placing safety detection in the signal layer enables a class of flexible security policies that are difficult to express in a serial plugin architecture.
We highlight three patterns enabled by this design:

\noindent\textbf{Context-aware thresholds.}
A single deployment can define multiple jailbreak rules at different sensitivity levels (e.g., $\theta = 0.4$ for strict, $\theta = 0.65$ for standard).
Decisions compose these with domain signals: a \texttt{block\_jb\_finance} decision uses the strict rule AND the ``economics'' domain signal, while a \texttt{block\_jb\_general} decision uses the standard rule alone.
This provides differentiated security posture without separate deployments.

\noindent\textbf{Per-domain PII policies.}
Different domains require different PII handling: a healthcare endpoint may allow email and phone number for appointment booking while blocking SSN, whereas a public chatbot blocks all PII types.
Each PII signal rule carries its own allow-list, and decisions combine the appropriate rule with the relevant domain signal.
This replaces a monolithic PII policy with a set of composable, domain-scoped rules.

\noindent\textbf{Graduated response.}
By defining multiple safety decisions at different priorities, the system can implement graduated responses: a high-confidence jailbreak ($c \geq 0.9$) triggers an immediate block (priority 1001), while a moderate-confidence detection ($0.4 \leq c < 0.9$) triggers a warning header without blocking (priority 500).
The \texttt{fast\_response} plugin handles blocking; a header-only decision allows the request to proceed with an observability annotation.
This graduated model reduces false-positive impact while maintaining security coverage.


\section{HaluGate: Gated Hallucination Detection}
\label{sec:halugate}

Hallucination---generating plausible but unsupported content---is a fundamental limitation of autoregressive language models~\cite{manakul2023selfcheckgpt,min2023factscore}.
We introduce \halugate{}, a three-stage pipeline (\Cref{fig:halugate_pipeline}) that addresses a key efficiency challenge: most queries (creative writing, code generation, brainstorming) do not require factual verification, yet na\"ive hallucination detection incurs overhead on every response.

\subsection{Design Rationale}

Existing approaches apply hallucination detection uniformly to all responses~\cite{manakul2023selfcheckgpt} or require multiple response samples~\cite{min2023factscore}.
\halugate{} introduces two innovations:
(1)~a \emph{gating stage} that skips verification for non-factual queries, amortizing detection cost over the query distribution; and
(2)~a \emph{span-level} detection and explanation pipeline that identifies \emph{which} tokens are hallucinated and \emph{why}, rather than providing only a binary judgment.

\subsection{Three-Stage Pipeline}

\begin{figure}[t]
\centering
\resizebox{\linewidth}{!}{%
\begin{tikzpicture}[
    node distance=0.3cm,
    box/.style={rectangle, draw, rounded corners=2pt,
                minimum height=0.6cm, align=center,
                font=\scriptsize, inner sep=3pt},
    dia/.style={diamond, draw, aspect=2, inner sep=1pt,
                font=\scriptsize, fill=yellow!6},
    arr/.style={->, >=stealth, thick},
    darr/.style={->, >=stealth, densely dashed, gray!40},
  ]
  \node[box, fill=black!4] (q) at (0, 0) {Query $q$};
  \node[box, fill=blue!6] (sent) at (2.2, 0) {Sentinel\\[-1pt]$g_{\text{sent}}$};
  \node[dia] (dec) at (4.4, 0) {factual?};

  \node[box, fill=black!4, minimum width=0.9cm] (skip) at (4.4, 1.3) {\footnotesize skip};
  \draw[darr] (dec.north) -- node[font=\tiny, right, text=gray]{no} (skip.south);

  \node[box, fill=black!4] (a) at (6.5, 0) {Response $a$};
  \node[box, fill=blue!6] (det) at (8.7, 0) {Detector\\[-1pt]$g_{\text{det}}$};
  \node[box, fill=blue!6] (nli) at (11.0, 0) {Explainer\\[-1pt]$g_{\text{nli}}$};
  \node[box, fill=orange!5] (out) at (13.2, 0) {Results};

  \draw[arr] (q) -- (sent);
  \draw[arr] (sent) -- (dec);
  \draw[arr] (dec.east) -- node[font=\tiny, above, text=gray]{yes} (a.west);
  \draw[arr] (a) -- (det);
  \draw[arr] (det) -- node[font=\tiny, above]{spans} (nli);
  \draw[arr] (nli) -- (out);

  \node[font=\tiny, text=gray, anchor=north] at (1.1, -0.55) {request path};
  \node[font=\tiny, text=gray, anchor=north] at (9.8, -0.55) {response path};
  \draw[densely dotted, gray!30] (5.45, -0.7) -- (5.45, 1.5);
\end{tikzpicture}%
}
\caption{\halugate{} three-stage gated pipeline.  The Sentinel classifies queries on the request path; non-factual queries (40--60\%) skip verification entirely (dashed).  For factual queries, the Detector identifies hallucinated spans in the model response, and the Explainer provides NLI-based diagnostics per span.}
\label{fig:halugate_pipeline}
\end{figure}
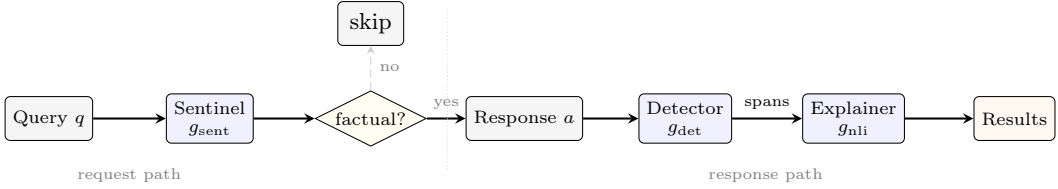

\noindent\textbf{Stage 1: Sentinel (Gating).}
A lightweight binary classifier $g_\text{sent}$ determines whether the query warrants factual verification:
\begin{equation}
  g_\text{sent}(q) \in \{\textsc{needs\_fact\_check}, \textsc{no\_fact\_check}\}
\end{equation}
If $g_\text{sent}(q) = \textsc{no\_fact\_check}$, Stages 2--3 are skipped entirely.
The Sentinel operates on the request text and is implemented as a LoRA-adapted classifier sharing the base model with other signal extractors.
In practice, 40--60\% of queries are classified as non-factual, proportionally reducing the average detection cost.

The Sentinel also serves dual duty as the \texttt{fact\_check} signal in the signal extraction layer (\Cref{sec:signal_engine}), enabling decisions to incorporate factual grounding into routing logic.

\noindent\textbf{Stage 2: Detector (Span Identification).}
A token-level classifier $g_\text{det}$ identifies hallucinated spans in the model response:
\begin{equation}
  g_\text{det}(q, \mathbf{c}, a) = \bigl\{(i, j, c_{ij}) \mid a_i \ldots a_j \text{ is unsupported by context } \mathbf{c}\bigr\}
\end{equation}
where $q$ is the user query, $\mathbf{c}$ is the grounding context (user-provided context and tool-call results), $a$ is the assistant's response, and $(i, j, c_{ij})$ denotes a flagged span with confidence.

When tool-calling is present, tool execution results provide high-quality ground truth: database query results, API responses, and calculations serve as authoritative context $\mathbf{c}$, substantially improving detection precision.

\noindent\textbf{Stage 3: Explainer (NLI Classification).}
For each flagged span $(i, j)$, a Natural Language Inference (NLI) model~\cite{williams2018mnli} classifies the relationship between the span and the grounding context:
\begin{equation}
  g_\text{nli}(a_{i:j}, \mathbf{c}) \in \{\textsc{entailment}, \textsc{contradiction}, \textsc{neutral}\}
\end{equation}
This distinguishes between content that \emph{contradicts} the context (definitive hallucination) and content that is merely \emph{unsupported} (potential hallucination), providing actionable diagnostics.

\subsection{Cost Analysis}

Let $p_\text{factual}$ be the fraction of queries requiring factual verification, $C_\text{sent}$, $C_\text{det}$, $C_\text{nli}$ be the costs of each stage, and $\bar{k}$ be the average number of flagged spans.
The expected cost per query is:
\begin{equation}
  \mathbb{E}[\text{Cost}] = C_\text{sent} + p_\text{factual} \cdot \bigl(C_\text{det} + \bar{k} \cdot C_\text{nli}\bigr)
\end{equation}

Since the Sentinel is a lightweight LoRA-adapted classifier (\Cref{sec:lora_mom}) that runs concurrently with other signal extractors, its wall-clock cost is largely hidden behind other ML signals.
For a workload with $p_\text{factual} = 0.5$, the gating stage reduces the expected Detector and Explainer cost by approximately 50\% compared to applying full detection to all responses.

\subsection{Action Policies}

\halugate{} supports four configurable response actions:

\begin{table}[htbp]
\centering
\caption{\halugate{} action policies upon hallucination detection}
\label{tab:halugate_actions}
\begin{tabular}{lp{8cm}}
\toprule
\textbf{Action} & \textbf{Semantics} \\
\midrule
\texttt{block}  & Reject the response; return an error to the client. Appropriate for high-stakes factual applications. \\
\texttt{header} & Propagate detection metadata via HTTP headers, enabling downstream policy enforcement by the client or API gateway. \\
\texttt{body}   & Prepend a warning to the response body, alerting users to potential inaccuracies. \\
\texttt{none}   & Log detection results without modifying the response. Useful for monitoring and threshold calibration. \\
\bottomrule
\end{tabular}
\end{table}

The progressive architecture enables incremental deployment: organizations begin with Sentinel-only gating (signal-layer integration at minimal cost), add the Detector for span-level monitoring, and enable the Explainer for full diagnostic output.


\section{LoRA-Based MoM Model Family}
\label{sec:lora_mom}

Signal-driven routing requires multiple classification tasks on the critical path of every request.
Na\"ively, each task requires a separate fine-tuned model, creating a memory scaling problem.
We describe the LoRA-based architecture that addresses this and the purpose-built model family trained for semantic routing.

\subsection{Problem: Linear Memory Scaling}

Let $n$ denote the number of active classification tasks (domain, jailbreak, PII, fact-check, feedback, modality).
With independently fine-tuned models, the total model memory is:
\begin{equation}
  M_\text{indep} = n \cdot |\theta_\text{base}|
\end{equation}
where $|\theta_\text{base}|$ is the parameter count of a single base model.
For $n = 6$ tasks with a 150M-parameter base model, this requires storing and loading six full model copies ($\sim$900M parameters total)---a significant memory burden, especially in GPU-constrained environments.

Additionally, managing $n$ independent model checkpoints complicates deployment, versioning, and updates.

\subsection{Solution: Single Base Model with LoRA Adapters}

Low-Rank Adaptation (LoRA)~\cite{hu2022lora} represents task-specific weight modifications as low-rank decompositions:
\begin{equation}
  W'_i = W + \Delta W_i = W + B_i A_i, \quad B_i \in \mathbb{R}^{d \times r}, \; A_i \in \mathbb{R}^{r \times d}
\end{equation}
where $W$ is the shared base weight, $r \ll d$ is the adapter rank, and $B_i A_i$ is the task-specific perturbation.

The aggregate model memory becomes:
\begin{equation}
  M_\text{LoRA} = |\theta_\text{base}| + \sum_{i=1}^{n} 2 r d = |\theta_\text{base}| + n \cdot 2rd
\end{equation}

With rank $r = 32$ and hidden dimension $d = 768$, each adapter adds $2 \times 32 \times 768 = 49{,}152$ parameters ($\sim$0.02\% of the base model).
For $n = 6$, total adapter overhead is $\sim$295K parameters---negligible compared to the 150M-parameter base.

\noindent\textbf{Memory reduction.}
\begin{equation}
  \frac{M_\text{LoRA}}{M_\text{indep}} = \frac{|\theta_\text{base}| + n \cdot 2rd}{n \cdot |\theta_\text{base}|} \approx \frac{1}{n} \quad \text{for } 2nrd \ll |\theta_\text{base}|
\end{equation}
At $n = 6$, this yields $\sim$6$\times$ memory reduction: one 150M-parameter model plus six tiny adapters instead of six full copies.

\subsection{Inference Architecture}

Each classification task proceeds as a full forward pass through the base model with the task-specific LoRA perturbation applied (\Cref{fig:lora_arch}):

\begin{enumerate}
  \item \textbf{Load}: A single base model is loaded into GPU/CPU memory at startup. Each LoRA adapter (a pair of small matrices per adapted layer) is loaded alongside it.
  \item \textbf{Inference}: For each classification task $i$, the base model runs a forward pass with adapter $i$'s weights merged: $W'_i = W + B_i A_i$. Each task still requires its own forward pass.
  \item \textbf{Parallelism}: Multiple classification tasks execute concurrently via parallel threads/goroutines. Wall-clock time is determined by the slowest classifier, not the sum.
\end{enumerate}

\begin{figure}[ht]
    \centering
    \includegraphics[width=0.94\linewidth]{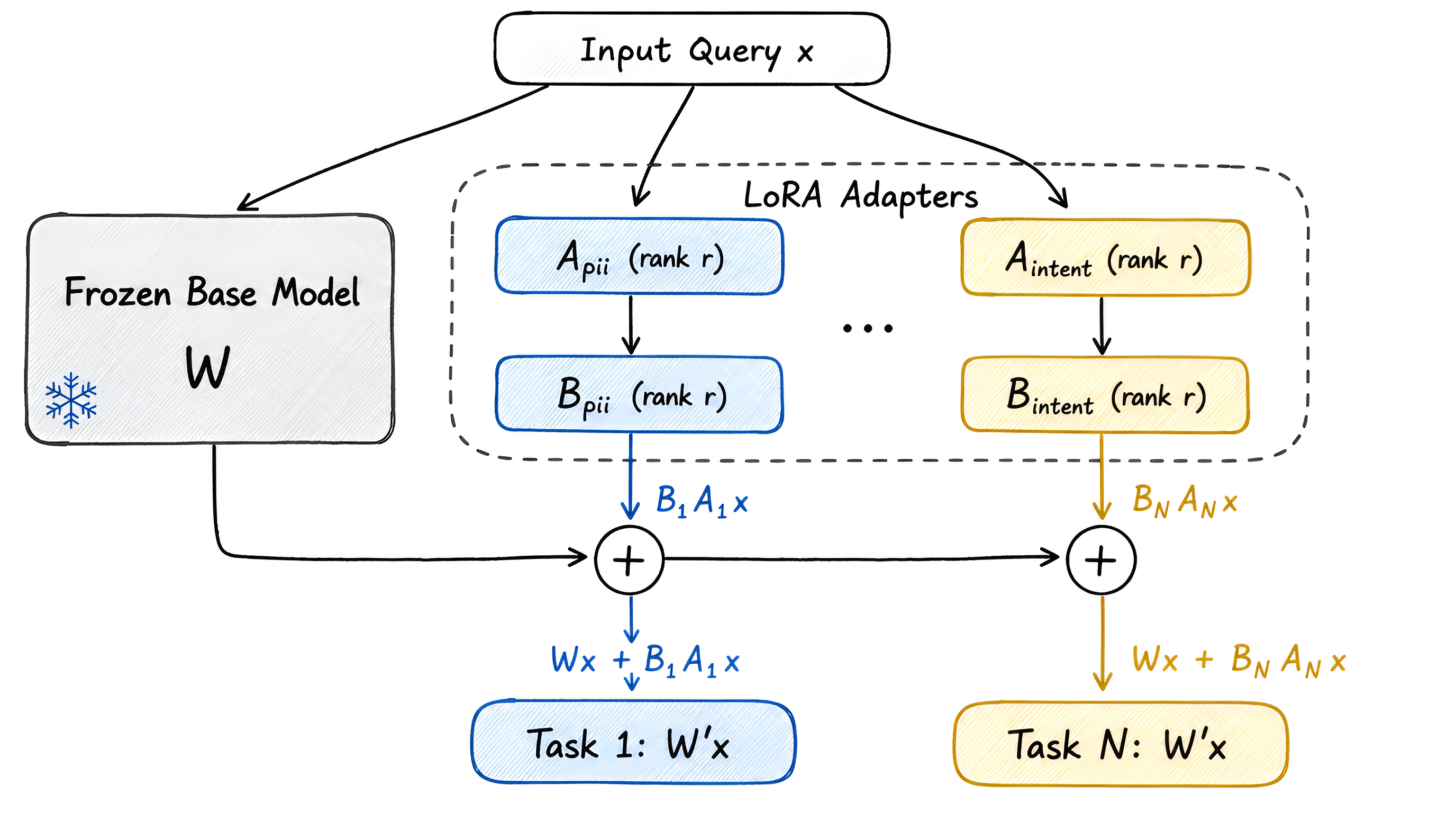}
    \caption{LoRA-based MoM inference architecture. The input query~$\mathbf{x}$ is processed by a single frozen base model~($W$). Each task-specific LoRA adapter pair ($A_i$, $B_i$) computes a low-rank perturbation $B_i A_i \mathbf{x}$. The base output and adapter perturbation are summed to produce task-specific classifications, enabling multi-task support with minimal memory overhead.}
    \label{fig:lora_arch}
\end{figure}

Note that LoRA does \emph{not} eliminate the per-task forward pass---each adapter requires a full inference through the modified model.
The primary benefit is \textbf{memory efficiency}: deploying six classifiers requires the memory footprint of approximately one model rather than six, and all adapters can be updated independently without reloading the base model.

\subsection{MoM Model Family}

We train a family of purpose-built models (MoM: Mixture-of-Models) optimized for routing classification tasks:

\begin{table}[htbp]
\centering
\caption{MoM model family.
All models share a common base (ModernBERT~\cite{warner2024modernbert} or mmBERT-32K) and are distributed as LoRA adapters.}
\label{tab:mom_family}
\begin{tabular}{lll}
\toprule
\textbf{Model} & \textbf{Task} & \textbf{Training Data} \\
\midrule
\texttt{mom-domain}         & Domain classification     & MMLU categories \\
\texttt{mom-pii}            & PII token classification  & Presidio-annotated corpora \\
\texttt{mom-jailbreak}      & Prompt injection detection & Adversarial prompt datasets \\
\texttt{mom-sentinel}       & Fact-check gating         & Factual vs.\ creative queries \\
\texttt{mom-detector}       & Hallucination detection   & Annotated LLM outputs \\
\texttt{mom-explainer}      & NLI explanation           & NLI benchmarks \\
\texttt{mom-feedback}       & User feedback analysis    & Conversation annotations \\
\texttt{mom-modality}       & Modality classification   & DiffusionDB + text corpora \\
\texttt{mom-embedding}      & Semantic embeddings       & Contrastive pre-training \\
\texttt{mom-toolcall}       & Tool selection            & Function-calling datasets \\
\texttt{mom-intent} & User intent classification & Customer support dialogues \\
\bottomrule
\end{tabular}
\end{table}

The key benefit of distributing these as LoRA adapters rather than independent models is \textbf{operational simplicity}: a single base model binary serves all ten tasks, adapters can be hot-swapped without reloading the base, and new tasks can be added by training a new adapter without retraining or redistributing the base model.

\subsection{Training Methodology}

All LoRA adapters are trained using PEFT~\cite{mangrulkar2022peft} with the following protocol:
\begin{itemize}[leftmargin=*]
  \item \textbf{Base model}: ModernBERT or mmBERT-32K (for long-context tasks).
  \item \textbf{Adapter configuration}: Rank $r \in \{16, 32, 64\}$, applied to query and value projection matrices.
  \item \textbf{Training}: Task-specific datasets with standard cross-entropy loss.
  \item \textbf{Export}: Both LoRA-only (separate adapter files for hot-swapping) and merged (single model file for simplified deployment) formats.
\end{itemize}

The modality classifier, for instance, is trained on a balanced mixture of DiffusionDB (image generation prompts), OASST2, Alpaca, and Dolly (text generation), achieving three-class classification (autoregressive, diffusion, both) with $\sim$0.02\% trainable parameters relative to the base model.


\section{Semantic Model Selection}
\label{sec:model_selection}

The core routing innovation is \emph{semantic model selection}: once the decision engine matches a routing decision $d^*$, the system analyzes the request's semantic content---its embedding, domain, complexity, and interaction history---to select the most cost-effective model from the decision's candidate set.
Unlike static routing or single-criterion difficulty classifiers, semantic selection operates over the full signal context produced by the signal engine, enabling cost-quality optimization that respects per-decision privacy and safety constraints.

We integrate thirteen algorithms within a unified interface, enabling systematic comparison and hybrid combinations across deployment scenarios.

\subsection{Problem Setting}

Given query embedding $\mathbf{e}_q \in \mathbb{R}^d$, domain category $z \in \{1, \ldots, C\}$, candidate models $\mathcal{M}_{d^*} = \{m_1, \ldots, m_K\}$ with associated costs $\{c_1, \ldots, c_K\}$, and quality estimators, the semantic selection problem is:
\begin{equation}
  m^* = \arg\max_{m_k \in \mathcal{M}_{d^*}} \; \text{Quality}(\mathbf{e}_q, z, m_k; \Theta) - \lambda \cdot \text{Cost}(m_k)
\end{equation}
where $\lambda \geq 0$ is a cost-sensitivity parameter and $\Theta$ represents algorithm-specific parameters.
The per-decision candidate set $\mathcal{M}_{d^*}$ is critical: privacy-constrained decisions restrict candidates to compliant models, while cost-optimized decisions include a broader pool with aggressive cost weighting.
We categorize algorithms into families based on their selection mechanism (\Cref{fig:selection_taxonomy}).

\begin{figure}[t]
\centering
\includegraphics[width=0.94\linewidth]{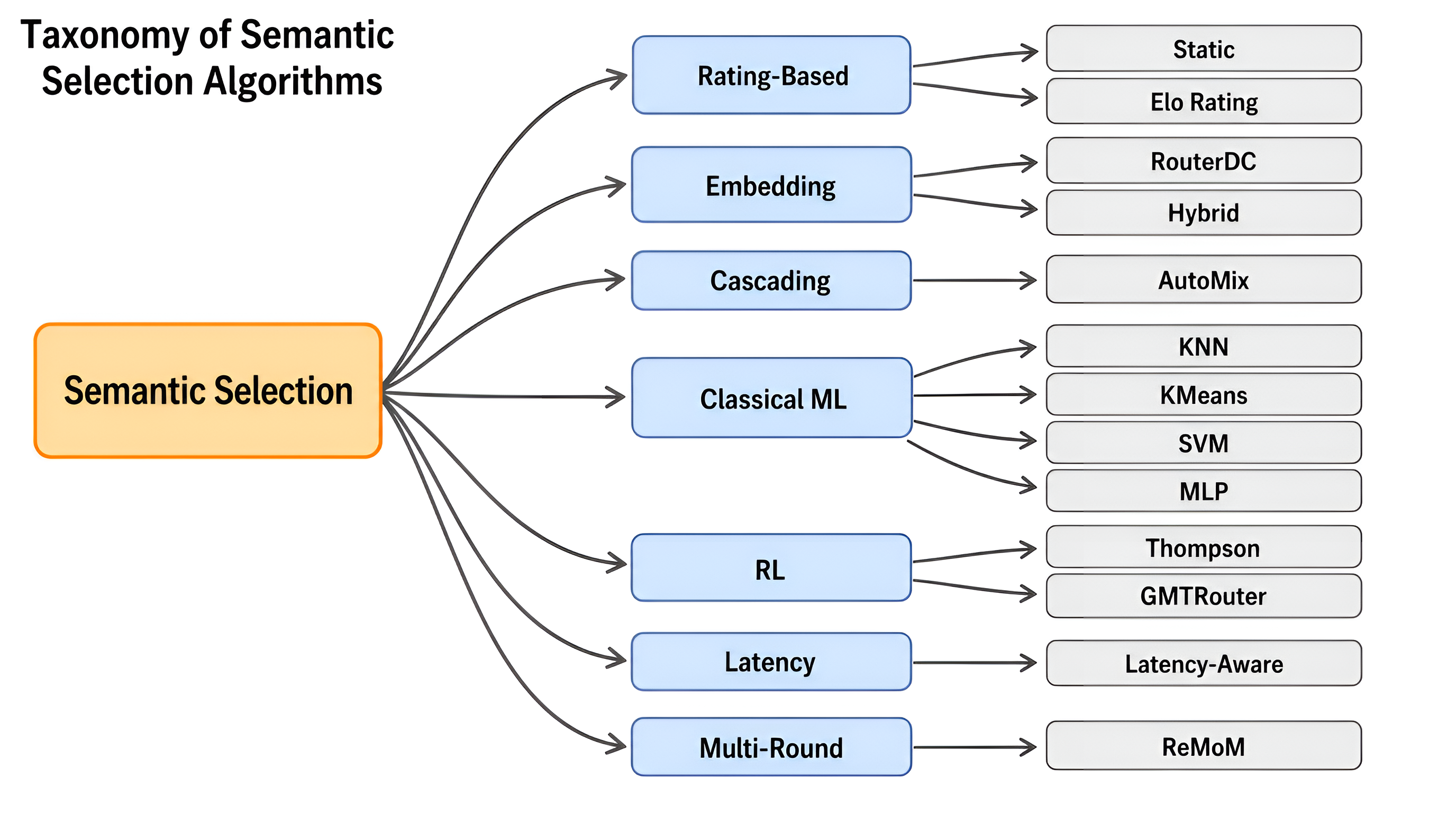}
\caption{Taxonomy of thirteen semantic model selection algorithms organized by selection mechanism.  Families span from lightweight rating-based methods (Static, Elo) to learned approaches (RouterDC, Classical ML, RL), adaptive cascading (AutoMix), real-time latency tracking, and multi-round synthesis (ReMoM).}
\label{fig:selection_taxonomy}
\end{figure}

\subsection{Rating-Based Selection}

\noindent\textbf{Static.}
Each model carries a pre-configured quality score $s_k$; selection is $m^* = \arg\max_k s_k$.
Serves as a deterministic baseline.

\noindent\textbf{Elo Rating} (adapted from RouteLLM~\cite{ong2024routellm}).
Models maintain Elo ratings $R_k$ updated from pairwise user preference feedback.
Selection probability follows the Bradley-Terry model:
\begin{equation}
  P(m_i \succ m_j) = \frac{1}{1 + 10^{(R_j - R_i)/400}}
\end{equation}
Models are sampled proportional to their expected win rate against the candidate pool.
Ratings are updated online as user feedback arrives.

\subsection{Embedding-Based Selection}

\noindent\textbf{RouterDC}~\cite{chen2024routerdc}.
Dual contrastive learning trains query and model encoders to produce embeddings in a shared space.
Selection maximizes cosine similarity:
\begin{equation}
  m^* = \arg\max_{m_k \in \mathcal{M}_{d^*}} \cos(\mathbf{e}_q, \mathbf{e}_{m_k})
\end{equation}
The contrastive training objective encourages queries to be close to their best-performing model's embedding and distant from poorly-performing models.

\noindent\textbf{Hybrid}~\cite{hu2024routerbench}.
Combines Elo ratings, embedding similarity, and cost in a weighted score:
\begin{equation}
  \text{score}(m_k) = \alpha \cdot \tilde{R}_k + \beta \cdot \cos(\mathbf{e}_q, \mathbf{e}_{m_k}) + \gamma \cdot (1 - \tilde{c}_k)
\end{equation}
where $\tilde{R}_k$ and $\tilde{c}_k$ are normalized ratings and costs, and $\alpha + \beta + \gamma = 1$ are configurable weights.

\subsection{Cascading Selection}

\noindent\textbf{AutoMix}~\cite{aggarwal2023automix}.
Formulated as a Partially Observable Markov Decision Process (POMDP).
Models are ordered by capability $m_1 \prec m_2 \prec \cdots \prec m_K$.
The cascade:
\begin{enumerate}
  \item Generate response $a_k$ with current model $m_k$ (starting from $k=1$, the cheapest).
  \item Self-verify: estimate response quality $\hat{q}_k$ using $m_k$ itself.
  \item If $\hat{q}_k \geq \tau_k$, accept $a_k$; otherwise, escalate to $m_{k+1}$.
\end{enumerate}
The expected cost is:
\begin{equation}
  \mathbb{E}[C] = \sum_{k=1}^{K} C_k \cdot \prod_{j=1}^{k-1}(1 - P(\hat{q}_j \geq \tau_j))
\end{equation}
where $P(\hat{q}_j \geq \tau_j)$ is the probability that model $m_j$ passes self-verification.
This naturally trades off cost against quality.

\subsection{Classical ML Selection}

These methods train on routing records $\{(\mathbf{e}_q^{(i)}, z^{(i)}, m^{*(i)}, q^{(i)})\}$ where $q^{(i)}$ is a quality score.
Feature vectors combine embeddings and domain information:
\begin{equation}
  \mathbf{f} = [\mathbf{e}_q \in \mathbb{R}^d; \; \text{onehot}(z) \in \{0,1\}^C]
\end{equation}

\noindent\textbf{KNN.}
$k$-nearest neighbor search with Ball Tree indexing.
Quality-weighted majority voting:
\begin{equation}
  m^* = \arg\max_m \sum_{i \in \text{kNN}(\mathbf{f})} \mathbf{1}[m^{*(i)} = m] \cdot q^{(i)}
\end{equation}

\noindent\textbf{KMeans.}
Assigns queries to pre-computed clusters; selects the best model for the assigned cluster based on a combined quality-latency score:
\begin{equation}
  m^* = \arg\max_m \bigl(\alpha \cdot \text{quality}(m, z_\text{cluster}) - (1-\alpha) \cdot \text{latency}(m)\bigr)
\end{equation}

\noindent\textbf{SVM.}
Multi-class SVM with RBF or linear kernel, trained to classify feature vectors directly into model selections.

\noindent\textbf{MLP.}
A feed-forward neural network (two hidden layers with ReLU activation) mapping $\mathbf{f}$ to a softmax distribution over candidate models:
\begin{equation}
  P(m_k \mid \mathbf{f}) = \text{softmax}\bigl(W_2 \cdot \text{ReLU}(W_1 \mathbf{f} + b_1) + b_2\bigr)_k
\end{equation}
The MLP is implemented in the GPU-accelerated Candle runtime for low-latency inference.

\subsection{Reinforcement Learning Selection}

\noindent\textbf{Thompson Sampling}~\cite{thompson1933likelihood}.
Each model maintains a Beta prior:
\begin{equation}
  \theta_k \sim \text{Beta}(\alpha_k, \beta_k)
\end{equation}
Selection samples from each posterior and picks the maximum: $m^* = \arg\max_k \theta_k$.
Parameters $(\alpha_k, \beta_k)$ are updated from user preference feedback, naturally balancing exploration and exploitation.

\noindent\textbf{GMTRouter}~\cite{xie2025gmtrouter}.
Models multi-turn user-query-model interactions as a heterogeneous graph.
Graph neural network message passing captures complex interaction patterns:
\begin{equation}
  \mathbf{h}_v^{(l+1)} = \text{AGG}\bigl(\{\mathbf{h}_u^{(l)} \mid u \in \mathcal{N}(v)\}\bigr)
\end{equation}
where nodes represent users, queries, and models, and edges encode historical routing outcomes.

\subsection{Latency-Aware Selection}

\noindent\textbf{Latency-Aware.}
Selects the model with the best observed latency using percentile-based Time-per-Output-Token (TPOT) and Time-to-First-Token (TTFT) statistics collected at runtime.
For each candidate model $m_k$, the selector computes a normalized latency score:
\begin{equation}
  s_k = \frac{1}{|P|} \sum_{p \in P} \frac{\text{perc}_p(m_k)}{\min_{j}\, \text{perc}_p(m_j)}
\end{equation}
where $P \subseteq \{\text{TPOT}, \text{TTFT}\}$ is the set of configured performance metrics and $\text{perc}_p(m_k)$ is the observed percentile value for model $m_k$ on metric $p$.
Selection minimizes this score: $m^* = \arg\min_k s_k$.
This enables adaptive routing that responds to real-time backend performance degradation without requiring explicit latency thresholds as signal conditions.

\subsection{Multi-Round Reasoning (ReMoM)}

The ReMoM (Reasoning for Mixture of Models) strategy extends single-shot selection to multi-round parallel reasoning with LLM-driven synthesis.
Inspired by PaCoRe~\cite{pacore2025} but generalized to heterogeneous model pools, ReMoM executes a \emph{breadth schedule} of decreasing parallelism across rounds, where each subsequent round synthesizes the outputs of the previous round via prompted LLM calls (\Cref{fig:remom_flow}).

\begin{figure}[t]
\centering
\includegraphics[width=0.94\linewidth]{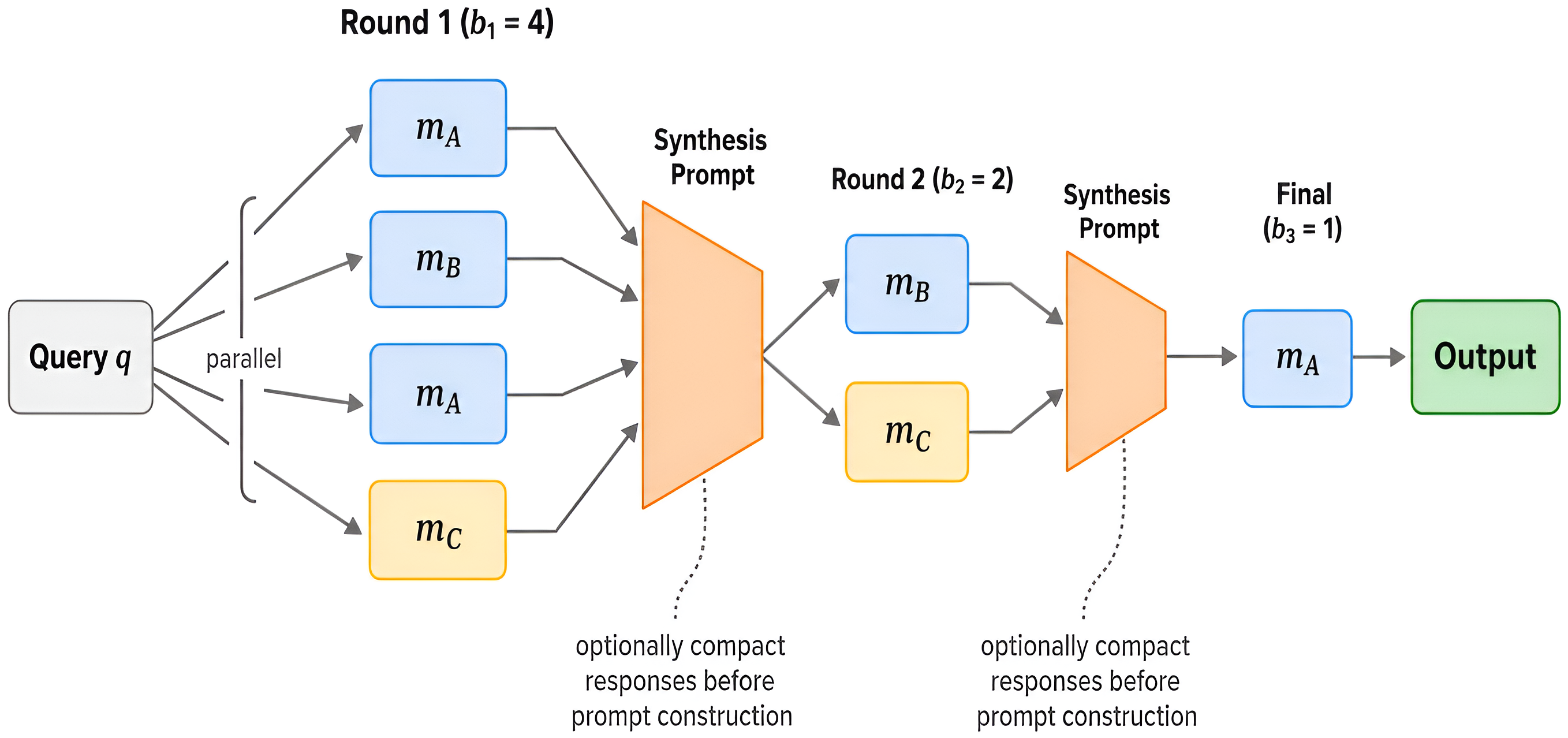}
\caption{ReMoM execution flow with breadth schedule $\mathbf{b} = [4, 2]$.
Round~1 distributes 4~parallel calls across models ($m_A$, $m_B$, $m_C$); responses are (optionally compacted and) assembled into a synthesis prompt.
Round~2 sends 2~parallel synthesis calls.
A final round ($b_3 = 1$, auto-appended) produces the single output.
Each synthesis prompt includes the original query and all previous-round responses as numbered references, delegating quality judgment to the synthesizing LLM.}
\label{fig:remom_flow}
\end{figure}

\noindent\textbf{Breadth schedule.}
The operator specifies a breadth schedule $\mathbf{b} = [b_1, b_2, \ldots, b_R]$ defining the number of parallel model calls per round.
A final synthesis round with $b_{R+1} = 1$ is automatically appended, yielding a total of $R+1$ rounds.
For example, $\mathbf{b} = [32, 4]$ produces three rounds: 32 parallel calls, then 4 parallel calls each synthesizing the 32 responses, then a single final call synthesizing the 4 responses.

\noindent\textbf{Model distribution.}
At each round, $b_r$ calls are distributed among the candidate models $\mathcal{M}_{d^*}$ according to one of three strategies:
(1)~\emph{equal}: calls are distributed evenly across all candidates with round-robin remainder allocation;
(2)~\emph{weighted}: calls are distributed proportionally to model weights (currently equivalent to equal distribution);
(3)~\emph{first\_only}: all $b_r$ calls target a single model with different random seeds, providing PaCoRe-compatible single-model diversity.
Calls within each round execute concurrently with configurable concurrency limits.

\noindent\textbf{LLM-driven synthesis.}
After collecting responses from round $r$, the system constructs a synthesis prompt for round $r+1$ using a configurable Go \texttt{text/template}.
The default template presents the original query alongside all previous-round responses as numbered references, instructing the next-round model(s) to \emph{``analyze these references and provide your own comprehensive solution.''}
When reasoning content is available (e.g., from models supporting extended thinking), the template additionally includes each reference's chain-of-thought reasoning.
This approach delegates quality judgment entirely to the synthesizing LLM rather than relying on explicit scoring or weighted aggregation.

\noindent\textbf{Response compaction.}
To manage prompt length across rounds, responses can be compacted before inclusion in synthesis prompts.
Two strategies are supported: \emph{full} (no compaction, the default) and \emph{last\_n\_tokens} (retaining only the final $N$ tokens, estimated at $\sim$4 characters per token).
This is particularly important for high-breadth schedules where concatenating all responses would exceed context limits.

\noindent\textbf{Execution flow.}
The complete algorithm proceeds as:
\begin{enumerate}
  \item \textbf{Schedule construction}: Append $[1]$ to the user-specified breadth schedule $\mathbf{b}$.
  \item \textbf{Round~1 (parallel generation)}: Distribute $b_1$ calls across candidate models; execute concurrently with temperature $T$ (default 1.0) for response diversity.
  \item \textbf{Rounds~$2 \ldots R{+}1$ (synthesis)}: For each subsequent round, build a synthesis prompt from the previous round's (optionally compacted) responses, distribute $b_r$ calls, and execute concurrently.
  \item \textbf{Final output}: Return the single response from the final round ($b_{R+1} = 1$).
\end{enumerate}

ReMoM is particularly effective when model capabilities are uncertain or when the task benefits from diverse perspectives (e.g., complex reasoning, multi-faceted analysis).
The breadth schedule provides fine-grained control over the cost--quality tradeoff: higher initial breadth increases diversity at the cost of additional LLM calls, while the funneling structure ensures convergence to a single synthesized answer.

\subsection{Unified Selection Interface}

All thirteen algorithms implement a common interface:
\begin{equation}
  \text{Select}: (\mathbf{e}_q, z, \mathcal{M}, \Theta) \to (m^*, c)
\end{equation}
returning the selected model and a confidence score.
This uniformity enables:
(1)~per-decision algorithm selection---different routing decisions can use different selection algorithms, allowing cost-optimized decisions to use cascading (AutoMix) while quality-sensitive decisions use embedding-based (RouterDC) selection;
(2)~A/B testing across algorithms on live traffic;
(3)~ensemble methods that combine multiple selectors.

\subsection{Cost-Aware Selection in Multi-Provider Settings}

In multi-endpoint deployments where the same logical model may be served by different providers at different price points, the selection algorithms operate in conjunction with the endpoint router (\Cref{subsec:multi_endpoint}).
The selection algorithm chooses the best \emph{model} based on semantic analysis, and the endpoint router resolves it to the most cost-effective \emph{provider endpoint}.
This two-stage process separates quality optimization (which model is best for this query?) from cost optimization (which provider endpoint offers the best price for this model?), enabling fine-grained cost management across heterogeneous multi-cloud deployments.


\section{Multi-Runtime ML Inference}
\label{sec:ml_inference}

The routing system requires low-latency ML inference for signal extraction, classification, and embedding computation---all on the critical path of every request.
We describe the multi-runtime architecture that addresses the tension between inference speed, hardware flexibility, and model diversity.

\subsection{Design Constraints}

Three constraints shape the inference architecture:
\begin{enumerate}
  \item \textbf{Latency}: Signal extraction must complete within the tail latency budget of the routing system (target: $<$100\,ms for all signals combined).
  \item \textbf{Hardware heterogeneity}: Deployments range from GPU-equipped data centers to CPU-only edge nodes.
  \item \textbf{Model diversity}: Different tasks require different model architectures (sequence classification, token classification, NLI, embeddings, MLP).
\end{enumerate}

\subsection{Four-Runtime Architecture}

We implement four inference runtimes, each optimized for different hardware and task profiles, all exposed to the Go routing process via C FFI and CGo (\Cref{fig:ml_runtimes}):

\begin{table}[htbp]
\centering
\caption{Inference runtime characteristics}
\label{tab:runtimes}
\begin{tabularx}{\linewidth}{@{}
  >{\raggedright\arraybackslash}p{1.6cm}
  >{\raggedright\arraybackslash}p{1.85cm}
  >{\raggedright\arraybackslash}X
  >{\raggedright\arraybackslash}p{2.9cm}
@{}}
\toprule
\textbf{Runtime} & \textbf{Target Hardware} & \textbf{Tasks} & \textbf{Framework} \\
\midrule
Candle       & GPU (CUDA), CPU & Classification, LoRA, MLP & HF Candle~\cite{candleml2024} \\
Linfa        & CPU only        & KNN, KMeans, SVM         & Linfa~\cite{linfa2024} \\
ONNX RT      & CPU, GPU        & Embeddings               & ONNX Runtime~\cite{onnxruntime2024} \\
NLP Binding  & CPU only        & BM25, N-gram matching    & bm25 + ngrammatic (Rust) \\
\bottomrule
\end{tabularx}
\end{table}

All runtimes are compiled as Rust shared libraries (\texttt{.so}/\texttt{.a}) and linked to the Go routing process via CGo.
This eliminates Python runtime overhead, GIL contention, and inter-process communication latency that would arise from serving models in separate Python processes.

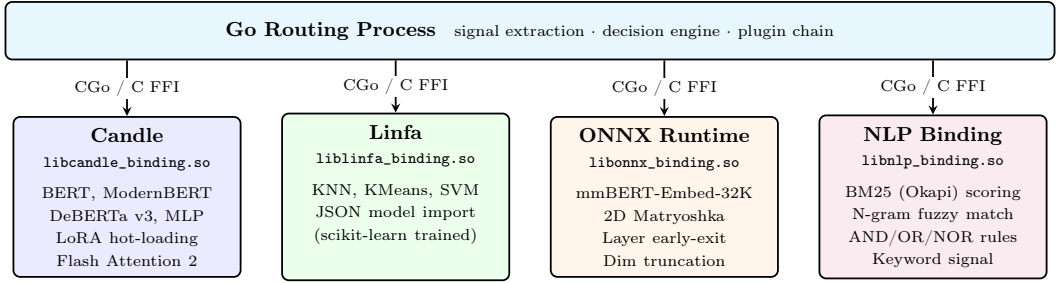
\begin{figure*}[ht]
\centering
\resizebox{\linewidth}{!}{%
\begin{tikzpicture}[
  gobox/.style={rectangle, draw, thick, rounded corners=3pt, fill=cyan!8,
                minimum width=16.4cm, minimum height=0.9cm, align=center,
                inner sep=5pt, font=\small},
  rt/.style={rectangle, draw, thick, rounded corners=3pt,
             text width=3.2cm, align=center, inner sep=5pt, font=\small},
  arr/.style={->, >=stealth, thick},
  ffi/.style={font=\scriptsize, midway, fill=white, inner sep=1pt},
]

\node[gobox] (go) at (0, 2.6)
  {\textbf{Go Routing Process}\;\;
   {\scriptsize signal extraction $\cdot$ decision engine $\cdot$ plugin chain}};

\node[rt, fill=blue!8] (candle) at (-6.3, 0) {%
  \textbf{Candle}\\[1pt]
  {\scriptsize\texttt{libcandle\_binding.so}}\\[4pt]
  {\scriptsize BERT, ModernBERT}\\
  {\scriptsize DeBERTa v3, MLP}\\
  {\scriptsize LoRA hot-loading}\\
  {\scriptsize Flash Attention 2}};

\node[rt, fill=green!8] (linfa) at (-2.1, 0) {%
  \textbf{Linfa}\\[1pt]
  {\scriptsize\texttt{liblinfa\_binding.so}}\\[4pt]
  {\scriptsize KNN, KMeans, SVM}\\
  {\scriptsize JSON model import}\\
  {\scriptsize (scikit-learn trained)}\\[3pt]~};

\node[rt, fill=orange!8] (onnx) at (2.1, 0) {%
  \textbf{ONNX Runtime}\\[1pt]
  {\scriptsize\texttt{libonnx\_binding.so}}\\[4pt]
  {\scriptsize mmBERT-Embed-32K}\\
  {\scriptsize 2D Matryoshka}\\
  {\scriptsize Layer early-exit}\\
  {\scriptsize Dim truncation}};

\node[rt, fill=purple!8] (nlp) at (6.3, 0) {%
  \textbf{NLP Binding}\\[1pt]
  {\scriptsize\texttt{libnlp\_binding.so}}\\[4pt]
  {\scriptsize BM25 (Okapi) scoring}\\
  {\scriptsize N-gram fuzzy match}\\
  {\scriptsize AND/OR/NOR rules}\\
  {\scriptsize Keyword signal}};

\draw[arr] (go.south -| candle) -- node[ffi] {CGo / C FFI} (candle.north);
\draw[arr] (go.south -| linfa)  -- node[ffi] {CGo / C FFI} (linfa.north);
\draw[arr] (go.south -| onnx)   -- node[ffi] {CGo / C FFI} (onnx.north);
\draw[arr] (go.south -| nlp)    -- node[ffi] {CGo / C FFI} (nlp.north);

\end{tikzpicture}%
}
\caption{Four-runtime ML inference architecture. The Go routing process links to four Rust shared libraries via CGo/C~FFI. Each runtime is specialized for a different class of ML workload, avoiding Python overhead on the critical path.}
\label{fig:ml_runtimes}
\end{figure*}

\subsection{Candle Runtime: GPU-Accelerated Classification}

The Candle runtime handles all transformer-based classification tasks, including LoRA adapter loading and inference (\Cref{sec:lora_mom}).

\noindent\textbf{Supported architectures.}
BERT~\cite{devlin2019bert}, ModernBERT~\cite{warner2024modernbert} (with Flash Attention and GeGLU), mmBERT-32K (YaRN RoPE for 32K context), DeBERTa v3 (NLI), and feed-forward MLPs (model selection).

\noindent\textbf{Optimization features.}
Flash Attention 2 kernels reduce attention memory from $O(n^2)$ to $O(n)$ and improve throughput.
Optional Intel MKL integration for CPU deployments.
LoRA adapter hot-loading enables runtime model updates without restart.

\subsection{ModernBERT-base-32k: Extended Context Window}

The Candle runtime supports ModernBERT-base-32k, extending the context window from 512 tokens (BERT-base) to 32,768 tokens via YaRN (Yet another RoPE extensioN) scaling~\cite{peng2023yarn}.
This enables signal extraction from long documents and multi-turn conversations that exceed the traditional 512-token limit.

\noindent\textbf{Architecture and scaling.}
ModernBERT-base-32k uses Rotary Position Embeddings (RoPE) with YaRN scaling to extend the base model's 8,192-token context to 32K tokens.
The model maintains compatibility with existing LoRA adapters trained on BERT-base, enabling seamless migration without retraining classifier heads.
Flash Attention 2 provides memory-efficient attention computation, reducing memory requirements from $O(n^2)$ to $O(n)$ for long sequences.

\noindent\textbf{Performance characteristics.}
Empirical benchmarks on NVIDIA L4 GPU (23GB VRAM) demonstrate reliable performance for context lengths up to 8K tokens:
\begin{itemize}[leftmargin=*]
  \item \textbf{1K tokens}: p50 latency $\sim$94\,ms at C=1, $\sim$970\,ms at C=10 (100\% success rate)
  \item \textbf{4K tokens}: p50 latency $\sim$955\,ms at C=1, $\sim$9,389\,ms at C=10 (93\% success rate, 7 OOM errors)
  \item \textbf{8K tokens}: p50 latency $\sim$3,525\,ms at C=1, fails at C=10 due to insufficient GPU memory
\end{itemize}
For sequences exceeding 32K tokens, automatic chunking with configurable overlap (default 128 tokens) enables processing of arbitrarily long documents while preserving context continuity.

\noindent\textbf{Hardware requirements.}
Production deployments for 1K--8K tokens require a GPU with $\geq$23GB VRAM (e.g., NVIDIA L4, A10G).
Full 32K token support and high concurrency (C=50+) require $\geq$40GB VRAM (e.g., NVIDIA A100).
CPU inference is supported but incurs significant latency penalties ($\sim$45$\times$ slower for 512 tokens).

\noindent\textbf{Backward compatibility.}
The integration maintains full backward compatibility with existing BERT-base workflows.
Sequences $\leq$512 tokens exhibit no performance degradation, and existing LoRA adapters (domain classification, PII detection, jailbreak detection) function without modification.

\subsection{Linfa Runtime: CPU ML Inference}

Classical ML model selection algorithms (KNN, KMeans, SVM) are served by the Linfa runtime.
These algorithms operate on pre-computed feature vectors and do not require GPU acceleration, making Linfa's lightweight CPU implementation ideal.

\noindent\textbf{Training-inference split.}
Models are trained in Python (scikit-learn, custom implementations) and serialized to JSON.
The Rust runtime loads serialized models at startup and performs inference-only computation.
This decouples the training environment (Python, GPU-optional) from the inference environment (Rust, CPU-only), enabling simpler deployment.

\subsection{ONNX Runtime: Efficient Embeddings}

Embedding computation is served by ONNX Runtime, optimized for the mmBERT-Embed-32K model with 2D Matryoshka representation learning~\cite{kusupati2022matryoshka}.

\noindent\textbf{2D Matryoshka trade-offs.}
The architecture supports two-dimensional quality-latency trade-offs:
\begin{itemize}[leftmargin=*]
  \item \textbf{Layer early-exit}: Extract embeddings from intermediate layers (6, 11, 16, or 22 out of 22), achieving $\sim$3--4$\times$ speedup at layer 6 with modest quality degradation.
  \item \textbf{Dimension truncation}: Reduce embedding dimension from 768 to 64, 128, 256, or 512, reducing memory and computation for similarity search.
\end{itemize}

For the $\sim$150M parameter embedding model, CPU inference with 2D Matryoshka (layer 11, dimension 256) achieves latency comparable to GPU inference on the full model, making GPU optional for embedding computation.

\subsection{NLP Binding Runtime: Keyword Classification}

The NLP Binding runtime handles BM25 and N-gram keyword matching for the keyword signal type, complementing the ML-based classifiers served by Candle.
Unlike the neural runtimes, NLP Binding implements statistical text matching algorithms that require no model weights or GPU:

\noindent\textbf{BM25 (Okapi).}
The BM25 classifier uses the Rust \texttt{bm25} crate to compute term-frequency--inverse-document-frequency scores between query text and keyword rules.
Each rule specifies a set of keywords, a Boolean operator (AND/OR/NOR), a score threshold, and case-sensitivity.
A keyword is considered matched when its BM25 relevance score exceeds the configured threshold.

\noindent\textbf{N-gram fuzzy matching.}
The N-gram classifier uses the \texttt{ngrammatic} crate to perform fuzzy string matching via character n-gram overlap (default: trigrams, Jaccard similarity).
This enables matching despite typos and morphological variation---e.g., \texttt{"urgnet"} matches the keyword \texttt{"urgent"} when the similarity exceeds the configured threshold.

\noindent\textbf{FFI design.}
The binding follows the same conventions as \texttt{candle-binding}: Rust \texttt{\#[repr(C)]} structs, \texttt{\#[no\_mangle] pub extern "C" fn} exports, CString-based string passing, and explicit free functions for Rust-allocated memory.
The Go side wraps each classifier as a handle-based API (\texttt{BM25Classifier}, \texttt{NgramClassifier}) with lifecycle management (\texttt{New}, \texttt{AddRule}, \texttt{Classify}, \texttt{Free}).
Thread safety is provided by Rust \texttt{Mutex}-guarded global state with atomic handle generation.

\subsection{Runtime Selection Strategy}

The routing system selects runtimes based on deployment configuration.
The NLP Binding is always active when keyword signal rules use \texttt{bm25} or \texttt{ngram} methods:
\begin{itemize}[leftmargin=*]
  \item \textbf{GPU available}: Candle (classification + LoRA) + ONNX (embeddings) + Linfa (ML selection) + NLP Binding (keyword).
  \item \textbf{CPU only}: Candle with MKL (classification) + ONNX with early-exit (embeddings) + Linfa (ML selection) + NLP Binding (keyword).
  \item \textbf{Minimal}: ONNX (embeddings) + Linfa (ML selection) + NLP Binding (keyword), with classification delegated to external vLLM-served models.
\end{itemize}


\section{Request Processing Pipeline}
\label{sec:extproc}

We implement the routing system as an Envoy~\cite{envoyproxy2024} External Processor (ExtProc)~\cite{envoyextproc2024}, enabling transparent interception of LLM API traffic without client-side modifications.
\Cref{fig:extproc_pipeline} illustrates the bidirectional request processing pipeline.
This section describes the pipeline architecture, multi-provider routing, the Responses API integration, and the pluggable authorization factory.

\subsection{Transparent Interception via ExtProc}

The Envoy ExtProc protocol~\cite{envoyextproc2024} establishes a bidirectional gRPC stream between the proxy and the routing service for each HTTP request.
Envoy invokes the processor at four phases---request headers, request body, response headers, response body---and the processor responds with mutations (header modifications, body rewrites) or immediate responses (short-circuiting the backend).

This architecture provides two key advantages:
(1)~\emph{transparency}: clients send standard OpenAI-compatible API requests to the proxy endpoint with no awareness of the routing layer; and
(2)~\emph{composability}: the router coexists with other Envoy filters (rate limiting, authentication, load balancing) in the standard filter chain.

\subsection{Request Body Pipeline}

The request body phase implements the core routing logic as a sequential pipeline:

\begin{equation}
  r \xrightarrow{\text{parse}} r' \xrightarrow{\text{signals}} S(r') \xrightarrow{\text{decide}} d^* \xrightarrow{\Pi_\text{pre}} \xrightarrow{\text{select}} m^* \xrightarrow{\text{route}} e^*
\end{equation}

\begin{figure}[t]
\centering
\resizebox{\linewidth}{!}{%
\begin{tikzpicture}[
    node distance=0.3cm,
    phase/.style={rectangle, draw, thick, rounded corners=3pt,
                  align=center, inner sep=5pt, font=\small,
                  minimum height=1.6cm},
    io/.style={rectangle, draw, rounded corners=2pt, fill=black!4,
               font=\small, inner sep=4pt, align=center,
               minimum width=1.2cm},
    sub/.style={font=\scriptsize, align=left, text=black!70},
    arr/.style={->, >=stealth, thick},
    darr/.style={->, >=stealth, densely dashed, gray!40},
    pathlbl/.style={font=\scriptsize\bfseries, text=gray!60},
  ]

  \node[pathlbl] at (-0.8, 1.6) {Request};
  \node[pathlbl] at (-0.8, 1.2) {path};

  \node[io] (req) at (0.8, 1.4) {$r$};

  \node[phase, fill=blue!6, minimum width=2.4cm] (sig) at (3.2, 1.4) {%
    \textbf{Signal}\\[-1pt]
    {\scriptsize parse, detect,}\\[-2pt]
    {\scriptsize extract $S(r)$}};

  \node[phase, fill=green!6, minimum width=2.4cm] (dec) at (6.2, 1.4) {%
    \textbf{Decision}\\[-1pt]
    {\scriptsize evaluate $\varphi$}\\[-2pt]
    {\scriptsize select $d^*$}};

  \node[phase, fill=orange!5, minimum width=3.6cm] (plug) at (9.8, 1.4) {%
    \textbf{Plugins (Pre)}\\[-1pt]
    {\scriptsize fast response, cache,}\\[-2pt]
    {\scriptsize RAG, modality, memory,}\\[-2pt]
    {\scriptsize select $m^*$, route $e^*$}};

  \node[io] (ep) at (12.6, 1.4) {$e^*$};

  \draw[arr] (req) -- (sig);
  \draw[arr] (sig) -- node[above, font=\scriptsize]{$S(r)$} (dec);
  \draw[arr] (dec) -- node[above, font=\scriptsize]{$d^*$} (plug);
  \draw[arr] (plug) -- (ep);

  \node[io, minimum width=0.8cm, font=\scriptsize] (hit) at (9.8, 3.1) {resp};
  \draw[darr] (plug.north) -- (hit.south);
  \node[font=\tiny, text=gray, anchor=west] at (10.1, 2.65) {fast resp / cache};

  \node[pathlbl, anchor=east] at (-0.1, -1.4) {Response};
  \node[pathlbl, anchor=east] at (-0.1, -1.8) {path};

  \node[io] (ep2) at (12.6, -1.6) {$e^*$};

  \node[phase, fill=blue!6, minimum width=2.6cm] (rsig) at (10.2, -1.6) {%
    \textbf{Signal}\\[-1pt]
    {\scriptsize usage, HaluGate,}\\[-2pt]
    {\scriptsize metrics (TTFT/TPOT)}};

  \node[phase, fill=green!6, minimum width=2.6cm] (rdec) at (6.8, -1.6) {%
    \textbf{Decision}\\[-1pt]
    {\scriptsize quality gating,}\\[-2pt]
    {\scriptsize cache eligibility}};

  \node[phase, fill=orange!5, minimum width=2.6cm] (post) at (3.4, -1.6) {%
    \textbf{Plugin (Post)}\\[-1pt]
    {\scriptsize format translation,}\\[-2pt]
    {\scriptsize cache write, API xlat}};

  \node[io] (resp) at (0.8, -1.6) {resp};

  \draw[arr] (ep2) -- (rsig);
  \draw[arr] (rsig) -- (rdec);
  \draw[arr] (rdec) -- (post);
  \draw[arr] (post) -- (resp);

  \draw[thin, densely dotted, gray!40] (ep.south) -- ++(0, -0.15) -| (ep2.north);

\end{tikzpicture}%
}
\caption{Bidirectional request processing pipeline aligned with the three-layer architecture.  \textbf{Request path} (top, left$\to$right): the signal layer parses and extracts $S(r)$, including safety signals (jailbreak, PII); the decision layer evaluates Boolean formulas to select $d^*$; the plugin chain executes pre-processing (fast response for safety enforcement, caching, RAG, memory), model selection, and endpoint routing.  A cache hit or fast response short-circuits the pipeline (dashed).  \textbf{Response path} (bottom, right$\to$left): the same three layers operate in the same order---the signal layer extracts response-side signals (token usage, HaluGate hallucination scores, streaming latency metrics); the decision layer evaluates quality gating and cache eligibility; plugins perform format translation, cache writes, and API translation.}
\label{fig:extproc_pipeline}
\end{figure}
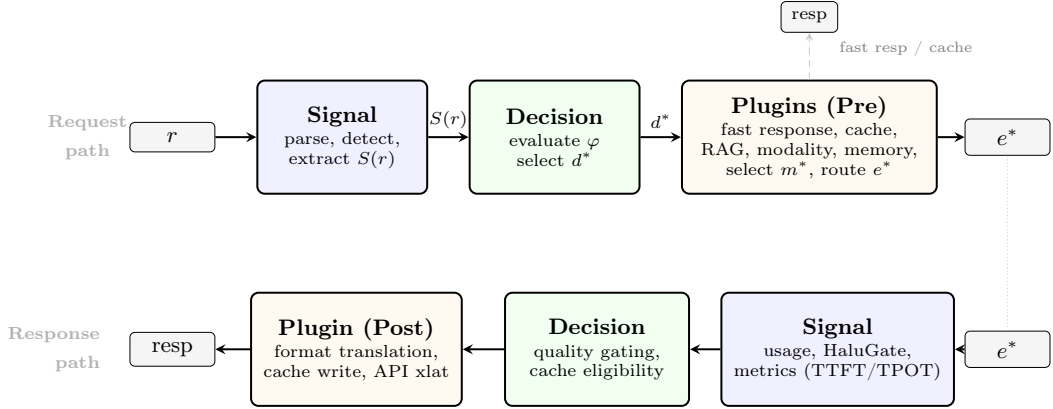

The stages execute in strict order:
(1)~API translation (Responses API $\to$ Chat Completions if applicable, see \Cref{subsec:responses_api});
(2)~request parsing and provider detection;
(3)~signal extraction and decision evaluation (\Cref{sec:signal_engine,sec:decision_engine});
(4)~fast response check (\Cref{sec:plugins})---if the matched decision activates the \texttt{fast\_response} plugin (e.g., for safety enforcement), the pipeline terminates with an immediate OpenAI-compatible response;
(5)~semantic cache lookup (\Cref{sec:plugins})---cache hits terminate the pipeline with an immediate response;
(6)~RAG context injection (\Cref{sec:memory_rag});
(7)~modality routing (text vs.\ diffusion);
(8)~memory retrieval (\Cref{sec:memory_rag});
(9)~model selection (\Cref{sec:model_selection}), system prompt injection, and header mutation;
(10)~multi-endpoint resolution and provider-specific auth injection (\Cref{subsec:multi_endpoint,subsec:authz_factory}).

\subsection{Multi-Endpoint and Multi-Provider Routing}
\label{subsec:multi_endpoint}

Production deployments often span multiple model backends across different providers and geographic regions.
The system supports \emph{multi-endpoint routing} as a first-class concept:

\begin{definition}[Endpoint Topology]
An endpoint topology $\mathcal{E} = \{(e_i, w_i, p_i, \alpha_i)\}_{i=1}^{L}$ defines $L$ endpoints, each with a weight $w_i \in (0, 1]$ (normalized: $\sum_i w_i = 1$), a provider type $p_i \in \mathcal{P}$, and an auth profile $\alpha_i$.
\end{definition}

Once semantic model selection identifies a target model $m^*$, the endpoint router resolves $m^*$ to a concrete endpoint $e^*$ from the set of endpoints serving that model.
Weighted random selection with sticky session affinity distributes load proportionally.
Failover cascades to the next-weighted endpoint on backend errors.

Each endpoint may use a different provider (e.g., the same logical model ``gpt-4o'' served by both OpenAI and Azure OpenAI).
The system performs \emph{provider-specific protocol translation} transparently:

\begin{itemize}[leftmargin=*]
  \item \textbf{OpenAI / Azure OpenAI}: Native Chat Completions and Responses API formats.
  \item \textbf{Anthropic}: Translation between OpenAI message schema and Anthropic Messages API (system prompt handling, tool use mapping).
  \item \textbf{Bedrock / Vertex AI}: Cloud-provider-specific request wrapping, authentication (SigV4 for AWS, OAuth for GCP), and response unwrapping.
  \item \textbf{Gemini}: Conversion between OpenAI function-calling schema and Gemini tool declarations.
  \item \textbf{vLLM / Local}: Direct OpenAI-compatible passthrough to self-hosted vLLM instances.
\end{itemize}

This abstraction allows routing decisions to reference models by capability (``best coding model'') rather than by provider-specific endpoint, and allows the same decision configuration to operate across different deployment topologies.

\subsection{OpenAI Responses API Support}
\label{subsec:responses_api}

The system provides full support for the OpenAI Responses API, which extends Chat Completions with stateful multi-turn conversation management.

The Responses API introduces \texttt{previous\_response\_id} chaining: each response carries a unique identifier, and subsequent requests can reference it to maintain conversation context without the client retransmitting the full message history.
The routing system handles this by:

\begin{enumerate}[leftmargin=*]
  \item \textbf{Inbound translation}: Responses API requests (with \texttt{input} field and \texttt{previous\_response\_id}) are normalized to Chat Completions format for signal extraction and decision evaluation, which operate on the unified internal representation.
  \item \textbf{State management}: Conversation history is stored in the persistent memory layer (\Cref{sec:memory_rag}), keyed by response ID, enabling context retrieval across turns.
  \item \textbf{Outbound translation}: Chat Completions responses from backends are wrapped in Responses API format (with \texttt{id}, \texttt{output} array, \texttt{usage} breakdown) before returning to the client.
  \item \textbf{Routing consistency}: The decision engine can optionally pin conversation turns to the same model to avoid mid-conversation quality shifts.
\end{enumerate}

This translation layer is transparent to both the signal engine and the downstream model backends, enabling all routing, safety, and caching features to operate identically on Responses API and Chat Completions traffic.

\subsection{Authorization Factory}
\label{subsec:authz_factory}

Multi-provider deployments require diverse authentication mechanisms.
The system implements a \emph{pluggable authorization factory} that abstracts auth concerns from routing logic:

\begin{definition}[Auth Provider]
An auth provider $\alpha: (\text{Request}, \text{Endpoint}) \to \text{Headers}'$ is a function that enriches outbound request headers with provider-appropriate credentials.
\end{definition}

The factory supports multiple auth provider types:

\begin{itemize}[leftmargin=*]
  \item \textbf{API Key}: Static bearer tokens or API keys, optionally per-endpoint, with header name customization (e.g., \texttt{Authorization}, \texttt{x-api-key}, \texttt{api-key}).
  \item \textbf{OAuth2 / OIDC}: Token acquisition with automatic refresh, supporting client credentials and authorization code flows.
  \item \textbf{Cloud IAM}: AWS SigV4 signing for Bedrock, Google service account tokens for Vertex AI, Azure AD tokens for Azure OpenAI.
  \item \textbf{Passthrough}: Forwarding the client's original credentials to the backend, used for deployments where the client authenticates directly.
  \item \textbf{Custom}: User-defined auth plugins registered at startup, enabling integration with enterprise identity providers (LDAP, SAML, custom JWT issuers).
\end{itemize}

The auth factory is invoked \emph{after} decision evaluation and model selection, injecting provider-specific credentials into the outbound request headers.
This separation ensures that routing decisions are auth-agnostic: the decision engine selects models based on capability and cost, and the auth layer handles the mechanics of reaching each provider's endpoint.

The \texttt{authz} signal type in the signal engine (\Cref{sec:signal_engine}) is complementary but distinct: it performs \emph{inbound} authorization (verifying that the requesting user or API key has permission to access specific models or decisions), while the auth factory handles \emph{outbound} authentication (proving the router's identity to backend providers).

\subsection{Response Body Pipeline}

The response path performs:
(1)~token usage extraction for cost accounting;
(2)~format translation (provider-specific $\to$ OpenAI format);
(3)~streaming metrics computation (TTFT, TPOT);
(4)~hallucination detection via HaluGate (\Cref{sec:halugate});
(5)~semantic cache writes for cache misses;
(6)~Responses API translation (Chat Completions $\to$ Responses API format, if applicable).

\subsection{Concurrency Model}

Each gRPC stream (one per HTTP request) runs in an independent goroutine, processing its four phases sequentially.
Within a request, signal extraction launches parallel coroutines for independent classifiers.
Shared state (classifier models, cache backends, configuration, auth token caches) is read concurrently by all active streams, with synchronization limited to cache writes, auth token refreshes, and metric updates.


\section{Memory and Retrieval-Augmented Generation}
\label{sec:memory_rag}

Production routing systems must support multi-turn conversations with persistent context and knowledge-augmented responses.
We describe the memory and RAG subsystems that operate as plugins within the routing pipeline.

\subsection{Persistent Memory}

The memory system maintains user-scoped knowledge across conversation sessions, enabling personalized routing and context-aware responses.
\Cref{fig:memory_lifecycle} illustrates the full memory lifecycle.

\noindent\textbf{Memory storage.}
Each conversation turn is stored directly as an \emph{episodic chunk}---no external LLM is required.
The user message and assistant response are concatenated into a \texttt{Q:/A:} block, sanitized (UTF-8 validation, 16\,KB cap per entry), and embedded.
A lightweight \emph{entropy gate} discards turns that carry no retrievable signal---greetings, acknowledgments, and short one-word replies---before embedding, reducing index pollution.
Every $s$ turns a session-level sliding-window chunk spanning the last $w$ turns is additionally stored (defaults: $s{=}3$, $w{=}5$), creating overlapping windows so facts near turn boundaries co-occur in at least one chunk---improving multi-hop retrieval coverage.

\noindent\textbf{Retrieval gating.}
Not every query benefits from memory retrieval.
A lightweight heuristic determines whether memory search is warranted by filtering out general fact-check queries, tool-augmented requests, and simple greetings, avoiding unnecessary embedding lookups and reducing latency for queries where personal context is irrelevant.

\noindent\textbf{Retrieval.}
At query time, relevant memories are retrieved via the hybrid search pipeline (vector similarity, BM25, and n-gram matching) over the user's memory store.
An optional query-rewriting step reformulates the user's query for improved retrieval recall.
Adaptive thresholding adjusts the similarity cutoff based on retrieval mode, preventing cosine-calibrated thresholds from silently discarding all RRF-scored results.

\noindent\textbf{Post-retrieval filtering.}
Retrieved memories pass through a composable \textsc{ReflectionGate} before context injection:
\begin{enumerate}[nosep,leftmargin=*]
  \item \emph{Safety}: a regex block-list prevents prompt-injection payloads from being surfaced.
  \item \emph{Recency decay}: memories are weighted by recency, prioritizing recent context.
  \item \emph{Deduplication}: near-duplicate entries (high Jaccard similarity) are collapsed to a single representative.
  \item \emph{Budget}: the final set is capped at a configurable count to bound injected context length.
\end{enumerate}
Filtered context is injected as a separate conversation message positioned after system instructions but before user turns, following the context-injection pattern of the OpenAI Agents SDK.

\noindent\textbf{Background consolidation.}
A scheduled background job merges semantically related memories using greedy single-linkage clustering over word-level Jaccard similarity.
Memories within each cluster are replaced by a single deduplicated summary entry, reducing store redundancy and improving retrieval precision over time.

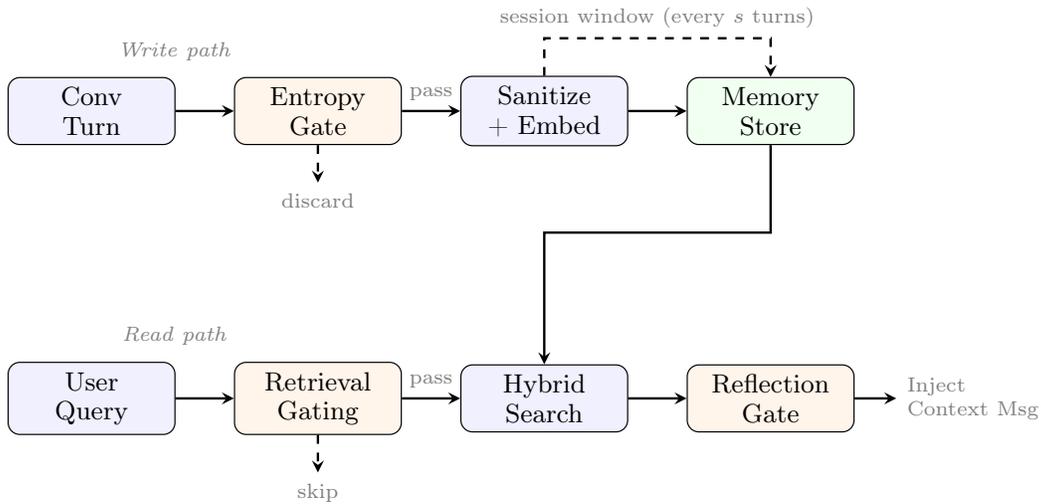
\begin{figure}[ht]
\centering
\resizebox{\linewidth}{!}{%
\begin{tikzpicture}[
    node distance=0.55cm and 0.7cm,
    box/.style={draw, rounded corners, fill=blue!6, minimum width=2.0cm,
                minimum height=0.8cm, align=center, font=\small},
    storebox/.style={draw, rounded corners, fill=green!6, minimum width=2.0cm,
                minimum height=0.8cm, align=center, font=\small},
    gatebox/.style={draw, rounded corners, fill=orange!8, minimum width=2.0cm,
                minimum height=0.8cm, align=center, font=\small},
    arr/.style={->, >=stealth, thick},
    darr/.style={->, >=stealth, thick, dashed},
    lbl/.style={font=\scriptsize, text=gray}
]

\node[box] (conv) {Conv\\Turn};
\node[gatebox, right=of conv] (entropy) {Entropy\\Gate};
\node[box, right=of entropy] (sanitize) {Sanitize\\+ Embed};
\node[storebox, right=of sanitize] (store) {Memory\\Store};

\draw[arr] (conv) -- (entropy);
\draw[arr] (entropy) -- node[above, lbl] {pass} (sanitize);
\draw[arr] (sanitize) -- (store);

\draw[darr] (entropy.south) -- ++(0,-0.45) node[below, lbl] {discard};

\draw[darr] (sanitize.north) -- ++(0,0.45) -| node[above, lbl, pos=0.25] {session window (every $s$ turns)} (store.north);

\node[box, below=2.6cm of conv] (query) {User\\Query};
\node[gatebox, right=of query] (rgate) {Retrieval\\Gating};
\node[box, right=of rgate] (hybrid) {Hybrid\\Search};
\node[gatebox, right=of hybrid] (reflect) {Reflection\\Gate};

\draw[arr] (query) -- (rgate);
\draw[arr] (rgate) -- node[above, lbl] {pass} (hybrid);
\draw[arr] (hybrid) -- (reflect);
\draw[arr] (reflect.east) -- ++(0.5,0) node[right, lbl, align=left] {Inject\\Context Msg};

\draw[arr] (store.south) -- ++(0,-1.05) -| (hybrid.north);

\draw[darr] (rgate.south) -- ++(0,-0.45) node[below, lbl] {skip};

\node[lbl, above=0.08cm of conv, xshift=1.0cm] {\textit{Write path}};
\node[lbl, above=0.08cm of query, xshift=1.0cm] {\textit{Read path}};

\end{tikzpicture}
}
\caption{Memory lifecycle. \textit{Write path}: each conversation turn passes an entropy gate (discarding low-signal turns such as greetings), is sanitized and embedded as an episodic \texttt{Q:/A:} chunk, and stored directly---no LLM inference required. Every $s$ turns an additional sliding-window chunk is stored. \textit{Read path}: a heuristic gate skips retrieval for irrelevant queries; qualifying queries run hybrid search (vector + BM25 + n-gram) over the user's store; retrieved memories pass through a \textsc{ReflectionGate} (safety blocklist, recency decay, deduplication, budget cap) before injection as a separate conversation message.}
\label{fig:memory_lifecycle}
\end{figure}

\subsection{Retrieval-Augmented Generation}

The RAG plugin retrieves relevant documents from vector stores and injects them as context before model invocation.

\noindent\textbf{Indexing pipeline.}
Documents are chunked (configurable size and overlap), embedded using the shared embedding model (\Cref{sec:ml_inference}), and indexed in a vector store.

\noindent\textbf{Hybrid retrieval.}
Pure vector search can miss lexically relevant results when the embedding model underweights rare terms.
We implement a three-signal hybrid retrieval pipeline (\Cref{fig:hybrid_search}) that scores each candidate chunk along three axes:
\begin{enumerate}[nosep,leftmargin=*]
  \item \emph{Vector similarity}: cosine similarity from the embedding index (already in $[0,1]$).
  \item \emph{BM25}: an Okapi BM25 inverted index built over chunk contents provides keyword-level relevance, with configurable $k_1$ (term-frequency saturation, default~1.2) and $b$ (length normalization, default~0.75) parameters.
  \item \emph{N-gram Jaccard}: character $n$-gram sets (default $n{=}3$) capture fuzzy lexical overlap, producing Jaccard similarity in $[0,1]$.
\end{enumerate}

\noindent\textbf{Score fusion.}
Two fusion modes combine the three retriever scores into a single ranking:
(1)~\emph{Weighted}: BM25 scores are min-max normalized to $[0,1]$, then the final score is $w_v \cdot s_{\text{vec}} + w_b \cdot s_{\text{bm25}} + w_n \cdot s_{\text{ngram}}$ with configurable weights (defaults: 0.7, 0.2, 0.1).
(2)~\emph{Reciprocal Rank Fusion (RRF)}: $\text{score}(d) = \sum_r 1/(k + \text{rank}_r(d))$, which is parameter-free beyond the constant $k$ (default~60).

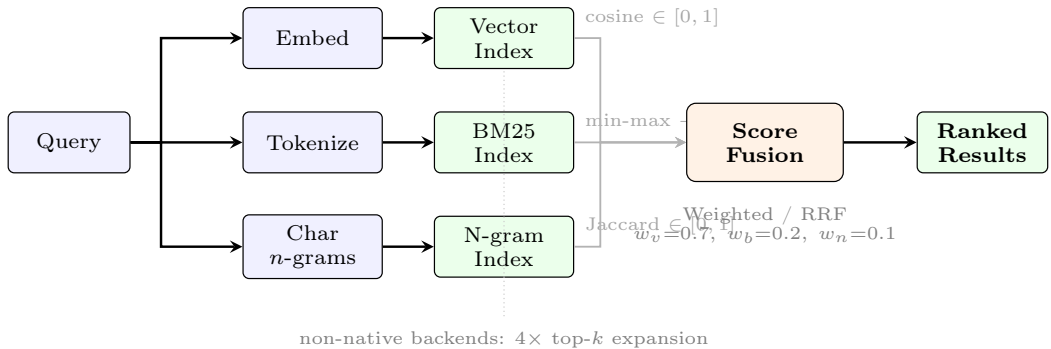
\begin{figure}[t]
\centering
\resizebox{\linewidth}{!}{%
\begin{tikzpicture}[
    node distance=0.5cm and 0.8cm,
    box/.style={draw, rounded corners=2pt, fill=blue!6,
                minimum height=0.7cm, minimum width=1.6cm,
                align=center, font=\scriptsize},
    idx/.style={draw, rounded corners=2pt, fill=green!8,
                minimum height=0.7cm, minimum width=1.6cm,
                align=center, font=\scriptsize},
    fuse/.style={draw, rounded corners=3pt, fill=orange!10,
                 minimum height=0.9cm, minimum width=1.8cm,
                 align=center, font=\scriptsize\bfseries},
    score/.style={font=\tiny, text=gray!70},
    arr/.style={->, >=stealth, thick},
    arrs/.style={->, >=stealth, semithick, gray!60},
    lbl/.style={font=\tiny, text=gray},
  ]

  \node[box, minimum width=1.4cm] (query) {Query};

  \node[box] (embed) at (2.8, 1.2) {Embed};
  \node[idx] (vecidx) at (5.0, 1.2) {Vector\\Index};

  \node[box] (tok) at (2.8, 0.0) {Tokenize};
  \node[idx] (bm25idx) at (5.0, 0.0) {BM25\\Index};

  \node[box] (ngram) at (2.8, -1.2) {Char\\$n$-grams};
  \node[idx] (ngramidx) at (5.0, -1.2) {N-gram\\Index};

  \draw[arr] (query.east) -- ++(0.35,0) |- (embed.west);
  \draw[arr] (query.east) -- ++(0.35,0) |- (tok.west);
  \draw[arr] (query.east) -- ++(0.35,0) |- (ngram.west);

  \draw[arr] (embed) -- (vecidx);
  \draw[arr] (tok) -- (bm25idx);
  \draw[arr] (ngram) -- (ngramidx);

  \node[score, above=0.01cm of vecidx.east, anchor=south west,
      xshift=-1.6cm, yshift=0.3cm] {cosine $\in [0,1]$};
  \node[score, above=0.01cm of bm25idx.east, anchor=south west,
      xshift=-1.8cm, yshift=0.3cm] {min-max $\to [0,1]$};
  \node[score, above=0.01cm of ngramidx.east, anchor=south west,
      xshift=-1.6cm, yshift=0.3cm] {Jaccard $\in [0,1]$};

  \node[fuse] (fusion) at (8.0, 0.0) {Score\\Fusion};

  \draw[arrs] (vecidx.east) -- ++(0.3,0) |- (fusion.west);
  \draw[arrs] (bm25idx.east) -- (fusion.west);
  \draw[arrs] (ngramidx.east) -- ++(0.3,0) |- (fusion.west);

  \node[lbl, below=0.15cm of fusion] {Weighted / RRF};

  \node[box, minimum width=1.5cm, fill=green!8, font=\scriptsize\bfseries] (out) at (10.5, 0.0) {Ranked\\Results};
  \draw[arr] (fusion.east) -- (out.west);

  \draw[densely dotted, gray!40] (5.0, -2.0) -- (vecidx.south);
  \node[lbl, anchor=north] at (5.0, -2.05)
    {non-native backends: $4{\times}$ top-$k$ expansion};

  \node[lbl, anchor=north] at (8.0, -0.85)
    {$w_v{=}0.7,\; w_b{=}0.2,\; w_n{=}0.1$};

\end{tikzpicture}
}
\caption{Hybrid search pipeline. Each candidate chunk is scored along three axes---vector cosine similarity, BM25 keyword relevance, and character $n$-gram Jaccard overlap---then combined via weighted fusion or Reciprocal Rank Fusion (RRF). Backends without native hybrid support fall back to a generic rerank path that fetches $4{\times}$ top-$k$ candidates from vector search before applying BM25 and $n$-gram scoring.}
\label{fig:hybrid_search}
\end{figure}

\noindent\textbf{Backend abstraction.}
The RAG plugin accesses vector stores through a common \texttt{VectorStoreBackend} interface, decoupling retrieval logic from storage implementation.
Six backend types are supported (\Cref{fig:rag_backends}):
\begin{itemize}[nosep,leftmargin=*]
  \item \emph{In-memory}: development and testing; no external dependencies.
  \item \emph{Milvus}~\cite{wang2021milvus}: production-grade distributed vector database with native hybrid search support.
  \item \emph{Llama Stack}: delegates vector storage and search to a Llama Stack deployment via its OpenAI-compatible \texttt{/v1/vector\_stores} API, enabling unified management of LLM serving and retrieval through a single platform. When the Llama Stack instance is configured with the Milvus \texttt{vector\_io} provider, the backend supports hybrid search by passing \texttt{ranking\_options: \{ranker: ``rrf''\}} in search requests, combining vector similarity with BM25 keyword matching at the provider level.
  \item \emph{External API}: any OpenAI-compatible vector store endpoint.
  \item \emph{MCP}: retrieval via Model Context Protocol tool servers.
  \item \emph{OpenAI file search}: delegated retrieval through OpenAI's hosted file search API.
\end{itemize}
Backends that implement native hybrid search (Milvus, Llama Stack with Milvus provider) use their own BM25 and keyword indexes; all other backends fall back to a generic rerank path that fetches an expanded candidate set ($4{\times}$ top-$k$) from vector search and applies BM25 and n-gram scoring as a post-retrieval reranking step.

\noindent\textbf{Score-range awareness.}
Different retrieval modes produce scores on fundamentally different scales: cosine similarity yields values in $[0, 1]$ where a threshold of $\sim$0.7 is typical, while RRF scores follow $\sum 1/(k + \text{rank})$ and typically range from 0.001 to 0.05.
Applying a cosine-calibrated threshold to RRF scores would silently discard all results.
The backend interface handles this transparently: when hybrid search is active, score-based filtering is bypassed and result volume is controlled solely by the top-$k$ parameter.

\begin{figure}[t]
\centering
\resizebox{\linewidth}{!}{%
\begin{tikzpicture}[
    node distance=0.5cm and 0.7cm,
    box/.style={draw, rounded corners, fill=blue!6, minimum width=2cm,
                minimum height=0.7cm, align=center, font=\small},
    backend/.style={draw, rounded corners, fill=green!6, minimum width=1.8cm,
                minimum height=0.6cm, align=center, font=\scriptsize},
    arr/.style={->, >=stealth, thick},
    lbl/.style={font=\scriptsize, text=gray}
]

\node[box] (query) {User Query};
\node[box, right=1.2cm of query] (rag) {RAG Plugin};
\node[box, right=1.2cm of rag] (iface) {VectorStore\\Interface};

\draw[arr] (query) -- (rag);
\draw[arr] (rag) -- (iface);

\node[backend, above right=0.4cm and 1.2cm of iface] (milvus) {Milvus};
\node[backend, right=1.2cm of iface] (llama) {Llama Stack};
\node[backend, below right=0.4cm and 1.2cm of iface] (others) {MCP / API /\\OpenAI};

\draw[arr] (iface) -- (milvus);
\draw[arr] (iface) -- (llama);
\draw[arr] (iface) -- (others);

\node[lbl, right=0.1cm of milvus] {hybrid};
\node[lbl, right=0.1cm of llama] {hybrid (via Milvus)};
\node[lbl, right=0.1cm of others] {vector only};

\node[box, below=1.5cm of rag] (inject) {Inject into Prompt};
\draw[arr] (rag) -- (inject);

\end{tikzpicture}
}
\caption{RAG backend architecture. The RAG plugin accesses vector stores through a common interface. Milvus and Llama Stack (with Milvus provider) support native hybrid search combining vector similarity with BM25 keyword matching via Reciprocal Rank Fusion. Other backends use vector-only retrieval with optional post-retrieval reranking.}
\label{fig:rag_backends}
\end{figure}
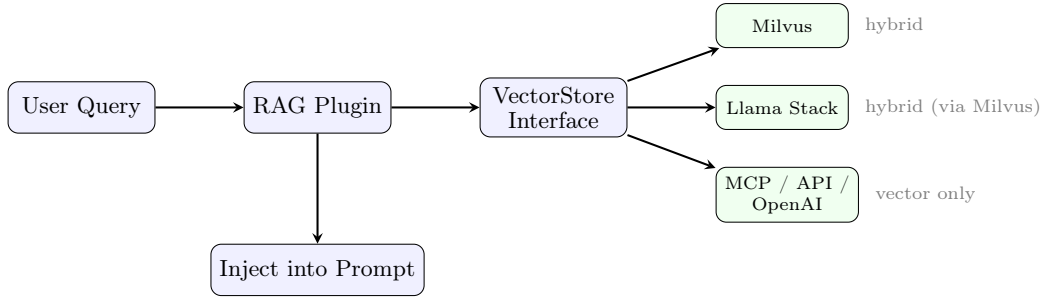

\subsection{Stateful Conversations (Response API)}

The system supports the OpenAI Responses API for stateful multi-turn conversations:

\noindent\textbf{Conversation chaining.}
Each response is stored with a unique ID.
Subsequent requests reference \texttt{previous\_response\_id} to reconstruct the full conversation history without retransmitting the full message sequence.
A bidirectional translator converts between the Response API format and Chat Completions format for backend model invocation, enabling all routing, safety, and caching features to operate identically on both API surfaces.

\noindent\textbf{Routing continuity.}
Stored responses include routing metadata (decision, model selection, signal results), enabling consistent routing across conversation turns and providing context for feedback-driven model selection.

\noindent\textbf{State backends.}
Conversation state is persisted via three backends:
(1)~\emph{in-memory} for development;
(2)~\emph{Redis} for production deployments requiring distributed state with high availability---supporting both standalone and cluster modes;
(3)~\emph{Milvus} for deployments that benefit from semantic retrieval over conversation history.
The Redis backend enables horizontal scaling: multiple router replicas share conversation state, and Redis persistence ensures that conversation chains survive pod restarts.

\subsection{Integration with Signal-Decision Architecture}

Both memory and RAG operate as per-decision plugins.
Different decisions can activate different RAG configurations (different vector stores, different $k$ values, different chunk strategies, different search modes) or disable retrieval entirely.
This enables, for example, a ``research assistant'' decision that activates RAG with hybrid search over a technical knowledge base while a ``casual chat'' decision disables retrieval.


\section{Observability}
\label{sec:observability}

Intelligent routing introduces a new observability surface: beyond standard model serving metrics, operators must monitor signal extraction quality, decision matching patterns, plugin effectiveness, and model selection outcomes.

\subsection{Metrics Taxonomy}

We instrument four metric categories using Prometheus~\cite{prometheus2024}:

\noindent\textbf{Model performance.}
Per-model request counts, token usage, estimated cost, completion latency, Time-to-First-Token (TTFT), and Time-per-Output-Token (TPOT).
These metrics enable real-time cost monitoring and performance regression detection across the model fleet.

\noindent\textbf{Routing behavior.}
Routing modification counts (original model vs.\ selected model), reason codes (which signal types triggered which decisions), and routing latency (overhead of the routing pipeline itself).
These metrics answer the question: \emph{how is the router changing traffic patterns?}

\noindent\textbf{Signal and decision quality.}
Per-signal-type extraction counts and match rates, per-decision match frequencies and confidence distributions, and per-plugin execution counts and outcomes.
These metrics enable calibration: if a signal type rarely matches, its threshold may need adjustment; if a decision matches too broadly, its conditions may be under-specified.

\noindent\textbf{Safety and cache effectiveness.}
PII violation rates by entity type, jailbreak detection rates, hallucination detection latency, cache hit rates, and cache operation latency.
These metrics quantify the value delivered by the plugin chain.

\subsection{Distributed Tracing}

We implement OpenTelemetry~\cite{opentelemetry2024} tracing with a hierarchical span model:
\begin{itemize}[leftmargin=*]
  \item \textbf{Root span}: Covers the full request lifecycle from receipt to response.
  \item \textbf{Signal spans}: Individual spans for each signal type evaluation, capturing latency and results.
  \item \textbf{Decision span}: Decision evaluation with the matched decision and confidence.
  \item \textbf{Plugin spans}: Per-plugin execution with type-specific attributes (cache hit/miss, PII types detected, hallucination spans found).
  \item \textbf{Upstream span}: Backend model invocation, with W3C Trace Context propagation enabling end-to-end tracing through vLLM~\cite{kwon2023vllm} and other inference frameworks.
\end{itemize}

This span hierarchy enables operators to diagnose routing latency (``which signal is slow?''), understand routing decisions (``why was this query routed to model X?''), and correlate routing behavior with model serving performance.


\section{Deployment}
\label{sec:deployment}

We describe the deployment architecture that enables the routing system to operate from single-node development to production Kubernetes~\cite{kubernetes2024} clusters.

\subsection{Deployment Modes}

\noindent\textbf{Local development.}
A single command (\texttt{pip install vllm-sr \&\& vllm-sr serve}) bootstraps the complete stack: router, Envoy proxy, and dashboard.
This lowers the barrier to experimentation with routing configurations.

\noindent\textbf{Kubernetes: Standalone mode.}
Envoy runs as a sidecar container alongside the router in the same pod.
The ExtProc filter connects via localhost gRPC.
Deployed via Helm charts with configurable replicas, resource limits, cache backends, and autoscaling.

\noindent\textbf{Kubernetes: Gateway mode.}
The router runs as an independent service behind an existing Istio or Envoy Gateway deployment, referenced via the gateway's ExtProc configuration.
This mode shares the gateway infrastructure across multiple services.

\subsection{Kubernetes Operator}

A custom operator manages the lifecycle of routing deployments through a \texttt{SemanticRouter} Custom Resource Definition (CRD).
The reconciliation loop manages: service accounts, configuration (ConfigMap or CRD-sourced), persistent storage, gateway/route resources, Envoy configuration, deployments, services, and horizontal pod autoscalers.

\noindent\textbf{Backend discovery.}
The operator discovers model backends via three mechanisms:
KServe InferenceService resources (for managed model serving),
label-based Llama Stack discovery, and
direct Kubernetes Service references.

\subsection{Dashboard}

A web console (React frontend, Go backend) provides:
configuration editing with live validation,
topology visualization of routing flows,
an interactive playground for testing routing decisions,
embedded Grafana/Prometheus/Jaeger dashboards for monitoring and tracing, and
an evaluation framework for benchmarking routing quality.


\section{Evaluation}
\label{sec:evaluation}

We evaluate the routing system across four dimensions: signal extraction efficiency, GPU-accelerated long-context inference, LoRA multi-task scaling, and end-to-end routing correctness.

\subsection{Signal Extraction Latency}

\Cref{tab:signal_latency} reports median and p99 latencies for each signal type on an NVIDIA A100 GPU with ModernBERT base model.

\begin{table}[htbp]
\centering
\caption{Signal extraction latency by type}
\label{tab:signal_latency}
\begin{tabular}{lrrl}
\toprule
\textbf{Signal Type} & \textbf{Median} & \textbf{p99} & \textbf{Requires ML} \\
\midrule
Keyword      & $< 0.1$\,ms & $< 0.5$\,ms  & No \\
Context      & $< 0.1$\,ms & $< 0.5$\,ms  & No \\
Language     & $< 0.5$\,ms & $< 1$\,ms    & No \\
Authorization & $< 0.1$\,ms & $< 0.5$\,ms & No \\
\midrule
Embedding    & $15$\,ms    & $45$\,ms     & Yes \\
Domain       & $60$\,ms    & $120$\,ms    & Yes \\
Fact-check   & $55$\,ms    & $110$\,ms    & Yes \\
Modality     & $50$\,ms    & $100$\,ms    & Yes \\
Feedback     & $55$\,ms    & $115$\,ms    & Yes \\
Complexity   & $50$\,ms    & $105$\,ms    & Yes \\
Preference   & $55$\,ms    & $110$\,ms    & Yes \\
\bottomrule
\end{tabular}
\end{table}

Heuristic signals complete in sub-millisecond time, while ML signals range from 15--120\,ms.
With parallel evaluation, the wall-clock time is dominated by the slowest active signal ($\sim$120\,ms for domain classification) rather than the sum.

\subsection{GPU-Accelerated Inference with SDPA}

To accelerate ML signal extraction on long prompts, we re-export the mmBERT-32K classifiers from PyTorch using Scaled Dot-Product Attention (SDPA), apply ONNX Runtime graph optimization (Gelu/LayerNorm fusion), and convert weights to FP16.
\Cref{tab:gpu_acceleration} reports single-session inference latency on an AMD Instinct MI300X using the ROCm execution provider, compared against CPU (x86-64, ORT default provider).

\begin{table}[htbp]
\centering
\caption{mmBERT-32K inference latency: CPU (FP32) vs.\ GPU SDPA (FP16), single session, batch size 1}
\label{tab:gpu_acceleration}
\begin{tabular}{rrrr}
\toprule
\textbf{Seq.\ Length} & \textbf{CPU (ms)} & \textbf{GPU (ms)} & \textbf{Speedup} \\
\midrule
512    & 120   & 6.0   & 20$\times$ \\
1{,}024  & 263   & 7.7   & 34$\times$ \\
2{,}048  & 809   & 14.1  & 57$\times$ \\
4{,}096  & 2{,}664 & 57.6  & 46$\times$ \\
8{,}192  & 9{,}656 & 237   & 41$\times$ \\
\bottomrule
\end{tabular}
\end{table}

The SDPA re-export reduces the ONNX graph from 2{,}902 to 1{,}342 operators; FP16 conversion halves model size from 1.2\,GB to 587\,MB.
At 8K tokens, single-session GPU inference completes in 237\,ms---a 41$\times$ speedup over CPU.
However, when all three classifiers are loaded concurrently (the operational configuration), SDPA's $\mathcal{O}(n^2)$ memory footprint causes OOM beyond $\sim$4K tokens.
The next section addresses this limitation with flash attention.

\subsection{Flash Attention via Composable Kernel}
\label{sec:flash_attention}

Standard SDPA materializes the full $\mathcal{O}(n^2)$ attention mask and $QK^\top$ matrix, causing out-of-memory (OOM) failures beyond $\sim$4{,}096 tokens when three classifiers are loaded concurrently.
We replace SDPA attention subgraphs with a custom ONNX Runtime operator backed by AMD's Composable Kernel (CK) tiled flash attention, which computes attention in $\mathcal{O}(n)$ working memory.
For ModernBERT's local-attention layers (sliding window of 128), the CK kernel's native \texttt{window\_size} parameters eliminate the dense mask entirely; a lightweight 1-D padding bias $[B,1,1,S]$ replaces the $[1,1,S,S]$ broadcast mask for all layers.

\Cref{tab:fa_vs_sdpa} compares SDPA and FA at concurrency 1.
\Cref{tab:fa_concurrency} then varies the number of concurrent requests to measure throughput under load.
All measurements are median latency on an AMD Instinct MI300X with three classifiers loaded concurrently (domain, jailbreak, PII evaluated in parallel per request).

\begin{table}[htbp]
\centering
\caption{End-to-end classification latency: SDPA vs.\ CK Flash Attention (FP16, ROCm EP, 3 classifiers in parallel, concurrency 1, MI300X)}
\label{tab:fa_vs_sdpa}
\begin{tabular}{rrrr}
\toprule
\textbf{Seq.\ Length} & \textbf{SDPA (ms)} & \textbf{FA (ms)} & \textbf{Speedup} \\
\midrule
512       & 19    & 19    & 1.0$\times$ \\
1{,}024   & 26    & 23    & 1.1$\times$ \\
2{,}048   & 51    & 32    & 1.6$\times$ \\
4{,}096   & 167   & 51    & 3.3$\times$ \\
8{,}192   & OOM   & 105   & --- \\
16{,}384  & OOM   & 259   & --- \\
32{,}768  & OOM   & 756   & --- \\
\bottomrule
\end{tabular}
\end{table}

\begin{table}[htbp]
\centering
\caption{CK Flash Attention latency under concurrent load (median / p95 in ms, MI300X)}
\label{tab:fa_concurrency}
\begin{tabular}{rrrrrrr}
\toprule
 & \multicolumn{2}{c}{\textbf{$C{=}1$}} & \multicolumn{2}{c}{\textbf{$C{=}10$}} & \multicolumn{2}{c}{\textbf{$C{=}20$}} \\
\cmidrule(lr){2-3}\cmidrule(lr){4-5}\cmidrule(lr){6-7}
\textbf{Seq.\ Len.} & \textbf{med} & \textbf{p95} & \textbf{med} & \textbf{p95} & \textbf{med} & \textbf{p95} \\
\midrule
512       & 19   & 19    & 77    & 123   & 142   & 238 \\
1{,}024   & 23   & 24    & 104   & 167   & 189   & 325 \\
2{,}048   & 32   & 32    & 141   & 234   & 275   & 456 \\
4{,}096   & 51   & 52    & 281   & 998   & 521   & 847 \\
8{,}192   & 105  & 105   & 601   & 873   & 1{,}058 & 1{,}732 \\
16{,}384  & 259  & 260   & 1{,}567 & 2{,}426 & 3{,}089 & 4{,}817 \\
32{,}768  & 756  & 763   & 5{,}406 & 7{,}503 & 9{,}872 & 14{,}862 \\
\bottomrule
\end{tabular}
\end{table}

At concurrency 1 (\Cref{tab:fa_vs_sdpa}), FA matches SDPA at short sequences and is 3.3$\times$ faster at 4K tokens; beyond 4K tokens SDPA cannot allocate the 1.5\,GB attention mask when three models share the GPU.
Under concurrent load (\Cref{tab:fa_concurrency}), FA scales gracefully: 20 simultaneous 4K-token requests complete with a median of 521\,ms (p95 847\,ms), and even 20 concurrent 32K-token requests complete without OOM---a regime entirely inaccessible to SDPA.
Latency scales approximately linearly with concurrency, confirming that the GPU compute pipeline serializes requests rather than experiencing pathological queuing.

\subsection{LoRA Memory Efficiency}

\Cref{tab:lora_memory} shows the memory advantage of serving classifiers via LoRA adapters versus independent fine-tuned models.

\begin{table}[htbp]
\centering
\caption{Model memory: independent fine-tuned models vs.\ LoRA adapters (ModernBERT base, 150M params)}
\label{tab:lora_memory}
\begin{tabular}{lrr}
\toprule
\textbf{Tasks ($n$)} & \textbf{Independent (MB)} & \textbf{LoRA (MB)} \\
\midrule
1 & 573   & 573 \\
3 & 1{,}719 & 574 \\
6 & 3{,}438 & 575 \\
\bottomrule
\end{tabular}
\end{table}

At $n = 6$, the LoRA architecture requires $\sim$6$\times$ less model memory (one base model + six $\sim$0.2\,MB adapters vs.\ six full model copies).
Each task still requires its own forward pass; latency reduction comes from \emph{parallel execution} of classifiers rather than from LoRA itself (\Cref{sec:lora_mom}).

\subsection{Decision Engine Overhead}

Decision evaluation adds negligible latency:
$< 0.1$\,ms for 10 decisions with 3 conditions each;
$< 0.5$\,ms for 100 decisions with 5 conditions each.
This confirms that the $O(M \cdot L_\text{max})$ complexity is dominated by signal extraction.

\subsection{Composable Orchestration Across Deployment Scenarios}

A key claim of this work is that the same architecture serves diverse deployment scenarios through configuration.
\Cref{tab:deployment_scenarios} demonstrates how different signal-decision-plugin compositions address different requirements:

\begin{table}[htbp]
\centering
\caption{Composable signal orchestration across deployment scenarios. Each scenario activates a different subset of the thirteen signal types, selection algorithms, and plugin chains---using the same system binary and architecture.}
\label{tab:deployment_scenarios}
\begin{tabularx}{\linewidth}{
  >{\raggedright\arraybackslash}p{2.0cm}
  >{\raggedright\arraybackslash}X
  >{\raggedright\arraybackslash}p{1.9cm}
  >{\raggedright\arraybackslash}X
}
\toprule
\textbf{Scenario} & \textbf{Active Signals} & \textbf{Selection} & \textbf{Key Plugins} \\
\midrule
Privacy-regulated (healthcare) & authz, domain, language & Static (compliant models only) & Strict PII redaction, no caching, audit logging \\
Cost-optimized (developer tool) & complexity, embedding, keyword & AutoMix cascade & Aggressive semantic cache, header mutation for LoRA adapter \\
Multi-cloud enterprise & domain, modality, authz & Latency-aware & Multi-endpoint failover, provider auth factory, system prompt injection \\
Multi-turn assistant & embedding, feedback, preference & Elo with session pin & Responses API state, memory retrieval, RAG injection \\
\bottomrule
\end{tabularx}
\end{table}

\subsection{End-to-End Routing Correctness}

The end-to-end test framework validates routing behavior across eight scenario profiles:

\begin{table}[htbp]
\centering
\caption{End-to-end test profiles}
\label{tab:e2e_profiles}
\begin{tabularx}{\linewidth}{
  >{\raggedright\arraybackslash}p{2.4cm}
  >{\raggedright\arraybackslash}X
}
\toprule
\textbf{Profile} & \textbf{Validated Behavior} \\
\midrule
Multi-endpoint      & Multi-provider routing with weighted distribution and failover across heterogeneous backends \\
Multi-provider auth & Provider-specific auth injection (API key, OAuth2, cloud IAM) via authorization factory \\
AuthZ-RBAC          & Role-based model access (admin/premium/free tiers) with authz signal \\
ML model selection  & KNN, KMeans, SVM, MLP selection accuracy on held-out queries \\
Keyword routing     & Keyword signal matching with AND/OR/NOR combinators \\
Embedding routing   & Embedding similarity thresholds and confidence-based decision selection \\
RAG + Responses API & Context retrieval, injection, and stateful multi-turn via Responses API \\
Routing strategies  & Static, Elo, RouterDC, AutoMix, Hybrid algorithm comparison \\
\bottomrule
\end{tabularx}
\end{table}

Each profile validates correct model selection, safety enforcement (jailbreak blocked, PII detected), cache behavior (hits after similar queries), multi-provider routing (correct endpoint resolution and auth injection), and header propagation.

\subsection{Semantic Cache Effectiveness}

At a similarity threshold $\theta = 0.92$:
exact-match queries achieve 100\% hit rate with $< 5$\,ms lookup latency;
paraphrased queries achieve 60--80\% hit rate depending on paraphrase distance.
Cache hits eliminate backend model invocation entirely, reducing per-request cost to embedding computation only.

\subsection{Unified MoM Evaluation Framework}

To validate the robustness of the Mixture of Models (MoM) collection, we implemented a unified evaluation pipeline that benchmarks both merged models and LoRA adapters. The framework standardizes the assessment of heterogeneous model variants across intent classification and PII detection tasks.

The evaluation architecture, shown in \Cref{fig:eval_pipeline}, utilizes the following components:

\begin{itemize}
    \item \textbf{Schema Normalization:} Inputs from diverse sources—including MMLU-Pro for intent and Presidio for token classification—are mapped to a common evaluation schema.
    \item \textbf{Comparative Quality Validation:} The system computes weighted F1-scores and per-class precision/recall to ensure that the memory efficiency gains reported in \Cref{tab:lora_memory} do not result in significant predictive degradation compared to full-parameter models.
    \item \textbf{Parallelized Benchmarking:} Large-scale evaluations are executed via \texttt{ProcessPoolExecutor} to minimize wall-clock time. The pipeline incorporates automated OOM recovery and exponential backoff for API-based signal providers to ensure benchmark reliability.
\end{itemize}

Results include p50 and p99 latency profiling to confirm that model orchestration remains within the operational bounds required for real-time routing.

\begin{figure}[ht]
    \centering
    \includegraphics[width=0.94\linewidth]{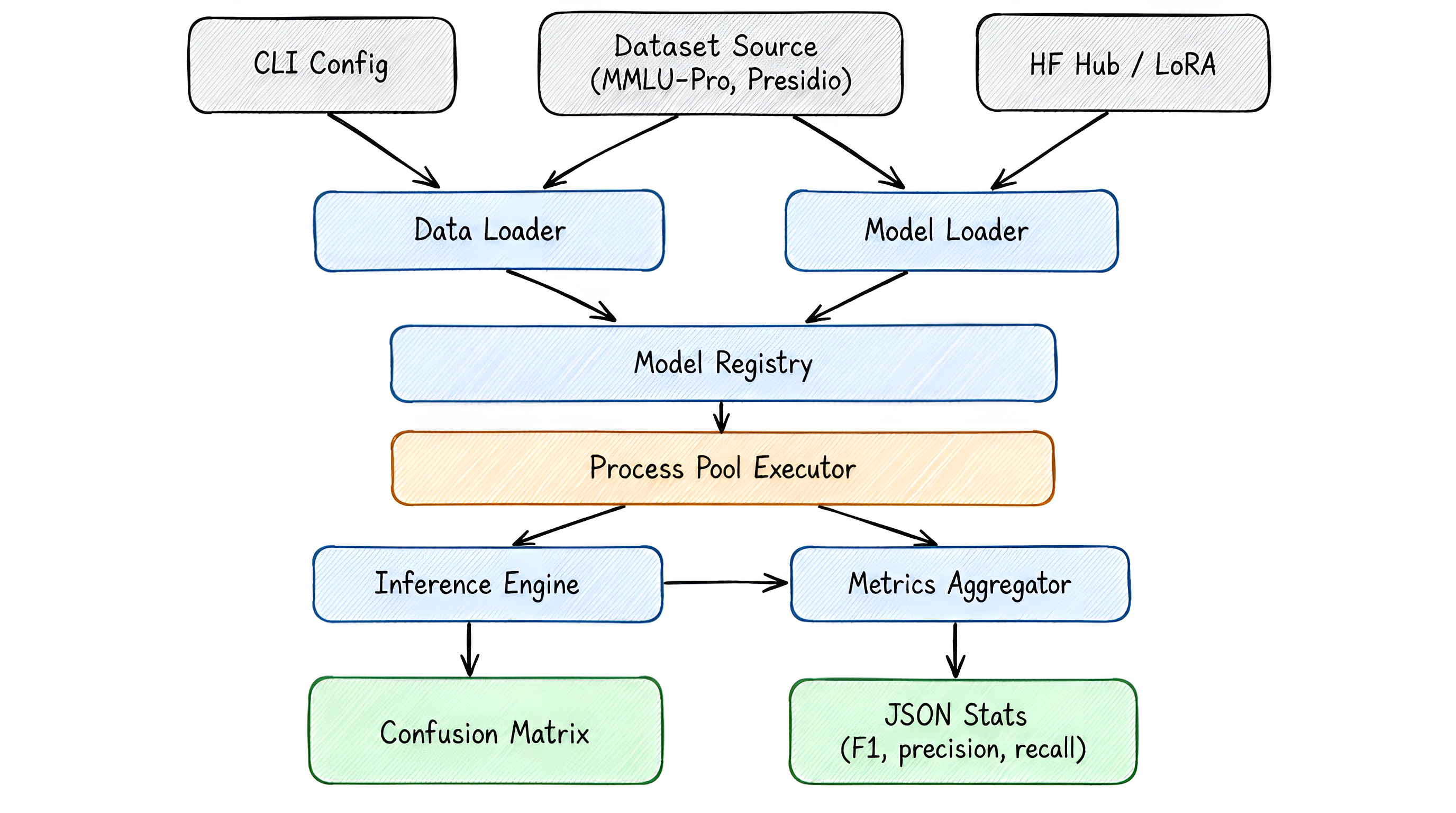}
    \caption{Unified MoM evaluation pipeline. Configuration and dataset sources feed into the data and model loaders, which populate the model registry. A process pool executor parallelizes inference across model variants, producing confusion matrices and aggregated metrics (weighted F1, per-class precision/recall) for both merged and LoRA models.}
    \label{fig:eval_pipeline}
\end{figure}

\subsection{Open Evaluation}

Detailed model selection quality comparisons across algorithms, HaluGate detection accuracy on standard benchmarks (HaluEval, FActScore), and large-scale routing quality evaluation are under preparation in collaboration with the RouterArena team.


\section{Related Work}
\label{sec:related_work}

\subsection{LLM Routing and Model Selection}

\noindent\textbf{Binary routing.}
RouteLLM~\cite{ong2024routellm} pioneered preference-data-driven routing between a strong and weak model, training BERT, MLP, and causal LLM classifiers to estimate query difficulty.
Our work extends this to multi-model, multi-signal routing with per-decision plugin chains.

\noindent\textbf{Contrastive selection.}
RouterDC~\cite{chen2024routerdc} learns shared query-model embeddings via dual contrastive learning.
We integrate RouterDC as one of thirteen selection algorithms and extend it with signal-conditioned features (domain category, complexity).

\noindent\textbf{Cascading.}
AutoMix~\cite{aggarwal2023automix} formulates model cascading as a POMDP with self-verification.
We integrate AutoMix within our plugin-aware framework, where safety checks and caching can prevent unnecessary escalation.

\noindent\textbf{Benchmarking.}
RouterBench~\cite{hu2024routerbench} proposed a benchmark for multi-LLM routing with hybrid scoring.
Our Hybrid selector builds on this approach.

\noindent\textbf{RL-based routing.}
Router-R1~\cite{zhang2025routerr1} formulates multi-LLM routing as a sequential decision process, using rule-based reward RL for multi-round routing and aggregation.
GMTRouter~\cite{xie2025gmtrouter} uses graph-based learning for personalized multi-turn interactions.
We integrate both within the unified selection interface and extend them with the ReMoM multi-round reasoning strategy.

A key distinction of our work is that prior approaches address model selection in isolation, while we embed selection within a composable signal orchestration framework that also handles signal extraction, safety enforcement, caching, context augmentation, and multi-provider routing---enabling the same selection algorithms to serve fundamentally different deployment scenarios through configuration.

\subsection{Multi-Provider and Multi-Endpoint Routing}

Commercial LLM gateway products (OpenRouter, AWS Bedrock, Azure AI Studio) provide multi-provider access but lack the composable signal-driven routing that enables differentiated policies per routing decision.
API management platforms (Kong, Apigee) offer gateway functionality but are not designed for semantic analysis of LLM requests.
Our system uniquely combines semantic model selection with multi-provider protocol abstraction, a pluggable authorization factory, and full OpenAI Responses API support for stateful conversations within the same composable framework.

\subsection{Mixture-of-Experts vs.\ Mixture-of-Models}

Sparse Mixture-of-Experts (MoE)~\cite{shazeer2017moe,jiang2024mixtral} routes \emph{tokens} to specialized sub-networks \emph{within} a single model architecture.
Our system operates at the \emph{request level}, routing entire requests across \emph{different model deployments}---a Mixture-of-Models (MoM) approach.
The two paradigms are complementary: our router can route to MoE models as backends.

\subsection{LLM Safety}

Prompt injection defenses range from perplexity filtering and input preprocessing~\cite{jain2023baseline} to fine-tuned safety classifiers~\cite{inan2023llamaguard}.
PII detection systems~\cite{lison2021anonymisation} identify sensitive information using rule-based and ML approaches.
Our safety subsystem integrates both within the routing pipeline with per-decision thresholds and policies, using LoRA adapters for memory-efficient multi-task classification.

\subsection{Hallucination Detection}

SelfCheckGPT~\cite{manakul2023selfcheckgpt} detects hallucinations via multi-sample consistency.
FActScore~\cite{min2023factscore} evaluates factual precision at the atomic fact level.
HaluGate differs in three respects:
(1)~a gating Sentinel that skips verification for non-factual queries;
(2)~token-level span identification rather than sentence-level scoring;
(3)~NLI-based explanation distinguishing contradiction from neutral unsupported content.

\subsection{Semantic Caching and RAG}

Semantic caching for LLMs uses embedding similarity to match incoming queries against previously seen requests, avoiding redundant model invocations.
Our cache extends this with per-decision policies, multiple backends, and integration with the safety pipeline (cache lookups occur \emph{after} safety checks but \emph{before} model invocation).
RAG integration~\cite{lewis2020rag,wang2021milvus} augments responses with retrieved context; our contribution is embedding RAG as a per-decision plugin within the routing framework.


\section{Conclusion}
\label{sec:conclusion}

We have presented \sysname{}, a signal-driven decision routing system for Mixture-of-Modality model deployments.
The central contribution is \textbf{composable signal orchestration}: the three-layer architecture---signal extraction, Boolean decision evaluation, per-decision plugin chains---enables diverse deployment scenarios to be expressed as different configurations over the same framework, without code changes.

Privacy-regulated deployments activate authz and PII signals with strict filtering plugins; cost-optimized deployments enable cascading selection with aggressive semantic caching; multi-cloud enterprises configure weighted multi-endpoint routing with provider-specific auth injection.
All use the same signal-decision-plugin machinery, composed differently.

Within this framework, \textbf{semantic model selection} analyzes each request's content through thirteen algorithms---spanning rating-based, contrastive, cascading, classical ML, reinforcement learning, and latency-aware families---to find the most cost-effective model while respecting per-decision privacy and safety constraints.
The integration of full \textbf{OpenAI Responses API support} enables stateful multi-turn routing with conversation-consistent model assignment; \textbf{multi-endpoint and multi-provider routing} abstracts over heterogeneous backends (vLLM, OpenAI, Anthropic, Azure, Bedrock, Gemini, Vertex AI) with transparent protocol translation; and the \textbf{pluggable authorization factory} supports diverse auth mechanisms across providers without coupling auth logic to routing decisions.

Additional technical contributions include:
(1)~the LoRA-based multi-task classification architecture that serves $n$ classifiers from a single base model, reducing aggregate model memory by $\sim$$n\times$;
(2)~HaluGate's gated three-stage hallucination detection pipeline that reduces average detection cost by $\sim$50\% through sentinel-based filtering; and
(3)~Rust-native ML inference bindings (Candle, Linfa, ONNX Runtime) that achieve sub-10\,ms signal extraction latency.

The system has been validated in production with over 600 merged contributions from 50+ engineers and is deployed as an Envoy ExtProc with Kubernetes operator support.

\subsection{Future Directions}

Several research directions emerge from this work:

\noindent\textbf{Learned decision policies.}
Replacing hand-crafted Boolean rules with learned routing policies (e.g., neural routing networks trained on production traffic) could improve routing quality while maintaining interpretability through attention-based explanation.

\noindent\textbf{Differentiable entropy-folding policies.}
A natural extension of the layered entropy-folding view is to relax symbolic priority gates into trainable soft gates while preserving explicit policy constraints, enabling gradient-based optimization of control depth with auditable safety boundaries.

\noindent\textbf{Adaptive cost optimization.}
Online learning approaches that continuously adapt model and provider selection based on real-time cost signals, latency measurements, and user feedback---extending the current offline-trained ML selectors to fully adaptive cost-quality optimization.

\noindent\textbf{Contrastive preference routing.}
The contrastive embedding method used for complexity classification generalizes to user preference learning: a contrastive preference classifier scores queries against exemplar sets representing different user preference profiles, enabling personalized model selection without per-user fine-tuning.

\noindent\textbf{Cross-provider consistency.}
Techniques for ensuring consistent behavior when routing the same conversation across different providers, addressing differences in instruction following, safety behavior, and output formatting.

\noindent\textbf{Multi-turn safety.}
Extending safety enforcement from single-turn to multi-turn conversations, detecting adversarial patterns that span multiple interaction rounds.
Preliminary work applies the contrastive embedding approach---already used for complexity classification---to multi-turn jailbreak detection: each user turn is scored against known-benign and known-adversarial exemplar sets, and accumulated turn-level signals are fused to detect gradual prompt escalation that evades single-turn classifiers.

\noindent\textbf{Federated signal orchestration.}
Extending composable signal orchestration to federated deployments where signals from multiple routing instances are aggregated for global optimization.

\noindent\textbf{Agent-based policy synthesis.}
The DSL configuration that specifies routing policies can be viewed as a program in a domain-specific language with a formally complete instruction set (\Cref{sec:disc_agent}).
This reframes router configuration as a \emph{program synthesis} problem: coding agents (LLMs fine-tuned for code generation) can translate natural-language routing specifications into valid DSL configurations, with routing quality feedback enabling RL-based optimization of the synthesis policy.
Preliminary experiments suggest that this meta-learning approach---learning to \emph{write} routing policies rather than to \emph{make} individual routing decisions---can significantly reduce configuration effort while maintaining the interpretability and verifiability of symbolic decision rules.

\noindent\textbf{Multi-protocol adapter abstraction.}
The current system is tightly coupled to Envoy's External Processor protocol. A multi-protocol adapter architecture would abstract the routing engine from protocol-specific code, enabling support for HTTP REST, native gRPC, Nginx/OpenResty, and custom protocols through thin translation layers.
This would also enable abstraction of backend proxying (currently Envoy-specific), external authorization mechanisms, and traffic management policies, making the routing engine truly protocol-agnostic and deployable in serverless, edge, and non-Envoy environments.

\subsection*{Acknowledgments}

We gratefully acknowledge the contributors who contributed to the project.
We thank the Hugging Face Candle team for collaboration on the Candle inference runtime, and the Envoy community for the ExtProc filter.
AMD has sponsored the project with resources and infrastructure as well as the vLLM project for the development of the project.
AI-assisted tools were used during the writing and proofreading of this paper.
The project is open-source at \url{https://github.com/vllm-project/semantic-router}.

\bibliographystyle{iclr2026_conference}
\bibliography{references}

\begin{thebibliography}{47}
\providecommand{\natexlab}[1]{#1}
\providecommand{\url}[1]{\texttt{#1}}
\expandafter\ifx\csname urlstyle\endcsname\relax
  \providecommand{\doi}[1]{doi: #1}\else
  \providecommand{\doi}{doi: \begingroup \urlstyle{rm}\Url}\fi

\bibitem[\AA{}str\"{o}m \& Murray(2008)\AA{}str\"{o}m and Murray]{astrom2008feedback}
Karl~Johan \AA{}str\"{o}m and Richard~M. Murray.
\newblock \emph{Feedback Systems: An Introduction for Scientists and Engineers}.
\newblock Princeton University Press, Princeton, NJ, 2008.

\bibitem[Aggarwal et~al.(2023)Aggarwal, Madaan, Anand, Potharaju, Mishra, Zhou, Gupta, Rajagopal, Kappaganthu, Yang, Upadhyay, Faruqui, and Mausam]{aggarwal2023automix}
Pranjal Aggarwal, Aman Madaan, Ankit Anand, Srividya~Pranavi Potharaju, Swaroop Mishra, Pei Zhou, Aditya Gupta, Dheeraj Rajagopal, Karthik Kappaganthu, Yiming Yang, Shyam Upadhyay, Manaal Faruqui, and Mausam.
\newblock {AutoMix}: Automatically mixing language models.
\newblock \emph{arXiv preprint arXiv:2310.12963}, 2023.

\bibitem[Bellman \& Zadeh(1970)Bellman and Zadeh]{bellman1970decision}
Richard~E. Bellman and Lotfi~A. Zadeh.
\newblock Decision-making in a fuzzy environment.
\newblock \emph{Management Science}, 17\penalty0 (4):\penalty0 B--141--B--164, 1970.

\bibitem[Brayton et~al.(1984)Brayton, Hachtel, McMullen, and Sangiovanni-Vincentelli]{brayton1984logic}
Robert~K. Brayton, Gary~D. Hachtel, Curtis~T. McMullen, and Alberto~L. Sangiovanni-Vincentelli.
\newblock \emph{Logic Minimization Algorithms for {VLSI} Synthesis}.
\newblock Kluwer Academic Publishers, Boston, MA, 1984.

\bibitem[Chen et~al.(2024)Chen, Jiang, Lin, Kwok, and Zhang]{chen2024routerdc}
Shuhao Chen, Weisen Jiang, Baijiong Lin, James~T.\ Kwok, and Yu~Zhang.
\newblock {RouterDC}: Query-based router by dual contrastive learning for assembling large language models.
\newblock \emph{arXiv preprint arXiv:2409.19886}, 2024.

\bibitem[Devlin et~al.(2019)Devlin, Chang, Lee, and Toutanova]{devlin2019bert}
Jacob Devlin, Ming-Wei Chang, Kenton Lee, and Kristina Toutanova.
\newblock {BERT}: Pre-training of deep bidirectional transformers for language understanding.
\newblock \emph{arXiv preprint arXiv:1810.04805}, 2019.

\bibitem[{Envoy Proxy Authors}(2024{\natexlab{a}})]{envoyextproc2024}
{Envoy Proxy Authors}.
\newblock Envoy external processing filter, 2024{\natexlab{a}}.
\newblock URL \url{https://www.envoyproxy.io/docs/envoy/latest/configuration/http/http_filters/ext_proc_filter}.

\bibitem[{Envoy Proxy Authors}(2024{\natexlab{b}})]{envoyproxy2024}
{Envoy Proxy Authors}.
\newblock Envoy proxy: An open source edge and service proxy, 2024{\natexlab{b}}.
\newblock URL \url{https://www.envoyproxy.io/}.

\bibitem[Fleisher \& Maissel(1975)Fleisher and Maissel]{fleisher1975introduction}
Harvey Fleisher and Lester~I. Maissel.
\newblock An introduction to array logic.
\newblock \emph{IBM Journal of Research and Development}, 19\penalty0 (2):\penalty0 98--109, 1975.

\bibitem[Hu et~al.(2022)Hu, Shen, Wallis, Allen-Zhu, Li, Wang, Wang, and Chen]{hu2022lora}
Edward~J.\ Hu, Yelong Shen, Phillip Wallis, Zeyuan Allen-Zhu, Yuanzhi Li, Shean Wang, Lu~Wang, and Weizhu Chen.
\newblock {LoRA}: Low-rank adaptation of large language models.
\newblock \emph{arXiv preprint arXiv:2106.09685}, 2022.

\bibitem[Hu et~al.(2024)Hu, Bieker, Li, Jiang, Keigwin, Ranganath, Keutzer, and Upadhyay]{hu2024routerbench}
Qitian~Jason Hu, Jacob Bieker, Xiuyu Li, Nan Jiang, Benjamin Keigwin, Gaurav Ranganath, Kurt Keutzer, and Shriyash~Kaustubh Upadhyay.
\newblock {RouterBench}: A benchmark for multi-{LLM} routing system.
\newblock \emph{arXiv preprint arXiv:2403.12031}, 2024.

\bibitem[{Hugging Face}(2024)]{candleml2024}
{Hugging Face}.
\newblock Candle: Minimalist {ML} framework for {Rust}, 2024.
\newblock URL \url{https://github.com/huggingface/candle}.

\bibitem[Huntington(1904)]{huntington1904sets}
Edward~V. Huntington.
\newblock Sets of independent postulates for the algebra of logic.
\newblock \emph{Transactions of the American Mathematical Society}, 5\penalty0 (3):\penalty0 288--309, 1904.

\bibitem[Inan et~al.(2023)Inan, Upasani, Chi, Rungta, Iyer, Mao, Tontchev, Hu, Fuller, Testuggine, and Khabsa]{inan2023llamaguard}
Hakan Inan, Kartikeya Upasani, Jianfeng Chi, Rashi Rungta, Krithika Iyer, Yuning Mao, Michael Tontchev, Qing Hu, Brian Fuller, Davide Testuggine, and Madian Khabsa.
\newblock {Llama Guard}: {LLM}-based input-output safeguard for human-{AI} conversations.
\newblock \emph{arXiv preprint arXiv:2312.06674}, 2023.

\bibitem[Jain et~al.(2023)Jain, Schwarzschild, Wen, Somepalli, Kirchenbauer, yeh Chiang, Goldblum, Saha, Geiping, and Goldstein]{jain2023baseline}
Neel Jain, Avi Schwarzschild, Yuxin Wen, Gowthami Somepalli, John Kirchenbauer, Ping yeh Chiang, Micah Goldblum, Aniruddha Saha, Jonas Geiping, and Tom Goldstein.
\newblock Baseline defenses for adversarial attacks against aligned language models.
\newblock \emph{arXiv preprint arXiv:2309.00614}, 2023.

\bibitem[Jiang et~al.(2024)Jiang, Sablayrolles, Roux, Mensch, Savary, Bamford, Chaplot, de~las Casas, Hanna, Bressand, et~al.]{jiang2024mixtral}
Albert~Q.\ Jiang, Alexandre Sablayrolles, Antoine Roux, Arthur Mensch, Blanche Savary, Chris Bamford, Devendra~Singh Chaplot, Diego de~las Casas, Emma~Bou Hanna, Florian Bressand, et~al.
\newblock Mixtral of experts.
\newblock \emph{arXiv preprint arXiv:2401.04088}, 2024.

\bibitem[Kusupati et~al.(2022)Kusupati, Bhatt, Rege, Wallingford, Sinha, Ramanujan, Howard-Snyder, Chen, Kakade, Jain, and Farhadi]{kusupati2022matryoshka}
Aditya Kusupati, Gantavya Bhatt, Aniket Rege, Matthew Wallingford, Aditya Sinha, Vivek Ramanujan, William Howard-Snyder, Kaifeng Chen, Sham Kakade, Prateek Jain, and Ali Farhadi.
\newblock Matryoshka representation learning.
\newblock \emph{arXiv preprint arXiv:2205.13147}, 2022.

\bibitem[Kwon et~al.(2023)Kwon, Li, Zuo, Sheng, Zheng, Yu, Gonzalez, Zhang, and Stoica]{kwon2023vllm}
Woosuk Kwon, Zhuohan Li, Sicheng Zuo, Ying Sheng, Lianmin Zheng, Cody~Hao Yu, Joseph~E.\ Gonzalez, Hao Zhang, and Ion Stoica.
\newblock Efficient memory management for large language model serving with {PagedAttention}.
\newblock In \emph{Proceedings of the 29th Symposium on Operating Systems Principles (SOSP)}, 2023.

\bibitem[Lewis et~al.(2020)Lewis, Perez, Piktus, Petroni, Karpukhin, Goyal, K{\"u}ttler, Lewis, tau Yih, Rockt{\"a}schel, Riedel, and Kiela]{lewis2020rag}
Patrick Lewis, Ethan Perez, Aleksandra Piktus, Fabio Petroni, Vladimir Karpukhin, Naman Goyal, Heinrich K{\"u}ttler, Mike Lewis, Wen tau Yih, Tim Rockt{\"a}schel, Sebastian Riedel, and Douwe Kiela.
\newblock Retrieval-augmented generation for knowledge-intensive {NLP} tasks.
\newblock \emph{arXiv preprint arXiv:2005.11401}, 2020.

\bibitem[Li et~al.(2010)Li, Chu, Langford, and Schapire]{li2010contextual}
Lihong Li, Wei Chu, John Langford, and Robert~E. Schapire.
\newblock A contextual-bandit approach to personalized news article recommendation.
\newblock In \emph{Proceedings of the 19th International Conference on World Wide Web}, pp.\  661--670. ACM, 2010.

\bibitem[{Linfa Contributors}(2024)]{linfa2024}
{Linfa Contributors}.
\newblock Linfa: A {Rust} machine learning framework, 2024.
\newblock URL \url{https://github.com/rust-ml/linfa}.

\bibitem[Lison et~al.(2021)Lison, Pil{\'a}n, S{\'a}nchez, Batet, and {\O}vrelid]{lison2021anonymisation}
Pierre Lison, Ildik{\'o} Pil{\'a}n, David S{\'a}nchez, Montserrat Batet, and Lilja {\O}vrelid.
\newblock Anonymisation models for text data: State of the art, challenges and future directions.
\newblock In \emph{Proceedings of the 59th Annual Meeting of the Association for Computational Linguistics (ACL)}, pp.\  4188--4203, 2021.

\bibitem[Manakul et~al.(2023)Manakul, Liusie, and Gales]{manakul2023selfcheckgpt}
Potsawee Manakul, Adian Liusie, and Mark J.\~F.\ Gales.
\newblock {SelfCheckGPT}: Zero-resource black-box hallucination detection for generative large language models.
\newblock \emph{arXiv preprint arXiv:2303.08896}, 2023.

\bibitem[Mangrulkar et~al.(2022)Mangrulkar, Gugger, Debut, Belkada, Paul, and Bossan]{mangrulkar2022peft}
Sourab Mangrulkar, Sylvain Gugger, Lysandre Debut, Younes Belkada, Sayak Paul, and Benjamin Bossan.
\newblock {PEFT}: State-of-the-art parameter-efficient fine-tuning methods.
\newblock In \emph{Hugging Face}, 2022.
\newblock URL \url{https://github.com/huggingface/peft}.

\bibitem[{Microsoft}(2024)]{onnxruntime2024}
{Microsoft}.
\newblock {ONNX Runtime}: Cross-platform, high performance {ML} inferencing and training accelerator, 2024.
\newblock URL \url{https://onnxruntime.ai/}.

\bibitem[Min et~al.(2023)Min, Krishna, Lyu, Lewis, tau Yih, Koh, Iyyer, Zettlemoyer, and Hajishirzi]{min2023factscore}
Sewon Min, Kalpesh Krishna, Xinxi Lyu, Mike Lewis, Wen tau Yih, Pang~Wei Koh, Mohit Iyyer, Luke Zettlemoyer, and Hannaneh Hajishirzi.
\newblock {FActScore}: Fine-grained atomic evaluation of factual precision in long form text generation.
\newblock \emph{arXiv preprint arXiv:2305.14251}, 2023.

\bibitem[Ong et~al.(2024)Ong, Almahairi, Wu, Chiang, Wu, Gonzalez, Kadous, and Stoica]{ong2024routellm}
Isaac Ong, Amjad Almahairi, Vincent Wu, Wei-Lin Chiang, Tianhao Wu, Joseph~E.\ Gonzalez, M~Waleed Kadous, and Ion Stoica.
\newblock {RouteLLM}: Learning to route {LLMs} with preference data.
\newblock \emph{arXiv preprint arXiv:2406.18665}, 2024.

\bibitem[{OpenTelemetry Authors}(2024)]{opentelemetry2024}
{OpenTelemetry Authors}.
\newblock {OpenTelemetry}: An observability framework for cloud-native software, 2024.
\newblock URL \url{https://opentelemetry.io/}.

\bibitem[Peng et~al.(2023)Peng, Alcaide, Anthony, Albalak, Arcadinho, Biderman, Cao, Cheng, Chung, Grella, et~al.]{peng2023yarn}
Bowen Peng, Eric Alcaide, Quentin Anthony, Amir Albalak, Samuel Arcadinho, Stella Biderman, Huanqi Cao, Xin Cheng, Matteo Chung, Matteo Grella, et~al.
\newblock {YaRN}: Efficient context window extension of large language models.
\newblock \emph{arXiv preprint arXiv:2309.00071}, 2023.

\bibitem[{Prometheus Authors}(2024)]{prometheus2024}
{Prometheus Authors}.
\newblock Prometheus: From metrics to insight, 2024.
\newblock URL \url{https://prometheus.io/}.

\bibitem[Rissanen(1978)]{rissanen1978modeling}
Jorma Rissanen.
\newblock Modeling by shortest data description.
\newblock \emph{Automatica}, 14\penalty0 (5):\penalty0 465--471, 1978.

\bibitem[Rivest(1987)]{rivest1987learning}
Ronald~L. Rivest.
\newblock Learning decision lists.
\newblock \emph{Machine Learning}, 2\penalty0 (3):\penalty0 229--246, 1987.

\bibitem[Shalev-Shwartz(2012)]{shalev2012online}
Shai Shalev-Shwartz.
\newblock Online learning and online convex optimization.
\newblock \emph{Foundations and Trends in Machine Learning}, 4\penalty0 (2):\penalty0 107--194, 2012.

\bibitem[Shannon(1938)]{shannon1938symbolic}
Claude~E. Shannon.
\newblock A symbolic analysis of relay and switching circuits.
\newblock \emph{Transactions of the American Institute of Electrical Engineers}, 57\penalty0 (12):\penalty0 713--723, 1938.

\bibitem[Shannon(1948)]{shannon1948mathematical}
Claude~E. Shannon.
\newblock A mathematical theory of communication.
\newblock \emph{The Bell System Technical Journal}, 27\penalty0 (3):\penalty0 379--423, 1948.

\bibitem[Shazeer et~al.(2017)Shazeer, Mirhoseini, Maziarz, Davis, Le, Hinton, and Dean]{shazeer2017moe}
Noam Shazeer, Azalia Mirhoseini, Krzysztof Maziarz, Andy Davis, Quoc Le, Geoffrey Hinton, and Jeff Dean.
\newblock Outrageously large neural networks: The sparsely-gated mixture-of-experts layer.
\newblock \emph{arXiv preprint arXiv:1701.06538}, 2017.

\bibitem[{The Kubernetes Authors}(2024)]{kubernetes2024}
{The Kubernetes Authors}.
\newblock {Kubernetes}: Production-grade container orchestration, 2024.
\newblock URL \url{https://kubernetes.io/}.

\bibitem[Thomas(2024)]{participle2024}
Alec Thomas.
\newblock Participle: A parser library for {Go}.
\newblock \url{https://github.com/alecthomas/participle}, 2024.
\newblock Accessed: 2025-01-15.

\bibitem[Thompson(1933)]{thompson1933likelihood}
William~R.\ Thompson.
\newblock On the likelihood that one unknown probability exceeds another in view of the evidence of two samples.
\newblock \emph{Biometrika}, 25\penalty0 (3/4):\penalty0 285--294, 1933.

\bibitem[Vaswani et~al.(2017)Vaswani, Shazeer, Parmar, Uszkoreit, Jones, Gomez, Kaiser, and Polosukhin]{vaswani2017attention}
Ashish Vaswani, Noam Shazeer, Niki Parmar, Jakob Uszkoreit, Llion Jones, Aidan~N.\ Gomez, {\L}ukasz Kaiser, and Illia Polosukhin.
\newblock Attention is all you need.
\newblock In \emph{Advances in Neural Information Processing Systems (NeurIPS)}, 2017.

\bibitem[Wang et~al.(2021)Wang, Yi, Guo, Jin, Xu, Li, Wang, Guo, Li, Xu, et~al.]{wang2021milvus}
Jianguo Wang, Xiaomeng Yi, Rentong Guo, Hai Jin, Peng Xu, Shengjun Li, Xiangyu Wang, Xiangzhou Guo, Chengming Li, Xiaohai Xu, et~al.
\newblock Milvus: A purpose-built vector data management system.
\newblock In \emph{Proceedings of the 2021 International Conference on Management of Data (SIGMOD)}, 2021.

\bibitem[Warner et~al.(2024)Warner, Chaffin, Clavie, Weller, Hallstrom, Taghadouini, Gallagher, Biswas, Ladhak, Aarsen, Cooper, Adams, Howard, and Poli]{warner2024modernbert}
Benjamin Warner, Antoine Chaffin, Benjamin Clavie, Orion Weller, Oskar Hallstrom, Said Taghadouini, Alexis Gallagher, Raja Biswas, Faisal Ladhak, Tom Aarsen, Nathan Cooper, Griffin Adams, Jeremy Howard, and Iacopo Poli.
\newblock Smarter, better, faster, longer: A modern bidirectional encoder for fast, memory efficient, and long context finetuning and inference.
\newblock \emph{arXiv preprint arXiv:2412.13663}, 2024.

\bibitem[Williams et~al.(2018)Williams, Nangia, and Bowman]{williams2018mnli}
Adina Williams, Nikita Nangia, and Samuel~R.\ Bowman.
\newblock A broad-coverage challenge corpus for sentence understanding through inference.
\newblock \emph{arXiv preprint arXiv:1704.05426}, 2018.

\bibitem[Xie et~al.(2025)Xie, Sun, Feng, and You]{xie2025gmtrouter}
Encheng Xie, Yihang Sun, Tao Feng, and Jiaxuan You.
\newblock {GMTRouter}: Personalized {LLM} router over multi-turn user interactions.
\newblock \emph{arXiv preprint arXiv:2511.08590}, 2025.

\bibitem[Zadeh(1965)]{zadeh1965fuzzy}
Lotfi~A. Zadeh.
\newblock Fuzzy sets.
\newblock \emph{Information and Control}, 8\penalty0 (3):\penalty0 338--353, 1965.

\bibitem[Zhang et~al.(2025{\natexlab{a}})Zhang, Feng, and You]{zhang2025routerr1}
Haozhen Zhang, Tao Feng, and Jiaxuan You.
\newblock {Router-R1}: Teaching {LLMs} multi-round routing and aggregation via reinforcement learning.
\newblock \emph{arXiv preprint arXiv:2506.09033}, 2025{\natexlab{a}}.

\bibitem[Zhang et~al.(2025{\natexlab{b}})Zhang, Ruan, and Hashimoto]{pacore2025}
Tianjun Zhang, Yangjun Ruan, and Tatsunori Hashimoto.
\newblock {PaCoRe}: Parallel cooperative reasoning with {LLMs}.
\newblock \emph{arXiv preprint arXiv:2601.05593}, 2025{\natexlab{b}}.

\end{thebibliography}

\end{document}